%% file: main.tex
\RequirePackage{ifthen}
\documentclass[a4paper,10pt]{article}

\newlength{\figurewidth}
\newcommand{\UGO}[1]{\textcolor{blue}{}}

\usepackage[cmex10]{amsmath}

\usepackage{array}
\usepackage{fixltx2e}
\usepackage{url}

\input{macros}
\input{macros_notations}

\begin{document}

\title{Metric Reasoning About $\lambda$-Terms:\\ The Affine Case\\ (Long Version)}

\author{Rapha\"elle Crubill\'e \and Ugo Dal Lago}
%{
%\author{
%\IEEEauthorblockN{Rapha\"elle Crubill\'e}
%\IEEEauthorblockA{ENS Lyon\\
%Email: \url{raphaelle.crubille@ens-lyon.fr}}
%\and
%\IEEEauthorblockN{Ugo Dal Lago}
%\IEEEauthorblockA{Universit\`a di Bologna \& INRIA\\
%Email: \url{ugo.dallago@unibo.it}}}}

\maketitle

\begin{abstract}
  Terms of Church's $\lambda$-calculus can be considered
  \emph{equivalent} along many different definitions, but context
  equivalence is certainly the most direct and universally accepted
  one. If the underlying calculus becomes probabilistic, however,
  equivalence is too discriminating: terms which have totally
  unrelated behaviours are treated the same as terms which behave very
  similarly. We study the problem of evaluating the \emph{distance}
  between affine $\lambda$-terms. The most natural definition for it,
  namely a natural generalisation of context equivalence, is shown to
  be characterised by a notion of \emph{trace} distance, and to be
  bounded from above by a coinductively defined distance based on the
  Kantorovich metric on distributions. A different, again
  fully-abstract, tuple-based notion of trace distance is shown to
  be able to handle nontrivial examples.
\end{abstract}

% no keywords

%%%%%%%%%%%%%%%%%%%%%%%%
\section{Introduction}
%%%%%%%%%%%%%%%%%%%%%%%%
Probabilistic models are formidable tools when abstracting the
behaviour of complicated, intractable systems by simpler ones, at the
price of introducing uncertainty. But there is more: randomness can be
seen as \emph{a way to compute}; in modern cryptography, as an
example, having access to a source of uniform randomness is essential
to achieve security in an asymmetric
setting~\cite{GoldwasserMicali}. Other domains where probabilistic
models play a key role include machine
learning~\cite{pearl1988probabilistic},
robotics~\cite{thrun2002robotic}, and
linguistics~\cite{manning1999foundations}. 

Probabilistic models of computation have been studied not only
directly, but also through concrete or abstract programming languages,
which are most often extensions of their deterministic siblings. Among
the many ways probabilistic choice can be captured in programming, the
simplest one consists in endowing the language of programs with an
operator modelling the flipping of a fair coin. This renders program
evaluation a probabilistic process, and under mild assumptions the
language becomes universal for probabilistic computation. Particularly
fruitful in this sense has been the line of work on the functional
paradigm, both at a theoretical~\cite{JonesPlotkin,Ramsey,Pfenning}
and at a more practical level~\cite{Church}. 

In presence of higher-order functions, program equivalence can be
captured by so-called \emph{context equivalence}: two programs
$\termone$ and $\termtwo$ are considered equivalent if they behave the
same no matter how the environment interacts with them: for every
context $\ctxone$, it holds that
$\obs{\fillc{\ctxone}{\termone}}=\obs{\fillc{\ctxone}{\termtwo}}$. However,
this definition has the drawback of being based on an universal
quantification over \emph{all} contexts: showing that two programs are
equivalent, requires considering their interaction with every possible
context. The problem of giving handier characterisations of context
equivalence can be approached in many different ways. As an example,
coinductive methodologies for program equivalence have been studied
thoroughly in deterministic~\cite{Abramsky1990,Pitts} and
non-deterministic~\cite{Lassen} computation, with new and exciting
results appearing recently also for probabilistic languages:
applicative bisimilarity, a coinductively defined notion of
equivalence for functional programs, has been shown to be sound,
and sometime even fully abstract, for probabilistic
$\lambda$-calculi~\cite{DalLagoSangiorgiAlberti2014POPL,CrubilleDalLago2014ESOP}.

In a probabilistic setting, however, equivalences are too strong if
defined as above. Indeed, two programs are equivalent if their
probabilistic behaviour is \emph{exactly} the same (in every
context). The actual value of probabilities in a probabilistic model
often comes from statistical measurements, and should be considered
more as an approximation to the actual probability law. Consequently,
we would like to compare programs by appropriately reflecting small
variations in them. Another scenario in which a richer, more
informative way of comparing programs is needed is cryptography, where
a central notion of equivalence, called \emph{computational
  indistinguishability}~\cite{Goldreich} is indeed based on
statistical distance rather than equality: the adversary \emph{can} win
the game, but with a \emph{small} probability. Summing up,
equivalences should be refined into metrics, and this is the path we
will follow in this paper.

In probabilistic $\lambda$-calculi, the notion of observation
$\obs{\cdot}$ is quantitative: it is either the \emph{probability} of
convergence to a certain observable base value (e.g. the empty
string), or the probability of convergence \textsl{tout court}. One can
then easily define a notion of \emph{context distance} as the maximal
distance contexts can achieve when separating two terms:
$$
\metrctx(\termone,\termtwo)=\sup_{\ctxone}\;\;\abs{\obs{\fillc{\ctxone}{\termone}}-\obs{\fillc{\ctxone}{\termtwo}}}.
$$
This looks very close to computational indistinguishability,
except for the absence of a security parameter: a scheme is
secure if the advantage of any adversary in a given game (e.g., 
consisting in distinguishing between the case where the scheme is used,
and the case where it is replaced by a truly random process) is
``small'' (e.g., negligible). Again, however, we find ourselves
in front of a definition which risks to be useless in proofs,
given that all contexts must be taken into account.
But how difficult is evaluating the distance between concrete
higher-order terms? Are there ways to alleviate the burden
of dealing with all contexts, like for equivalences? These
are the questions we address in this paper, and which have to
the authors' knowledge not been investigated before.

As we will discuss in Section~\ref{sect:anatomy} below, finding
handier characterisations of the context distance poses challenges
which are simply different (and often harder) than the ones
encountered in context equivalence. In particular, the context
distance tends to \emph{trivialise} and, perhaps worse, naively
applying the natural generalisation of techniques known for
equivalence is bound to lead to unsound methodologies. Indeed, one
immediately realises that the number of times contexts access
their argument is a crucial parameter, which must necessarily be
dealt with. This is the reason why we work with an affine
$\lambda$-calculus in this paper: this is a necessary first step, but
also points to the right way to tame the general, non-linear case.

An extended version of this paper with more details is available~\cite{EV}.

%%%%%%%%%%%%%%%%%%%%%%%%%%%%
\subsection*{Contributions}
%%%%%%%%%%%%%%%%%%%%%%%%%%%%
We introduce in this paper \emph{three} distinct notions of distance
for terms in an untyped, probabilistic, and affine $\lambda$-calculus.
The first one is a notion of \emph{trace distance}, in which terms are
faced with \emph{linear} tests, i.e. sequences of arguments. The
distance between two terms is then defined as the greatest separation
any linear test achieves. The first results of this paper are the
non-expansiveness of the trace distance, which implies (given that any
linear test can easily be implemented by an affine context) that the
trace and context distances coincide. This is the topic of
Section~\ref{sect:tracedistance} below. 

Section~\ref{sect:bisimulationdistance}, instead, focuses on another
notion of distance, which is coinductively defined following the
well-known Kantorovich metric~\cite{kantorovich1942transfer} for
distributions of states in any labelled Markov chain (LMC in the
following), and that we dub the \emph{bisimulation distance}. This
second notion of a distance is not only smaller than the trace
distance, which is well expected, but non-expansive itself. This is
proved by a variation on the Howe's method~\cite{Howe}, a well-known
technique for proving that bisimilarity is a congruence in an
higher-order setting, and which has never been used for metrics
before. On the other hand, the bisimulation distance \emph{does not}
coincide with the context distance, a fact that we do not only prove
 by giving a counterexample, but that we justify by relying on a
test-based characterisation of the bisimulation distance known from the
literature.

For the sake of simplicity, the trace and bisimulation distances are
analysed on a purely applicative $\lambda$-calculus, keeping in mind
that pairs could be very easily handled, and can even be encoded in
the applicative fragment, as discussed in
Section~\ref{sect:pairs}. The presence of pairs, however, allows us to
form very interesting examples of distance problems, one of which will
drive us throughout the paper but unfortunately turns out hard to
handle neither by the trace distance nor by the bisimulation
distance. This is the starting point for the third notion of distance
introduced in this paper, which is the subject of
Section~\ref{sect:tupledistance}, and which we call the
\emph{tuple distance}. Our third notion of distance can be proved
to coincide with the trace distance, and thus with the context
distance. But this is not the end of the story: in the tuple
distance, not a single but \emph{many} terms are compared, and this
makes the distance between concrete terms much easier to evaluate:
interaction is somehow internalised. In particular, our running
example can be handled quite easily. The way the tuple
distance is defined makes it adaptable to non-affine calculi, a topic
which is outside the scope of this paper, but which we briefly discuss
in Section~\ref{sect:exponentials}.

%%%%%%%%%%%%%%%%%%%%%%%%%%
\subsection*{Related Work}\label{sect:related}
%%%%%%%%%%%%%%%%%%%%%%%%%%
This is definitely not the first work on metrics for probabilistic
systems: notions of coinductively defined metrics for LMCs, as an
example, have been extensively studied (e.g.~\cite{DesharnaisLICS02,
  DesharnaisCONCUR99,Worell}). There has been, to our knowledge, not
so many investigations on the meaning of metrics for concrete
programming languages~\cite{GeblerTini}, and almost nothing on metric
for higher-order languages.

If the key property notions of \emph{equivalences} are required to
satisfy consists in being \emph{congruences}, the corresponding
property for metrics has traditionally been taken as
\emph{non-expansiveness}.  Indeed, many results from the literature
(e.g.~\cite{DesharnaisLICS02,Mitchell}) have precisely the form of
non-expansiveness results for metrics defined in various forms.
The underlying language, however, invariably take the form of a process algebra
without any higher-order feature.  The work of Gebler, Tini, and
co-authors shows that one could go beyond non-expansiveness and
towards uniform continuity~\cite{GeblerTini} but, again, higher-order
functions remain out of scope.

Notions of \emph{equivalence} for various forms of probabilistic
$\lambda$-calculi have also been extensively studied, starting from
the pioneering work by Plotkin and Jones~\cite{JonesPlotkin}, down to recent
results on probabilistic applicative
bisimulation~\cite{DalLagoSangiorgiAlberti2014POPL,CrubilleDalLago2014ESOP}, logical
relations~\cite{BizjakBirkedal}, and probabilistic coherent
spaces~\cite{DanosEhrhard,EhrhardTassonPagani2014POPL}.  None of the works above,
however, go beyond equivalences and deals with notions of distances
between terms.

%%%%%%%%%%%%%%%%%%%%%%%%%%%%%%%%%%%
\section{The Anatomy of a Distance}\label{sect:anatomy}
%%%%%%%%%%%%%%%%%%%%%%%%%%%%%%%%%%%
In this section, we describe the difficulties one encounters when
trying to characterise the context distance with either bisimulation
or trace metrics.

Suppose we have two terms $\termone$ and $\termtwo$ of boolean type
written in a probabilistic $\lambda$-calculus. As such, $\termone$ and
$\termtwo$ do not evaluate to \emph{a} value in the domain of booleans
but to a \emph{distribution} over the same domain. $\termone$
evaluates to the distribution assigning $\ttrue$ probability $1$,
while $\termtwo$ evaluates to the uniform distribution over booleans,
(i.e. the distribution which attributes probability $\frac{1}{2}$ to
both $\ttrue$ and $\ffalse$). Figure~\ref{fig:exa} depicts the
relevant fragment of a LMC, whose induced notion of probabilistic
bisimilarity has been proved to be sound for context
equivalence~\cite{CrubilleDalLago2014ESOP}.
\begin{figure}
\begin{center}
\fbox{
\begin{minipage}{\figurewidth}
\begin{center}
\begin{tikzpicture}[auto]
\node [draw, rectangle, rounded corners] (M) at (0,0) {\scriptsize$\termone =\ttrue $};
\node[draw, rectangle, rounded corners] (N) at (4,0) {\scriptsize$\termtwo = \psum \ttrue \ffalse$};
\node[draw, rectangle, rounded corners] (T) at (0,-2){\scriptsize$\dval \ttrue$};
\node[draw, rectangle, rounded corners] (F) at (4,-2){\scriptsize $\dval \ffalse$};
\node [draw, circle](A) at (0,-1){};
\node [draw, circle](B) at (4,-1){};
\draw[-](M) to node {eval} (A);
\draw[-](N) to node {eval}(B);
\draw[->](A) to node {$1 $}  (T);
\draw[->](B) to node {$\frac 1 2 $}  (T);
\draw[->](B) to node {$\frac 1 2 $}  (F);
\draw[->](T)[loop left] to node {$\text{true} $}  (T);
\draw[->](F) [loop right]to node {$\text{false} $}  (F);
\end{tikzpicture}
\end{center}
\end{minipage}}
\end{center}
\caption{$\termone$ and $\termtwo$ as states of a LMC}\label{fig:exa}
\end{figure}
$\termone$ and $\termtwo$ are not bisimilar. Indeed, $\dval{\ttrue}$
and $\dval{\ffalse}$ are trivially not bisimilar, while $\termone$ and
$\termtwo$ go to equivalent states with different probabilities.  The
two terms are non-equivalent also contextually. But what \emph{should
  be} the distance between $\termone$ and $\termtwo$?

For the moment, let us forget about the context distance, and
concentrate on the notions of distance for LMCs we mentioned in
Section \ref{sect:related}. In all cases we are aware of, we obtain
that $\termone$ and $\termtwo$ are at distance $\frac{1}{2}$. As an
example, if we consider a trace metric, we have to compare the
success probability of linear tests, starting from $\termone$ and
$\termtwo$. More precisely, the tests of interest with respect to
these two terms are: 
$$
\testone \defi \eval; \qquad
\testtwo \defi\concat \eval {\text{true}}; \qquad 
\testthree \defi\concat \eval{\text{false}}.
$$
Since neither $\termone$ nor $\termtwo$ has a non-zero divergence
probability, they both pass the test $\testone$ with probability
$1$. The success probability of the test $\testtwo$ corresponds to the
probability of evaluating to $\ttrue$: it is $1$ for $\termone$ and
$\frac 1 2$ for $\termtwo$. Similarly, the success probability of
$\testthree$ corresponds to the probability to obtain $\ffalse$ after
evaluation: it is $0$ for $\termone$ and $\frac 1 2$ for
$\termtwo$. So we can see that the maximal separation linear tests can
obtain is $\frac 1 2$. The situation is quite similar
for bisimulation metrics~\cite{DesharnaisLICS02}, which attribute
distance $\frac{1}{2}$ to $\termone$ and $\termtwo$.

It is easy, however, to find a family of contexts
$\{\ctxone_n\}_{n\in\NN}$ such that $\fillc{\ctxone_n}{\termone}$
evaluates to $\ttrue$ with probability $1$, and
$\fillc{\ctxone_n}{\termtwo}$ evaluates to $\ffalse$ with probability
$1-\frac{1}{2^n}$: define $\ctxone_n$ as a context that copies its
argument $n$ times, returning $\ffalse$ if \emph{at least one} of the
$n$ copies evaluates to $\ffalse$, and otherwise returns $\ttrue$.  As
a consequence, the context distance between $\termone$ and $\termtwo$
is $1$. In fact, this reasoning can be extended to any pair of
programs which are not equivalent but whose probability of convergence
is $1$: out of a context which separates them of $\epsone>0$, with
$\epsone$ very small, we can construct a context that separates them
of $1$ performing some statistical reasoning. The situation is more
complicated if we take the probability of convergence as an
observable: we cannot \emph{always} construct contexts that
discriminate terms based on their probability of convergence, although
something can be done if the terms' probabilities of convergence are
different but close to $1$. The context metric, in other words, risks
to be not continuous and close to trivial if contexts are too
powerful. What the example above shows, however, is something even
worse: if contexts are allowed to copy their arguments, then any
metric defined upon the usual presentation of probabilistic
$\lambda$-calculus as an LMC (a fragment of which is depicted in
Figure~\ref{fig:exa}) is bound to be \emph{unsound} w.r.t. the context
metric.

Whether bisimulation metrics are sound, how close they are to the
context distance, and whether they are useful in relieving the burden
of evaluating it, are however open and interesting questions even in
absence of the copying capability, i.e., when the underlying language
is \emph{affine}. This is the main reason why we focus in this work on
such a $\lambda$-calculus, whose expressive power is limited (although
definitely non-trivial~\cite{MairsonJFP}) but which
is anyway higher-order. We discover this way an elegant and deep
theory in which trace and bisimulation metrics are indeed sound, 
At the end of this paper, some hints will be given about how the case
of the untyped $\lambda$-calculus can be handled, a problem which we
leave for future work.

Evaluating the context distance between affine terms is already an
interesting and nontrivial problem. Consider, as an example, a
sequence of terms $\{\termone_n\}_{n \in \NN}$ defined inductively as follows (where $\diver$ stands for a term with zero probability of converging):
$$
\termone_0 = \pair {\abstr\varone \diver}{\abstr \varone \diver};\qquad
\termone_{n+1} = \pair {\abstr\varone{\termone_n}}{\abstr\varone\diver}.
$$ 
$ \termone_0$ is the pair whose components are both equal to $\abstr
\varone \diver$, and $\termone_{n+1}$ is defined as a pair whose first
component is the function which returns $\termone_n$ whatever its
argument is, and the second component is again $\abstr \varone
\diver$. We are now going to define another sequence of terms
$\{\termtwo_n\}_{n\in \NN}$, which can be seen as a noisy variation on
$\{\termone_n\}_{n\in \NN}$. More precisely, $\termtwo_0$ is the same
as $\termone_0$, and for each $n\in \NN$, $\termtwo_{n+1}$ is
constructed similarly to $\termone_{n+1}$, but adding some negligible
noise in \emph{both} components:
\begin{align*}
\termtwo_0&=\pair{\abstr \varone \diver}{\abstr \varone \diver};\\
\termtwo_{n+1}&=\pair{\abstr\varone {\psumindex
    {\termtwo_n}{\frac 1 {2^{n+1}}}{\diver}}}{\abstr \varone
  {\psumindex{\diver}{\frac 1 {2^{n+1}}}{\identity}}}.
\end{align*}
($\identity$ stands for the identity: $\abstr \varone \varone$, while the term $\psumindex \termthree p \termfour$ has the same behaviour
as $\termthree$ with probability $(1-p)$, and the same behaviour as
$\termfour$ with probability $p$.)  We would like to study how the
distance between $\termone_n$ and $\termtwo_{n}$ evolves when $n$
tends to infinity: do the little differences we apply at each step $n$
accumulate, and how can we express this accumulation quantitatively?

Intuitively, it is easy for the environment to separate $\termone_n$
and $\termtwo_n$ of ${\frac 1 {2^{n}}}$: it is enough to consider a
context $\ctxone$ which simply takes the second component of the pair,
passes any argument to it, and evaluates it: the convergence
probability of $\fillc \ctxone {\termone_n}$ is $0$, while the
convergence probability of $\fillc \ctxone {\termtwo_{n}}$ is $ \frac
1 {2^{n}}$. But the environment can also decide to take the \emph{first}
component of the pair, in order to use the fact that $\termone_{n-1}$
and $\termtwo_{n-1}$ can be distinguished: more precisely, let us
suppose that we have a context $\ctxone$ which separates
$\termone_{n-1}$ and $\termtwo_{n-1}$. Then we can construct a context
$\ctxtwo$ which takes the first element of the pair, passes any argument
to it, tries to evaluate it, and if it succeeds, gives the
result as an argument to $\ctxone$. We would like to express the
supremum of the separation that such a context can obtain as a
function of the distance between $\termone_{n-1}$ and
$\termtwo_{n-1}$. Unfortunately, this is not so simple: if $\ctxone$
is such that the convergence probability of $\fillc \ctxone
{\termone_{n-1}}$ is $\epsone$ and the convergence probability of
$\fillc \ctxone {\termtwo_{n-1}}$ is $\epstwo$, we can see that the
convergence probability of $\fillc \ctxtwo {\termone_n}$ is $\epsone$,
whereas the convergence probability of $\fillc \ctxtwo {\termtwo_n}$ is
$(\epstwo \cdot (1-\frac 1 {2^{n}}))$. But it is not possible to
express $\abs{\epsone - \epstwo \cdot (1- \frac 1 {2^{n}}) }$ as a
function of $\abs{\epsone - \epstwo}$ and $n$: intuitively, the
separation that the context $\ctxtwo$ can achieve depends not only on
the separation that the context $\ctxone$ can achieve, but also on
\emph{how} $\ctxone$ achieves it.  And moreover, the environment may
of course decide to use the two components of the pair, and to make
them interact in an arbitrary way. Summing up, although the
mechanism of construction of these terms seems to be locally easy to
measure, it is complicated to have any idea about how the distance
between them evolves when $n$ tends to infinity.
%%%%%%%%%%%%%%%%%%%%%%%%%
\section{Preliminaries}
%%%%%%%%%%%%%%%%%%%%%%%%%
In this section, an affine and probabilistic $\lambda$-calculus, which is
the object of study of this paper, will be introduced formally,
together with a notion of context distance for it.
%%%%%%%%%%%%%%%%%%%%%%%%%%%%%%%%%%%%%%%%%%%%%%%%%%%%%%%%%%%%%%%%%%
\subsection{An Affine, Untyped, Probabilistic $\lambda$-Calculus}
%%%%%%%%%%%%%%%%%%%%%%%%%%%%%%%%%%%%%%%%%%%%%%%%%%%%%%%%%%%%%%%%%%
We endow the $\lambda$-calculus with a probabilistic operator $\psum{}{}$,
which corresponds to the possibility for the program to choose one 
between two arguments, each with the same probability. \emph{Terms} 
are expressions generated by the following grammar:
$$
\termone\bnf\varone\midd\abstr{\varone}{\termone}\midd\termone\termone\midd\psum{\termone}{\termone}\midd\diver,
$$
where $\diver$ models divergence\footnote{since we only consider
  affine terms, we cannot encode divergence by the usual constructions
  of $\lambda$-calculus}, and $\varone$ ranges over a countable
set $\variables$ of variables.

The class of affine terms, which model functions using their
arguments at most once, can be isolated by way of a formal system,
whose judgements are in the form $\wfj{\contone}{\termone}$ (where
$\contone$ is any finite set of variables) and whose rules are the
following (where $\contone, \conttwo$ stands for the union of two disjoints contexts):
$$
\AxiomC{}
\UnaryInfC{$\wfj{\contone,\varone}{\varone}$}
\DisplayProof
\qquad
\AxiomC{$\wfj{\contone,\varone}{\termone}$}
\UnaryInfC{$\wfj{\contone}{\abstr{\varone}{\termone}}$}
\DisplayProof
\qquad
\AxiomC{$\wfj{\contone}{\termone}$}
\AxiomC{$\wfj{\conttwo}{\termtwo}$}
\BinaryInfC{$\wfj{\contone,\conttwo}{\termone\termtwo}$}
\DisplayProof
$$
$$
\AxiomC{$\wfj{\contone}{\termone}$}
\AxiomC{$\wfj{\contone}{\termtwo}$}
\BinaryInfC{$\wfj{\contone}{\psum{\termone}{\termtwo}}$}
\DisplayProof
\qquad
\AxiomC{}
\UnaryInfC{$\wfj{\contone}{\diver}$}
\DisplayProof  
$$

A \emph{program} is a term such that $\wfj{\emcon}{\termone}$, and
$\programs$ is the set of all such terms. We will call them 
\emph{closed terms}. We say that a program is a \emph{value} if it is of
the form $\abstr \varone \termone$, and $\valset$ is the set of such programs. The semantics of the just defined
calculus is expressed as a binary relation $\bss$ between programs and
\emph{value subdistributions} (or simply \emph{value distributions}),
i.e.  functions from values to real numbers whose sum is
\emph{smaller} or equal to $1$. The relation $\bss$ is inductively
defined by the following rules:
$$
\AxiomC{}
\UnaryInfC{$\bssp{\diver}{\emdist}$}
\DisplayProof
\qquad
\AxiomC{$\valone \text{ a value }$}
\UnaryInfC{$\bssp{\valone}{\{\valone^1\}}$}
\DisplayProof
\qquad
\AxiomC{$\bssp{\termone}{\distrone}$}
\AxiomC{$\bssp{\termtwo}{\distrtwo}$}
\BinaryInfC{$\bssp{\psum{\termone}{\termtwo}}{\frac{1}{2}\distrone+\frac{1}{2}\distrtwo}$}
\DisplayProof$$ $$
\AxiomC{$
  \begin{array}{c}
    \bssp{\termone}{\distrone}\qquad\bssp{\termtwo}{\distrtwo}\\
    \{\bssp{\subst{\termthree}{\varone}{\valone}}{\distrthree_{\termthree,\valone}}\}_{\abstr{\varone}{\termthree}\in\supp{\distrone},\valone\in\supp{\distrtwo}}
  \end{array}$}
\UnaryInfC{$\bssp{\termone\termtwo}{\sum\distrone(\abstr{\varone}{\termthree})\cdot\distrtwo(\valone)\cdot\distrthree_{\termthree,\valone}}$}
\DisplayProof
$$
where $\supp{\distrone}$ stands for the support of the distribution
$\distrone$.  The divergent program $\diver$, as expected, evaluates
to the \emph{empty} value distribution $\emdist$ which assigns $0$ to
any value. The expression $\{\valone^1\}$ stands for the Dirac's value
distribution on $\valone$; more generally the expression
$\{\valone_1^{\probone_1},\ldots,\valone_n^{\probone_n}\}$ indicates
the value distribution assigning probability $\probone_i$ to each
$\valone_i$ (and $0$ to any other value).

For every program $\termone$, there exists precisely \emph{one} value
distribution $\distrone$ such that $\bssp \termone \distrone$, that we
note $\sem \termone$. This holds only because we restrict ourselves to
affine terms. Moreover, $\sem \termone$ is always a finite
distribution. The rule for application expresses the fact that the
semantics is call-by-value: the argument is evaluated before being
passed to the function. There is no special reason why we adopt
call-by-value here, and all we are going to say also holds for (weak)
call-by-name evaluation.

In some circumstances, we would need to have a more local view of how
the programs behave. For these reason, we define an equivalent notion
of small-steps semantics, which allows us to reason about every small
execution step. We define first a one-step semantics $\ssp {} {}$
between programs and distribution over programs: {
$$ 
\AxiomC{}
  \UnaryInfC{$\ssp \diver \emdist$}
  \DisplayProof \qquad
  \AxiomC{}
  \UnaryInfC{$\ssp{\psum{\termone}{\termtwo}}{\frac 1 2 \cdot \dirac\termone + \frac 1 2 \cdot\dirac{\termtwo}}$}
  \DisplayProof
$$ 
$$
  \AxiomC{}
  \UnaryInfC{$\ssp{\left(\abstr \varone \termone\right) \valone}{\dirac {\subst \termone \varone \valone}}$}
  \DisplayProof$$
  $$ \AxiomC{$\ssp \termone \distrone$}
  \UnaryInfC{$\ssp{\termone\termtwo}{\sum\distrone(\termthree)\cdot
      \dirac {\termthree \termtwo}}$} \DisplayProof \qquad
  \AxiomC{$\ssp \termtwo \distrone$}
  \UnaryInfC{$\ssp{\valone\termtwo}{\sum\distrone(\termthree)\cdot
      \dirac {\valone\termthree}}$} \DisplayProof$$ } 
Then we use it to define a small step semantics $\sssp {} {}$, which
is a relation between programs and value distributions, and
corresponds to do as much as possible steps of $\ssp {} {}$. The rules
are the following:
\begin{center}{\small
\AxiomC{}
\UnaryInfC{$\sssp \valone \dirac \valone $}
\DisplayProof
\quad
\AxiomC{$\ssp \termone \distrone $}
\AxiomC{$(\sssp \termtwo \distrtwo_\termtwo)_{\termtwo \in \supp \distrone} $}
\BinaryInfC{$\sssp \termone {\sum \distrone(\termtwo) \cdot \distrtwo_\termtwo}$}
\DisplayProof
}\end{center}
 
Big-step and small-steps semantics are equivalent: for
every program $\termone$, there exists a unique distribution
$\distrone$ such that $\sssp \termone \distrone$, and moreover
$\distrone = \sem \termone$.
%%%%%%%%%%%%%%%%%%%%%%%%%%%%%%%
\subsection{Context Distance}
%%%%%%%%%%%%%%%%%%%%%%%%%%%%%%%
We now want to define a notion of observation for programs which
somehow measures the convergence probability of a program. We will do
that following the previous literature on this subject. For any 
distribution $\distrone$ over a set $A$, its sum $\sum_{a \in A} \distrone(a)$
is indicated as $\sumdistr{\distrone}$ and is said to be the
\emph{weight} of $\distrone$.  The \emph{convergence probability} of a
term $\termone$, that we note $\sumsem{\termone}$, is simply
$\sumdistr{\sem \termone}$, i.e., the weight of its semantics. For
instance, the convergence probability of $\diver$ is zero.

The environment, as usual, is modelled by the notion of a
\emph{context}, which is nothing more than a term with a
single occurrence of the \emph{hole} $[\cdot]$.
 They are generated by the following grammar:
$$
\ctxone\bnf\hole \midd\termone \midd\abstr{\varone}{\ctxone}\midd\ctxone\termone\midd \termone \ctxone\midd\psum{\ctxone}{\ctxone}.
$$ Affine
contexts can be identified by a formal system akin to the
one for terms. We note as $\fillc \ctxone \termone$ the program obtained by replacing
$[\cdot]$ by the closed term $\termone$ in $\ctxone$. The interaction of a program $\termone$ with
a context $\ctxone$ is the execution of the program $\fillc \ctxone
\termone$. 

We now consider three different ways of comparing programs, based on
their behaviour when interacting with the environment: a preorder
$\leqctx {}{}$, an equivalence relation $\equivctx{}{}$, and a map $\metrctx$:
\condskip
\begin{definition}[Context Equivalence, Context Distance]
Let $\termone$ and $\termtwo$ be two programs. Then we write that
$\leqctx\termone\termtwo$ if and only if for every context $\ctxone$, it holds that
$\sumsem{\fillc\ctxone\termone}\leq\sumsem{\fillc\ctxone\termtwo}$. 
If $\leqctx{\termone}{\termtwo}$ and $\leqctx{\termtwo}{\termone}$,
then we say that the two terms are \emph{context equivalent}, and
we write $\equivctx \termone \termtwo$. With the same hypotheses,
we say the \emph{context distance} between $\termone$ and $\termtwo$
is the real number $\appl \metrctx \termone \termtwo$
defined as $\sup_{\ctxone} \abs{\sumsem {\fillc \ctxone \termone} - {\sumsem{\fillc \ctxone \termtwo}}}$.
\end{definition}
\condskip
Please observe that, following~\cite{DengZhang}, we only compare
programs and not arbitrary terms. This is anyway harmless in an affine
setting.
\begin{example}\label{exampleun} 
  Let $\identity$ be the identity $\abstr \varone
  \varone$. $\identity$ and $\diver$ are as far as two programs can
  be: $\appl \metrctx {\identity}{\diver}=1$. To prove that, finding a
  context which always converges for one of the terms, and always
  diverges for the other one, suffices. We can take $\ctxone=\hole$,
  and we have that $\sumsem {\fillc \ctxone \identity} = 1$ and
  $\sumsem {\fillc \ctxone \diver} = 0$. Of course, $\identity$ and
  $\diver$ are not context equivalent.
  Throwing in probabilistic choice can complicate matters a bit.
  Consider the two terms $\psum{\identity}{\diver}$ and $\identity$.
  One can easily prove that $\appl \metrctx {\psum \identity \diver}
  {\identity}\geq\frac{1}{2}$: just consider $\ctxone = \hole$.
  However, showing that the above inequality is in fact an equality,
  requires showing that there cannot exist any context that separates
  more, which is possible, but definitely harder. This will be shown
  in the next section, using a trace-based characterisation of context
  distance.
\end{example}
%%%%%%%%%%%%%%%%%%%%%%%%%%%%%%%%
\subsection{On Pseudometrics}\label{metricproperties}
%%%%%%%%%%%%%%%%%%%%%%%%%%%%%%%%
Which properties does the context distance satisfy, and which structure
it then gives to the set of programs? This section answers these
questions, and prepares the ground for the sequel by
fixing some terminology.
\condskip
\begin{definition}[Pseudometrics]\label{defmetric}
  Let $\setone$ be a set. A \emph{premetric} on $\setone$ is any
  function $\metrone:\setone \rightarrow \setone$ such that $0 \leq
  \appl\metrone\stateone\statetwo \leq 1$ and
  $\appl\metrone\stateone\stateone=0$. A \emph{pseudometric} on
  $\setone$ is any premetric such that for every $\stateone,\statetwo,\statethree\in\setone$, 
  it holds that
    $\appl\metrone\stateone\statetwo = \appl\metrone\statetwo\stateone$
     and $\appl\metrone\stateone\statetwo \leq \appl\metrone\stateone\statethree+\appl\metrone\statethree\statetwo$.
   The set of all pseudometrics on $\setone$ is indicated with $\metrs{\setone}$.
\end{definition}
\condskip 
Please observe that pseudometrics are not metrics in the usual sense,
since $\appl \metrone \stateone \statetwo = 0$ does not necessarily
imply that $\stateone=\statetwo$. If we have a pseudometric
$\metrone$, we can construct an equivalence relation by considering
the \emph{kernel} of $\metrone$, that is the set of those pairs
$(\stateone, \statetwo)$ such that $\appl \metrone \stateone \statetwo
= 0$. It is easy to prove that the context distance is indeed a
pseudometric, and that its kernel is context equivalence. We would
now want to define a preorder $\leqmetr{}{}$ on pseudometrics in such a way that if
$\leqmetr \metrone \metrtwo$, then the kernel of $\metrone$ is
included in the kernel of $\metrtwo$. The natural choice, then, is to
take the following definition, which is the reverse of the
pointwise order on $[0,1]$:
\condskip
\begin{definition}[Pseudometric Ordering]
  Let $\setone$ be any set, and let $\metrone$ and $\metrtwo$ be two metrics
  in $\metrs \setone$. Then we stipulate that $\leqmetr \metrone \metrtwo$ if and only if, for
  every $\stateone, \statetwo \in \setone$ we have that $
  \appl\metrtwo\stateone\statetwo \leq \appl \metrone
  \stateone\statetwo$.
\end{definition}
\condskip

\begin{lemma}\label{lattice}
  For any set $\setone$, $(\metrs \setone,\leqmetr\cdot\cdot)$ is a
  complete lattice.
\end{lemma}
\condskip
But when, precisely, can a pseudometric on programs be considered as a
sound notion of distance? First of all, we would like it to put two
programs at least as far as the difference between their convergence
probabilities, since this is precisely our notion of observation:
\condskip
\begin{definition}[Adequacy]
  Let $\metrone$ be a pseudometric on the set of programs. Then $\metrone$
  is an \emph{adequate pseudometric} if for any programs $\termone$ and
  $\termtwo$, we have that $\abs{\sumdistr {\sem \termone} -
    \sumdistr{\sem \termtwo}} \leq \appl \metrone \termone \termtwo$.
\end{definition}
\condskip
Secondly, we are interested in how programs behave when interacting
with the environment. Especially, if we have two terms $\termone$ and
$\termtwo$ at a given distance $\epsone$, and we put them in an
environment $\ctxone$, we would like a pseudometric $\metrone$ to give
us some information about the distance between $\fillc \ctxone
\termone$ and $\fillc \ctxone \termtwo$. This is the idea behind the
following, standard, definition:
\condskip
\begin{definition}[Non-Expansiveness]
  Let $\metrone$ be a pseudometric on programs. We say $\metrone$
  is \emph{non-expansive} if for every pair of programs $\termone$ and
  $\termtwo$ and for every context $\ctxone$, we have that $\appl\metrone{\fillc \ctxone\termone}{\fillc
    \ctxone \termtwo} \leq \appl\metrone \termone \termtwo $.
\end{definition}
\condskip 
Non-expansiveness is the natural generalisation of the usual notion of
\emph{congruence}: if $\relone$ is an equivalence
relation on program, it is congruent if for every context $\ctxone$,
if $\relate \relone \termone \termtwo$, then $\relate \relone {\fillc
  \ctxone \termone} { \fillc \ctxone \termtwo}$.
By construction,
$\metrctx$ is a non-expansive pseudometric. We can also adapt the notion of soundness 
to pseudometric; $\metrone$ is said to be a \emph{sound} pseudometric on programs 
if $\leqmetr \metrone \metrctx$. Clearly, any adequate and non-expansive pseudometric is sound.
In the rest of this paper, we will only deal with pseudometrics, but
for the sake of simplicity we will refer to them simply as metrics.
%%%%%%%%%%%%%%%%%%%%%%%%%%%
\section{The Trace Distance}\label{sect:tracedistance}
%%%%%%%%%%%%%%%%%%%%%%%%%%%
The first notion of metric we study is based on traces, i.e., linear
tests. This is handier than the context distance, since contexts are
replaced by objects with a simpler structure.
%%%%%%%%%%%%%%%%%%%%%
\subsection{Definition}
%%%%%%%%%%%%%%%%%%%%%
A \emph{trace} $\traceone$ is a sequence in the form $\concat{\app\valone_1}
\cdots {\app \valone_n}$, where $\valone_1, \cdots \valone_n$ are
values, and we note $\words$ the set of traces. In other words,
  traces are generated by the following grammar:
$$
\traceone \bnf \emptytr \mid \concat {\app\valone} \traceone
$$
We define the probability that a program accepts a trace inductively
on the length of the trace, as follows:
\begin{align*}
\probtr {\abstr \varone \termone} {\emptytr} &= 1;\\
\probtr{\abstr \varone \termone}{\concat {\app\valone} \traceone} & =\probtr{\subst \termone \varone \valone}{\traceone}; \\
\probtr \termone {\traceone} &=\sum_{\valone}\sem \termone(\valone) \cdot \probtr {\valone} {\traceone} \quad\text{ if } \termone \not \in \valset.
\end{align*}
Please observe that the probability that a term $\termone$ accepts a
trace $\traceone = {\app {\valone_1}}\cdots \app {\valone_n}$ is the
probability of convergence of $\termone \valone_1 \cdots
\valone_n$. We are now going to define a metric, based on the
probability that programs accept arbitrary traces:
\condskip
\begin{definition}
  Let $\termone,\termtwo$ be two programs. Then we define the
  \emph{trace distance} between them as
  $\appl{\metrtr}{\termone}{\termtwo}=\sup_{\traceone}\abs{\probtr
    \termone \traceone - \probtr \termtwo \traceone}$.
  One can then define \emph{trace equivalence}
  and the \emph{trace preorder}, in the expected way.

\end{definition}
\condskip 
Please observe that $\metrtr$ is a pseudometric on programs in the sense of
Definition \ref{defmetric}, and that it is an adequate one. The kernel
of $\metrtr$ is nothing more than trace equivalence.
\condskip
\begin{example}
   $\appl \metrtr {\identity}{\diver} = 1$: we have to find a
    trace that separates them as much.  It is enough to consider the
    empty trace: it holds that $\probtr \emptytr \identity = 1$, and
    $\probtr \emptytr \diver = 0$. 
  The trace distance $\appl \metrtr {\psum \identity \diver}
  {\identity}$ between $\psum{\identity}{\diver}$ and $\identity$ is
  $\frac 1 2$.  Showing that it is greater than $\frac 1 2$ is easy:
  it is sufficient to consider the empty trace. The other inequality,
  requires evaluating, for any trace $\traceone$, the probability of
  accepting it. This is however much easier than dealing with all
  contexts, because we can now control the structure of the overall
  program we obtain: for any trace $\traceone =\app \valone_1 \cdots
  \app\valone_n$, we can see that: $\probtr{\psum \identity
    \diver} {\traceone} = \frac 1 2 \cdot {\sumdistr{ \sem{\valone_1 \cdots
        \valone_n}}}$, and $\probtr{\identity} {\traceone} = \sumdistr
         {\sem{\valone_1 \cdots \valone_n}}$. The difference (in
         absolute value) between $\probtr {\psum \identity
    \diver}{\traceone}$ and $\probtr{\identity}{\traceone}$, then, cannot be greater than $\frac{1}{2}$.
\end{example}
\condskip 
The trace distance and the context distance indeed \emph{coincide}, as
well as the trace and context preorder, and the trace and context
equivalence. In the rest of this section, we will give the details of
the proof for the pseudometric case, but the proof is similar for
$\equivctx{}{}$ and $\leqctx{}{}$. It is easy to realise that the
context distance is a lower bound on the trace distance, since any
trace $\app\valone_1 \cdots \app\valone_n$ can be seen as the context
$\hole \valone_1 \cdots \valone_n$.
\condskip
\begin{lemma}
 $\metrctx \leq \metrtr$
\end{lemma}
\begin{proof}
  For any trace $\traceone = \app\valone_1 \cdots \app\valone_n$ which
  separates $\termone$ and $\termtwo$ of $\epsone$, we can easily
  construct a context which separates them of the same quantity: just
  take $\ctxone = \hole \valone_1 \cdots \valone_n$.
\end{proof}
\condskip
%{We thus focus on non-expansiveness.}
%%%%%%%%%%%%%%%%%%%%%%%%%%%%%%%%%%
\subsection{Non-Expansiveness}
%%%%%%%%%%%%%%%%%%%%%%%%%%%%%%%%%%
Are there contexts that can separate strictly more than traces? In
order to show that it is not the case, it is enough to show that
$\metrtr$ is non-expansive:
\condskip
\begin{theorem}\label{theo:tracenonexpansive}
Let be $\termone$ and $\termtwo$ two programs, and let be $\ctxone$ a
context. Then $\appl\metrtr{\fillc \ctxone\termone}{\fillc \ctxone
  \termtwo} \leq \appl\metrtr \termone \termtwo $.
\end{theorem}
\condskip 
Since $\metrtr$ is adequate, we can conclude that trace metric and 
context metric actually coincide:
\condskip
\begin{theorem}
$\metrctx=\metrtr$.
\end{theorem}
\condskip
The rest of this section is devoted to an outline of the proof of
Theorem~\ref{theo:tracenonexpansive}.
The proof we give here is roughly inspired by the proof of
congruence of trace equivalence for a non-deterministic
$\lambda$-calculus~\cite{DengZhang}. The overall structure of the
proof is the following: we first express the capacity of a program to
do a trace by means of a labelled transition system (LTS in the following) $\ltstrace$ 
whose states are distributions over programs. Then we consider
\emph{another} LTS $\ltstracect$, where the states are distributions over pairs of
contexts and programs that intuitively models the execution of $\fillc
\ctxone \termone$, but keeps the evolution of $\ctxone$ and $\termone$
apart.
\subsubsection{The LTS $\ltstrace$}
The first LTS, called $\ltstrace$, has distributions over programs as
states, and traces as actions. We indicate with $\wtrc{}{\cdot}{}$ the
transition relation associated to $\ltstrace$.
We're in fact going to define 
it on top of an auxilliary labelled relation $\wtr{}{\cdot}{}$. 
Intuitively, the idea behind $\wtr{}{\cdot}{}$ is to consider a term as a process who can make actions.
There are two kinds of possible actions: an internal action
$\internact$, which corresponds to the internal reduction of the term,
and external actions, which corresponds to the application of an
argument $\valone$ to the term. More precisely, this labelled relation 
$\wtr{}{\cdot}{}$ is defined as a subset of the set $\distrs\programs \times \actions \times \distrs \programs$
, where the set $\actions$ is defined as:
\condskip
\begin{definition} 
 We define the set of actions $\actions$ by:
 $$\actions = \{\internact \} \bigcup  \left\{ \app \valone \mid \valone \text{ value }\right\}.$$
\end{definition}
\condskip
Intuitively, a
$\internact$-step corresponds to an internal computation step for any
term in the support of the distribution, while a $\app{\valone}$-step
corresponds to an interaction with the environment, which provides
$\valone$ as an argument.
\condskip
\begin{definition}
We define a labelled transition relation $\wtr{}{\cdot}{} \subseteq \distrs\programs \times \actions \times \distrs \programs$,
by the rules of Figure \ref{tracesemunl}. (We write $\distrone\disjplus\distrtwo$ for $\distrone + \distrtwo$ when we want to insist on $\distrone$ and $\distrtwo$ to have disjoint supports).
\end{definition}
 \condskip
\begin{figure}[h]
\begin{center}
\fbox{
\begin{minipage}{\figurewidth}
$$
\AxiomC{$\ssp \termone \distrtwo$}
\UnaryInfC{$\wtr {\distrone \dotplus \alpha \cdot \dirac{\termone}}{\internact} {\distrone + \alpha \cdot \distrtwo}$}
\DisplayProof
\qquad
\AxiomC{$\distrone \text{ value distribution }$}
\UnaryInfC{$\wtr  {\distrone}{\app \valone}{\sum_\termone \distrone(\abstr \varone \termone)\cdot \dirac{\subst \termone \varone \valone }}$}
\DisplayProof$$
\end{minipage}}
\end{center}
\caption{One-step Trace Relations on Program Distributions.}\label{tracesemunl}
\end{figure}
 The relation $\wtrc{}{\cdot}{}$ is defined
as the accumulation of several steps of $\wtr{}{\cdot}{}$. 
We define now the LTS $\ltstrace$.
\condskip
\begin{definition}
We define the LTS $\ltstrace$ as:
\begin{itemize}
\item Its set of states is $\distrs \programs$
\item Its set of labels is the set of traces $\words$.
\item Its transition relation $\wtrc{}{\cdot}{}$ is defined by by the rules of Figure \ref{tracesemdeuxl}
\end{itemize}
\end{definition}
\condskip
\begin{figure}[h]
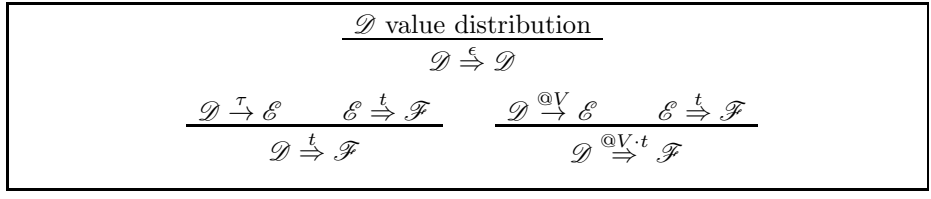

\begin{center}
\fbox{
\begin{minipage}{\figurewidth}
{
$$
\AxiomC{$\distrone$ value distribution}
\UnaryInfC{$\wtrc \distrone \emptytr \distrone $}
\DisplayProof
 $$
$$
\AxiomC{$\wtr \distrone \internact \distrtwo $}
\AxiomC{$\wtrc \distrtwo \tracetwo \distrthree $}
\BinaryInfC{$\wtrc  {\distrone}{ \tracetwo} \distrthree$}
\DisplayProof \qquad
\AxiomC{$\wtr \distrone {\app \valone} \distrtwo $}
\AxiomC{$\wtrc \distrtwo \tracetwo \distrthree $}
\BinaryInfC{$\wtrc  {\distrone}{\concat {\app\valone} \tracetwo} \distrthree$}
\DisplayProof
$$}
\end{minipage}}
\end{center}
\caption{Small-step Trace Relations on Program Distributions.}\label{tracesemdeuxl}
\end{figure}

Please observe that these relations are \emph{not} probabilistic. The
relation $\wtr {} \internact {}$ is non-deterministic, since at
each step we can decide which term of the distribution we
want to reduce. However, $\wtr{}\internact{}$ is strongly
normalising and confluent.
\condskip
\begin{lemma}\label{strongnormp}
The relation:
$\wtr {} \internact {}$ is strongly normalising 
\end{lemma}
\begin{proof}
\begin{itemize}
\item We show first that it is terminating: for a term $\termone$, we define a quantity $\sizec \termone \in \NN$ which corresponds to the size of the term:
$$\sizec{\diver}  = 0; \qquad \sizec \varone  = 1; \qquad \sizec {\abstr \varone \termone} = 1 + \sizec{\termone};$$ 
$$ \sizec{\termone \termtwo} = \sizec \termone + \sizec \termtwo; \qquad \sizec{\psum \termone \termtwo} = 1 + \max \{\sizec{\termone}, \sizec \termtwo\};$$
Since our $\lambda$-calculus is linear, $\sizec \termone$ decreases during the execution for every program $\termone$. More precisely: 
If $\ssp \termone \distrone$, then for every $\termtwo \in \supp \distrtwo$, $\sizec \termtwo < \sizec \termone$. (It is easily checked by observing the rules of $\ssp {}{}$).\\
Moreover, if $\ssp \termone \distrone$, then the cardinal of $\supp \distrone$ is at most 2. So, if for a distribution $\distrone$ we note: 
 $$\sizec {\distrone} = \sum_{( \termone) \in \supp \distrone} 3^{\sizec{ \termone}}$$, we can see that: for every $\distrone$, if $\wtr {\distrone}\internact{\distrtwo}$, we have that $\sizec \distrtwo  < \sizec \distrone$.
\item Moreover, let $\distrone$ be a distribution over program, and let be $\distrtwo$ such that $\wtr{\distrone}{\internact^n}{{\distrtwo}}$, and $\distrtwo$ is a normal form for $\wtr{}{\internact}{}$. Then we are going to show by induction over $n \in \NN$ that $\distrtwo = \sum_{\termone \in \supp \distrone}\distrone (\termone) \cdot \sem \termone$.
\begin{itemize}
\item if $n = 0$, then $\distrone = \distrtwo$, and moreover $\distrone$ is a distribution over values. So the result holds.
\item if $n >0$, it is a consequence of rules of Figure \ref{tracesemunl}.
\end{itemize}
%: indeed, the non-determinism in the relation $\wtr{}{\internact}{}$ comes from the differents $\termone$ we can pick in $\support \distrone$ and choose to reduce. 
 
\end{itemize}
\end{proof}
\condskip
Moreover, we can show that $\wtrc{}{\cdot}{}$ is in fact deterministic. That is, we have the following Lemma:
\condskip
\begin{lemma}
For every trace $\traceone$, for every $\distrone$, there exist an only one $\distrtwo$ such that $\wtrc \distrone \traceone \distrtwo$.
\end{lemma}
\condskip 
The interest of the relation $\wtrc{}{\cdot}{}$ is that it gives an
alternate formulation for the probability that a program succeeds in
doing a trace:
\condskip
\begin{lemma}
Let be $\termone$ a program, $\traceone$ a trace, and let be $\distrtwo$ the distribution such that $\wtrc {\dirac \termone}{\traceone}{\distrtwo}$. Then
$\probtr {\traceone}{\termone} = \sumdistr \distrtwo.$ 
\end{lemma} 
\condskip
In fact, the labelled transition system $\ltstrace$ allows us to extend the notion of probability of success for a trace to the case where we start not from a program, but from a probability distribution over program: 
  $$\probtr \traceone \distrone = \sumdistr \distrtwo \text{ when }
\wtrc \distrone \traceone \distrtwo$$ In the same way we extend the
preorder $\leqtr\cdot\cdot$, the equivalence relation
$\equivtr\cdot\cdot$, and the metric $\metrtr$ to distributions.
We can now use the relation $\wtrc{}{\cdot}{}$ to give an equivalent
formulation of Theorem \ref{theo:tracenonexpansive}: if $\termone$ and $\termtwo$
are such that $\appl \metrtr \termone \termtwo \leq \epsone$, then for
every trace $\traceone$, and context $\ctxone$, if $\wtrc {\dirac
  {\fillc \ctxone\termone}}{\traceone}{\distrone}$ and $\wtrc {\dirac
  {\fillc \ctxone \termtwo}}{\traceone}{\distrtwo}$, then it holds
that $\abs{\sum \distrone - \sum \distrtwo} \leq \epsone$. This
statement, however, cannot be proved directly, yet, because the way
$\ctxone$ and the argument terms interact is lost.

\subsubsection{The LTS $\ltstracect$}
It is then time to introduce our second LTS, called $\ltstracect$,
which will allow us to relate $\wtrc {\dirac{\fillc \ctxone
    \termone}}{\cdot}{\cdot}$ to the behaviour of $\termone$: we want
to talk about the evolution of a system consisting of the program
$\termone$ and the environment $\ctxone$, while keeping the system and
the environment as separate as possible. $\ensct$ is the set
of pairs of the form $(\ctxone, \termone)$, where $\ctxone$ is a
context and $\termone$ is a program. The states of
  $\ltstracect$ are distributions over $\ensct$, and the labels of
  $\ltstracect$ are traces. The transition relation of $\ltstracect$
  corresponds to the transition relation of $\ltstrace$, where we keep
  the information about what part of the whole system is the program,
  and what part is the environment interacting with it. We'll use
  the following notation, which will be useful in the formal
  definition of $\ltstracect$: If $\termtwo$ is a term, $\distrone$ a
  distribution over $\ensct$, we define $\distrone \cdot \termtwo$ and
  $\termtwo \cdot \distrone$ as the distributions over $\ensct$ given
  by:
\begin{align*}
\distrone \cdot \termtwo &= \sum_{(\ctxone, \termone) \in \supp \distrone} \distrone(\ctxone, \termone)\cdot \dirac{(\ctxone\termtwo, \termone)}\\
\termtwo \cdot \distrone & = \sum_{(\ctxone, \termone) \in \supp \distrone} \distrone(\ctxone, \termone)\cdot \dirac{(\termtwo\ctxone, \termone)} 
\end{align*}
And if $\ctxone$ is a context, $\termone$ a term, and $\distrone$ a distribution over $\programs$, we define $(\ctxone \distrone, \termone)$ and $(\distrone \ctxone, \termone)$ as the distributions over $\ensct $ given by:
\begin{align*}
(\ctxone \distrone, \termone) &= \sum_{\termtwo \in \supp \distrone}\distrone(\termone \dirac{(\ctxone\termtwo, \termone)}) \\
(\distrone \ctxone, \termone) &= \sum_{\termtwo \in \supp \distrone}\distrone(\termone \dirac{(\termtwo\ctxone, \termone)}) 
\end{align*}
If $(\ctxone, \termone) \in \ensct$ is such that $\fillc \ctxone \termone$ is a value, we say by abuse of notation that $(\ctxone, \termone)$ is a value. If $\distrone$ is a distribution over $\ensct$ such that every $(\ctxone, \termone) \in \supp \distrone$ is a value, we say that $\distrone$ is a value distribution. 
%\remarque{ (cela aurait du aller dans la version courte, mais je crois qu'on a oublie.)} 
\condskip
\begin{definition}
We define  $\ltstracect$ as the labelled transition system such that:
\begin{itemize}
\item its set of states is the set of probability distributions over $\ensct$.
\item its set of labels is the set of traces.
\item its transition relations $\wtrcct{}{\cdot}{}$ is defined by the rules of Figure \ref{tracesemct}. The definition uses an auxiliary one-step transition relation $\wtrct{\distrone}{\actone}{\distrtwo}$, where $\actone \in \actions$, and $\distrone$, $\distrtwo$ are distributions over $\ensct$.
\end{itemize}
\end{definition}
\condskip

\begin{figure}[!h]
\begin{center}
\fbox{\footnotesize
\begin{minipage}{.97\textwidth}
$$\AxiomC{$\ssp \termtwo \distrtwo $}
  \UnaryInfC{$\wtrct{\distrone \disjplus p \cdot \dirac{(\termtwo, \termone)}}{\internact}{\distrone + p\cdot (\distrtwo, \termone) }$} 
  \DisplayProof 
  \qquad
  \AxiomC{$ $}
  \UnaryInfC{$ \wtrct{p \cdot \dirac{(\abstr\varone\termtwo, \termone)}}{\app\valone}{ p\cdot (\subst\termtwo\varone \valone, \termone) }$}
  \DisplayProof $$ 
  $$
  \AxiomC{$\wtrct {\dirac{(\ctxone, \termone)}} \internact \distrtwo $}
  \UnaryInfC{$\wtrct{\distrone \disjplus p \cdot \dirac{(\ctxone\termtwo, \termone)}}{\internact}{\distrone + p\cdot (\distrtwo\cdot \termtwo) }$}
  \DisplayProof
  $$
  $$
  \AxiomC{$ \fillc \ctxone \termone \text{ is a value } $}
  \AxiomC{$\ssp \termtwo \distrtwo$}
  \BinaryInfC{$\wtrct{\distrone \disjplus p\cdot \dirac{(\ctxone\termtwo, \termone)}}{\internact}{\distrone + p\cdot(\ctxone\distrtwo, \termone)}$}
  \DisplayProof 
 \qquad
  \AxiomC{$\ssp \termtwo \distrtwo $}
  \UnaryInfC{$\wtrct{\distrone \disjplus p \cdot \dirac{(\termtwo\ctxone, \termone)}}{\internact}{\distrone + p \cdot (\distrtwo\ctxone, \termone) }$}
  \DisplayProof
$$
$$
  \AxiomC{$\wtrct{\dirac{(\ctxone, \termone)}}{\internact}{\distrtwo}$}
  \UnaryInfC{$\wtrct{\distrone \disjplus p \cdot \dirac{(\valone\ctxone, \termone)}}{\internact}{\distrone + p \cdot(\valone\distrtwo, \termone)} $}
  \DisplayProof
  \qquad
\AxiomC{$\ssp \termone  \distrtwo$}
\UnaryInfC{$\wtrct{\distrone \disjplus p\cdot \dirac{\hole, \termone}}{\internact}{\distrone + p\cdot (\hole, \distrtwo)}$}
\DisplayProof 
$$
$$
\AxiomC{}
\UnaryInfC{$\wtrct{\distrone \disjplus p \cdot\dirac{\hole\valone, \abstr\varone\termtwo}}{\internact}{\distrone + p \cdot\dirac {\hole, \subst \termtwo \varone \valone}}$}
\DisplayProof
\qquad
\AxiomC{$\fillc \ctxone \termone$ \text{value}}
\UnaryInfC{$\wtrct{\distrone \disjplus p \cdot\dirac{(\abstr\varone \termtwo)\ctxone, \termone}}{\internact}{\distrone + p \cdot \dirac {\subst \termtwo \varone \ctxone, \termone}}$}
\DisplayProof
$$ $$
 \AxiomC{}
 \UnaryInfC{$\wtrct {p \cdot \dirac{\abstr \varone \ctxone, \termone}}{\app \valone}{p \cdot\dirac{({\subst\ctxone \varone \valone}, \termone)}}  $} 
\DisplayProof
\qquad
\AxiomC{$ $}
\UnaryInfC{$\wtrct{p \cdot \dirac{(\hole, \abstr\varone \termone)}}{\app \valone}{p \cdot \dirac {(\hole, \subst \termone \varone \valone)}}$}
\DisplayProof
$$ $$
\AxiomC{$\wtrct {\distrone_i}{\app \valone}{\distrtwo_i} $}
%\AxiomC{$\forall i, \distrone_i \text{ value distribution } $}
\UnaryInfC{$\wtrct {\stackrel{\cdot}{\sum_i} \distrone_i } {\app \valone} {\sum_i{\distrtwo_i}} $}
\DisplayProof
$$
$$
\AxiomC{$\distrone$ value distribution}
\UnaryInfC{$\wtrcct \distrone \emptytr \distrone $}
\DisplayProof
\qquad
\AxiomC{$\wtrct \distrone \internact \distrtwo $}
\AxiomC{$\wtrcct \distrtwo \tracetwo \distrthree $}
\BinaryInfC{$\wtrcct  {\distrone}{ \tracetwo} \distrthree$}
\DisplayProof 
\qquad
\AxiomC{$\wtrct \distrone {\app \valone} \distrtwo $}
\AxiomC{$\wtrcct \distrtwo \tracetwo \distrthree $}
\BinaryInfC{$\wtrcct  {\distrone}{\concat \valone \tracetwo} \distrthree$}
\DisplayProof
$$
\end{minipage}}
\end{center}
\caption{small-step trace relations on distributions over $\ensct $ }\label{tracesemct}
\end{figure}

\condskip
\begin{lemma}\label{strongnormct}
The relation $\wtr {}\internact{}$ on distributions over $\ensct$ is
strongly normalising, and normal forms of $\wtr{}{\internact}{}$ are
value distributions.
\end{lemma}
\begin{proof}
The proof is exactly the same that for the relation
$\wtr{}{\internact}{}$ for distribution over programs. We extend the
definition of $\sizec \cdot$ to distribution over $\ensct$, by:
$\sizec {\distrone} = \sum_{(\ctxone, \termone) \in \supp \distrone}
3^{\sizec{\fillc \ctxone \termone}}$, and we do the same reasoning.
\end{proof}
\condskip
For $\distrone$ a distribution over $\ensct$, we note $\normal
\distrone$ the normal form of $\distrone$ for the relation $\wtr
         {}{\internact}{}$. Please observe that Lemma
         \ref{strongnormct} implies that for any distribution
         $\distrone$, there exists only one distribution $\distrtwo$
         such that $\wtrc \distrone \emptytr \distrtwo$, and moreover
         $\distrtwo = \normal \distrone$. The
  trace semantics for distributions over $\ensct$ allows us to extend
  the notions of trace equivalence, trace preorder and trace metric on
  distributions over $\ensct$ in a natural way.

\subsubsection{Relating $\ltstrace$ and $\ltstracect$}
Intuitively, considering a semantics for distributions over $\ensct$ allows us to separate the part of the semantics which talks about the program, and the part which talks about the context. We would like to obtain the trace semantics for $\fillc \ctxone \termone$, just by looking at the semantics of $(\ctxone , \termone)$. We are going to express this idea by relating the two trace semantics.
\condskip
\begin{lemma}\label{lemprobtreq}
Let be $\termone$ a closed term, $\ctxone$ a context and $\traceone$ a trace. Let be $\distrone$ and $\distrtwo$ such that  $\wtrc {\dirac{\fillc \ctxone\termone}}{\traceone}{\distrone}$, and $\wtrcct {\dirac{(\ctxone, \termone)}}{\traceone}{\distrtwo}$
Then ${\sumdistr \distrone}  = {\sumdistr \distrtwo } $.
\end{lemma}
\begin{proof}
The proof of Lemma \ref{lemprobtreq} is relatively technical, and is based on three auxilliary lemmas : Lemma \ref{auxun}, and Lemma \ref{lemmesept}.
If $\distrone$ is a distribution over $\ensct$, we call $\forget \distrone$ the distribution obtained by filling each context by its associated term. 
To express this idea more formally, we define an operator $\forget{}$ on distributions over $\ensct$, which transforms every distribution in its corresponding distribution over terms.
 $$\forget  \distrone = \sum_{\ctxone, \termone}\distrone{((\ctxone, \termone))}\cdot \dirac{\fillc \ctxone \termone}. $$ 
 We can now express the correspondence between the trace semantics on distributions over programs, and the trace semantics on distributions over $\ensct$, by the following lemma.
\condskip
  \begin{lemma}\label{auxun}
Let  $\distrone, \distrtwo$ be  distributions over $\ensct$. If $\wtrcct \distrone \traceone \distrtwo$, then we have that: $\wtrc {\forget \distrone}{\traceone}{\forget {\distrtwo}} $
\end{lemma} 
\condskip
But we would like to have some information in the other directions too: if we have the trace semantics of the term $\fillc \ctxone\termone$, is it possible to deduce something about the trace semantics of $(\ctxone, \termone )$ ? The following lemma give a positive answer:
\condskip 
\begin{lemma}\label{lemmesept}
  Let  $\distrone$ be a distribution over $\ensct$ such that $\wtrc {\forget \distrone} \tracetwo \distrthree$.
  Let be  $\distrfour$ such that $\wtrcct \distrone \tracetwo \distrfour$. Then $\distrthree = \forget \distrfour$.
  \end{lemma}
\begin{proof}
We need first to show an auxiliary lemma, in order to express the correspondence between the one-step relation on distributions over programs, and the one-step relation on distributions over $\ensct$.
\condskip
\begin{lemma}\label{corresp}
Let be  $\ctxone$ a context, and $\termtwo$ a term.
 Let be $\distrone$ such that  $\ssp {\fillc \ctxone \termtwo} \distrone$. Then there exists $\distrtwo$ such that: $\wtrct {\dirac {(\ctxone, \termtwo)}}{\internact}{\distrtwo}$, and $\forget \distrtwo = \distrone$.
\end{lemma}
\condskip
Using this lemma we are now going to show Lemma \ref{lemmesept}.
The proof is by induction on the derivation  of $\wtrc {\forget \distrone} \tracetwo \distrthree$:
\begin{itemize}
\item The basic case is the case where $\forget \distrone$ is a value distribution (and consequently a normal form for $\wtr \cdot \internact \cdot$), and where we are interested in the empty trace. The the derivation tree of $\wtrc {\forget \distrone} \tracetwo \distrthree$ is of the form:
$$\AxiomC{$\forget \distrone$ value distribution}
\UnaryInfC{$\wtrc {\forget\distrone} \emptytr {\forget\distrone} $}
\DisplayProof$$
Then $\forget \distrone$ is a value distribution. By definition of values for distribution over $\ensct$, it means that $\distrone$ is a value distribution too. And so we can observe that $\wtrcct \distrone \emptytr \distrone$, and the result holds.
\item The first induction case is the case we don't start from a value distribution. Then  the derivation tree of $\wtrc {\forget \distrone} \tracetwo \distrthree$ is of the form:
$$\AxiomC{$\wtr {\forget\distrone} \internact \distrfour $}
\AxiomC{$\wtrc \distrfour \tracetwo \distrthree $}
\BinaryInfC{$\wtrc  {\forget \distrone}{ \tracetwo} \distrthree$}
\DisplayProof
$$ 
The only possible way to have obtained: $\wtr {\forget\distrone} \internact \distrfour $ is to have used a derivation of the form: 
$$\AxiomC{$\ssp \termone \distrsix $}
  \UnaryInfC{$\wtr{\forget \distrone = \distrfive \disjplus p \cdot \dirac{( \termone)}}{\internact}{\distrfour = \distrfive + p\cdot (\distrsix) }$} 
  \DisplayProof$$
Since $\forget \distrone = \distrfive \disjplus p \cdot \dirac{( \termone)}$, $\distrone = \distrseven \disjplus p \cdot \distreight$, with $\forget \distrseven = \distrfive$ and  $\forget \distreight = \dirac \termone$.
So for any $(\ctxone, \termtwo) \in \supp \distreight$, we have that $\ssp {\fillc \ctxone \termtwo} \distrsix$. By Lemma \ref{corresp}, there exist $\distrnine_{\ctxone, \termtwo}$ such that $\forget {\distrnine_{\ctxone, \termtwo}} = \distrsix$, and $\wtrct {\dirac{(\ctxone, \termtwo)}} \internact {\distrnine_{\ctxone, \termtwo}}$.
And now we can see by the rules of trace semantics for distributions over $\ensct$ that:
\begin{align} & \wtrct \distrone \internact {\cdots \wtrct {\cdots} \internact {\distrseven + p \cdot \sum_{(\ctxone, \termtwo) \in \supp \distreight} \distrnine_{(\ctxone, \termtwo)}}}\label{eq:eqn10} \end{align} 
 Moreover, we can see that $\forget {\distrseven + p \cdot \sum_{(\ctxone, \termtwo) \in \supp \distreight} \distrnine_{(\ctxone, \termtwo)}} = \distrfive + p \cdot (\distrsix) = \distrfour$.
So now we can apply the induction hypothesis, and we have that there exists a distribution $\distrtwo$ such that: 
\begin{align}
&\wtrcct  {\distrseven + p \cdot \sum_{(\ctxone, \termtwo) \in \supp \distreight} \distrnine_{(\ctxone, \termtwo)}} \traceone \distrtwo \label{eq:eqn20}\\
& \forget \distrtwo = \distrthree.
\end{align} And now we can conclude (by equations \eqref{eq:eqn10} and \eqref{eq:eqn20} ) that  $\wtrcct \distrone \traceone \distrtwo$, and so the results holds. 

\item The second induction case is the case where we start from a value distribution, and we are interested in a non-empty trace.
 Then  the derivation tree of $\wtrc {\forget \distrone} \tracetwo \distrthree$ is of the form:
$$\AxiomC{$\wtr {\forget\distrone} {\app \valone} \distrfour $}
\AxiomC{$\wtrc \distrfour \traceone \distrthree $}
\BinaryInfC{$\wtrc  {\forget \distrone}{\tracetwo = \concat {\app\valone} \traceone} \distrthree$}
\DisplayProof
$$ 
The only possible way to have obtained: $\wtr {\forget\distrone} {\app \valone} \distrfour $ is to have used a derivation of the form: 
$$
\AxiomC{$\forget \distrone \text{ value distribution }$}
\UnaryInfC{$\wtr  {\forget \distrone}{\app \valone}{\distrfour = \sum {\forget\distrone}(\abstr \varone \termone)\cdot \dirac{\subst \termone \varone \valone }}$}
\DisplayProof
$$
 For every $(\ctxone,\termtwo) \in \supp \distrone$, let be $\termone_{(\ctxone, \termtwo)}$such that  $\fillc \ctxone \termtwo = \abstr \varone \termone_{(\ctxone, \termtwo)}$. Using this notation, we can now express $\distrfour$ as a sum over the support of the distribution $\distrone$:  
\begin{align}
\distrfour &=  \sum_{(\ctxone, \termtwo) \in \supp \distrone } \distrone((\ctxone, \termtwo)) \cdot \dirac{(\subst {\termone_{(\ctxone, \termtwo)}} \varone \valone)}\label{eq:eqna1}
\end{align} 
We are going to define a distribution $\distrfive_{(\ctxone, \termtwo)}$ over $\ensct$ for every  $(\ctxone, \termtwo) $ in the support of $\distrone$. We can see that for every $(\ctxone,\termtwo) \in \supp \distrone$, we have two possible cases:
\begin{itemize}
\item Or $\ctxone = \hole$, and $\termtwo = \abstr \varone (\termone_{(\ctxone, \termtwo)})$.
Then let be $\distrfive_{(\ctxone, \termtwo)} = \dirac{(\hole, \subst {\termone_{(\ctxone, \termtwo)}} \varone \valone)}$
\item Or $\ctxone = \abstr \varone \ctxtwo$, and $\fillc \ctxone \termtwo = \termone_{(\ctxone, \termtwo)}$. Then let be $\distrfive_{(\ctxone , \termtwo)} = \dirac {(\subst \ctxtwo \varone \valone, \termtwo)}$. Please observe that, since the calculus is linear, $\subst \ctxtwo \varone \valone$ is indeed a context.
\end{itemize}
Now we can write the equation \eqref{eq:eqna1} the following way:
\begin{align}
\distrfour &=  \forget {\sum_{(\ctxone, \termtwo) \in \supp \distrone } \distrone((\ctxone, \termtwo)) \cdot  {\distrfive_{(\ctxone, \termtwo)}}}\label{eq:eqna2}
\end{align}
 Moreover, for every $(\ctxone, \termtwo) \in \supp \distrone$, we have:
$\wtrct {\dirac{\ctxone, \termtwo}}{\app \valone}{\distrfive_{(\ctxone, \termtwo)}}$, and so the rules of one-step trace semantics for distribution over $\ensct$ allow us to say that: 
\begin{align}
&\wtrct \distrone {\app \valone}  \sum_{(\ctxone, \termtwo) \in \supp \distrone } \distrone((\ctxone, \termtwo)) \cdot {\distrfive_{(\ctxone, \termtwo)}} \label{eq:eqna3}
\end{align}
 By applying the induction hypothesis to $\wtrc \distrfour \tracetwo \distrthree$ and using equation \eqref{eq:eqna2}, we know that there exists $\distrsix$ such that:
\begin{align}
& \wtrcct  {\sum_{(\ctxone, \termtwo) \in \supp \distrone } \distrone((\ctxone, \termtwo)) \cdot {\distrfive_{(\ctxone, \termtwo)}}} {\traceone}{\distrsix} \\
&\text{ and }  \forget \distrsix = \distrthree \label{eq:eqna4}
\end{align}
And now we can conclude by using the rules of trace semantics for distributions over $\ensct$ that $\wtrcct \distrone {\concat {\app \valone} \traceone} \distrsix$, and since we have equation \eqref{eq:eqna4} the result holds.
\end{itemize}
\end{proof}
\end{proof}
\condskip
\subsection{$\epsone$-parents distributions}
Lemma \ref{lemprobtreq} allows us
 to give yet another equivalent formulation of Theorem
\ref{theo:tracenonexpansive}: if $\appl \metrtr\termone \termtwo \leq
\epsone$, then if $\wtrcct {\dirac{(\ctxone,
    \termone)}}{\traceone}{\distrone}$ and $\wtrcct {\dirac{(\ctxone,
    \termtwo)}}{\traceone}{\distrtwo}$, it holds that $\abs {\sumdistr
  \distrone - \sumdistr \distrtwo} \leq \epsone$.  We are in fact
going to show a stronger result, which uses the notion of
$\epsone$-related distributions:
\condskip
\begin{definition}
We say that two distributions $\distrone$ and $\distrone'$ over
$\ensct$ are $\epsone$-related, and we note $\mleqpar \distrone
{\distrone'} \epsone $ if there exist $n \in \NN$, and
$\ctxone_1,...,\ctxone_n$ distinct contexts, $p_1,...,p_n$ positive real numbers with $\sum_i p_i \leq 1$, and
$\distrtwo_1,...,\distrtwo_n$, and $\distrtwo'_1,...,\distrtwo'_n$ distributions over $\programs$,
such that:
\begin{itemize}
\item $\distrone = \sum_{1 \leq i \leq n} p_i \cdot  (\ctxone_i, \distrtwo_i) $
\item $\distrone' = \sum_{1 \leq i \leq n} p_i \cdot  (\ctxone_i, \distrtwo'_i) $
\item $\forall i$, $\appl \metrtr {\distrtwo_i} {\distrtwo'_i} \leq \epsone$
\end{itemize}
\end{definition}
\condskip
Please observe that, if $\appl \metrtr \termone \termtwo \leq
\epsone$, then for every context $\ctxone$, the distributions
$\dirac{(\ctxone, \termone)}$ and $\dirac{(\ctxone, \termtwo)}$ are
$\epsone$-related. In fact, the notion of $\epsone$-relatedness is a
way to capture the idea of a pair of distributions over $\ensct$
representing \emph{the same} environment, in which we put programs
which are close for the trace pseudometric.  The following can be seen
as a stability result: if we start from $\epsone$-related
distributions, and we do a trace $\traceone$, we end up in two
distributions which are still $\epsone$-related.
\condskip
\begin{lemma}\label{lemmehuit}
Let be $\distrone$, $\distrtwo$ distributions over $\ensct$, and $\epsone \in [0,1]$ such that  $\mleqpar \distrone \distrtwo \epsone$. Let be $\traceone$ a trace.
Let be $\distrthree$ and $\distrfour$  such that:
 $\wtrcct \distrone \traceone \distrthree$, 
and  $\wtrcct \distrtwo \traceone \distrfour$.
Then  $\mleqpar \distrthree \distrfour \epsone$
%ici, j'ai modifie la version courte : on avait $\wtrc \distrone \traceone \distrthree$ et $\wtrc \distrtwo \traceone \distrfour$.
\end{lemma}

\begin{proof}
If $\distrone$ is a distribution over $\programs$, we note $\normal \distrone$ the distribution such that $\wtrc \distrone \emptytr {\normal \distrone}$. Please observe that it is the normal form of $\distrone$ for the transition relation $\wtr \cdot \internact \cdot$. Similarly, if $\distrone$ is a distribution over $\ensct$, $\normal \distrone$ the distribution such that $\wtrcct \distrone \emptytr {\normal \distrone}$.
We are first going to show two auxiliary lemma:
\condskip
\begin{lemma}\label{lemmehuitauxmetr}
Let be $\distrone$, $\distrtwo$ two distributions over $\ensct$ such that $\mleqpar \distrone \distrtwo\epstwo$. Then $\mleqpar{\normal \distrone}{\normal \distrtwo}\epstwo$ .
\end{lemma}
\begin{proof}
We will use $\kronecker{a}{b}$ as an integer being $1$ if $a$ is equal to $b$, and $0$ otherwise.  
 Let be $\distrthree$ any distribution over $\ensct$. We note $$n_\text{max}(\distrthree) = \max\{n \mid \wtrct \distrthree {\internact^n} {\normal\distrthree}\} $$
  We are going to show the lemma by induction on: $n = \max{(n_\text{max}(\distrone), n_\text{max}(\distrtwo))}$.

\begin{itemize}
 \item If $n = 0$ then $\distrone^\star = \distrone$, and $\distrtwo = \distrtwo^\star$, and the result holds.
 \item If $n>0$:
 Then we have: there exist $p_1,...,p_n$, and $\distrtwo_1,...,\distrtwo_n$, and $\distrtwo'_1,...,\distrtwo'_n$ such that:
\begin{align*}
 \distrone = &\sum_{1 \leq i \leq n} p_i \cdot  (\ctxone_i, \distrone_i) \\
\distrtwo = & \sum_{1 \leq i \leq n} p_i \cdot  (\ctxone_i, \distrtwo_i) \\
\forall i, & \appl \metrtr {\distrone_i}  {\distrtwo_i} \leq \epsone
\end{align*}

 Then there exists $i$ such that: 
 there exists $\termone \in \supp{\distrone_i} \cup \supp{ \distrtwo_i}$, such that $\fillc {\ctxone_i} \termone$ is not an irreducible term.
We consider every possible case for the form of $\fillc {\ctxone_i} \termone$:

 \begin{itemize}
 \item or $\ctxone_i$ is an evaluation context, and there exist $\termone \in \supp{\distrone_i} \cup \supp{ \distrtwo_i}$, such that $\termone$ is not an irreducible term. Intuitively, we want to reduce $\distrone_i$ and $\distrtwo_i$ as much as possible, since they are in evaluation position. And the two resulting distributions should be again $\epsone$-related distributions. 
More precisely,  for every $\termone \in \supp \distrone_i\cup\supp \distrtwo_i$, we note $\distrthree_\termone = \normal{\left(\dirac \termone\right)}$. 
Then  the rules of $\wtrct {} \internact {}$ allow us to see that there exist $k_1$ and $k_2$ such that $k_1 + k_2 > 0$, and:
$$\wtrct{\distrone}{ \internact^{k_1}}{\distrone' = \sum_{j\neq i}p_j\cdot{(\ctxone_j, \distrone_j)} +p_i {(\ctxone_i, \sum_{\termone}\distrone_i(\termone)\cdot \distrthree_\termone)} }, $$ and 
 $$ \wtrct{\distrtwo}{ \internact^{k_2}}{\distrtwo' = \sum_{j\neq i}p_j\cdot{(\ctxone_j, \distrtwo_j)} +p_i {(\ctxone_i, \sum_{\termone}\distrtwo_i(\termone)\cdot \distrthree_\termone)} }.$$
  We can easily show that $\mleqpar  {\distrone'}{\distrtwo'} \epsone$, and moreover, $\max{(n_\text{max}(\distrone'), n_\text{max}(\distrtwo'))} < \max{(n_\text{max}(\distrone), n_\text{max}(\distrtwo))}$. So we can apply the induction hypothesis, and we have that $\mleqpar{\normal\distrone} {\normal \distrtwo}\epsone$.
 \item If $\ctxone_i$ is such that the reduction depends only of the $\ctxone_i$, that is if there exist $q_1,..,q_m$,..., such that for every term $\termtwo$, $$\ssp{\fillc{\ctxone_i}{\termtwo}} {q_1 \cdot \fillc {\ctxtwo_1} \termtwo +...+ q_m \cdot \fillc {\ctxtwo_m}{\termtwo}}.$$
  Then the rules of $\wtrct {} \internact {}$ allows us to show that there exist $k_1$ and $k_2$, such that $k_1 + k_2 > 0$:
 $$\wtrct{\distrone}{ \internact^{k_1}}{\distrone' = \sum_{j\neq i}p_j\cdot{(\ctxone_j, \distrone_j)} +p_i \cdot \sum_{1 \leq k \leq m} q_k \cdot {(\ctxtwo_k, \distrone_i)}}$$ and that 
  $$\wtrct{\distrtwo}{ \internact^{k_2}}{\distrtwo' = \sum_{j\neq i}p_j\cdot{(\ctxone_j, \distrtwo_j)} +p_i \cdot \sum_{1 \leq k \leq m} q_k \cdot {(\ctxtwo_k, \distrtwo_i)}}.$$
In the definition of $\epsone$-related distribution, we consider contexts $(\ctxone_j)_j$ disjoints. So since we want to show that the new distributions we have obtained are still $\epsone$-related, we have to regroup the identical contexts (for instance, it can be the case that: $\ctxtwo_k = \ctxone_j$):
We note $\textbf C =  \{(\ctxone_j)_{j \neq i} \cup (\ctxtwo_k)_{1\leq k \leq m} \}$ the set of all contexts that can have been obtained at this step.
For $\ctxone \in \textbf C$, we take $p'_\ctxone$ his total probability: $p'_\ctxone = \sum_{j \neq i} {\kronecker{\ctxone} {\ctxone_j}} \cdot p_j + \sum_{1\leq k\leq m}\kronecker{\ctxone}{ \ctxtwo_k} \cdot p_i\cdot q_k$, and similarly: $$\distrone'_\ctxone = \sum_{j \neq i} {\kronecker{\ctxone}{ \ctxone_j}}\cdot\frac{p_j}{p'_\ctxone}\cdot\distrone_j + \sum_{1\leq k\leq m}{\kronecker{\ctxone} {\ctxtwo_k}}\cdot\frac{p_i \cdot q_k}{p'_\ctxone} \cdot \distrone_i$$, and
 $$\distrtwo'_\ctxone = \sum_{j \neq i} {\kronecker{\ctxone} {\ctxone_j}}\cdot\frac{p_j}{p'_\ctxone}\cdot\distrtwo_j + \sum_{1\leq k\leq m}\kronecker{\ctxone} {\ctxtwo_k}\cdot\frac{p_i \cdot q_k}{p'_\ctxone} \cdot \distrtwo_i.$$
And now we have $$\distrone' = \sum_{\ctxone \in \textbf C} p'_\ctxone (\ctxone, \distrone'_\ctxone) $$, and similarly: 
$$\distrtwo' = \sum_{\ctxone \in \textbf C} p'_\ctxone (\ctxone, \distrtwo'_\ctxone), $$ and for every $\ctxone \in \textbf C$, $\appl \metrtr {\distrone'_\ctxone}{\distrtwo'_\ctxone} \leq \epsone$. 

 \item The last case is the case where the term and the context really interact: more precisely,
 $\distrone_i$ is a value distribution, and moreover:
 \begin{itemize}
 \item Either $\ctxone_i = \fillc\ctxtwo{\hole \valone}$, which means that we are in the case where the contexts pass values to the program.  Then the following facts are derivable with the rules of $\wtrct  {}\internact{}$:
\begin{align*}
 \wtrct{\distrone}{ \internact^{k_1}}{\distrone' & = \sum_{j\neq i}p_j\cdot{(\ctxone_j, \distrone_j)} \\ & +p_i  \cdot {(\ctxtwo, \sum_{\abstr \varone \termtwo}\distrone_i(\abstr \varone \termtwo) \cdot \dirac{\subst \termtwo \varone \valone})}} 
\end{align*}
and 
\begin{align*}
 \wtrct{\distrtwo}{ \internact^{k_2}}{\distrtwo' & = \sum_{j\neq i}p_j\cdot{(\ctxone_j, \distrtwo_j)}\\ & +p_i  \cdot {(\ctxtwo, \sum_{\abstr \varone \termtwo}\distrtwo_i(\abstr \varone \termtwo) \cdot \dirac{\subst \termtwo \varone \valone})}}
\end{align*} 
  We are now going to show that $\distrone'$ and $\distrtwo'$ are $\epsone$-related.
We should again regroup the identical contexts (for instance, it can be the case that: $\ctxtwo = \ctxone_j$):
We note $\textbf C =  \{(\ctxone_j)_{j \neq i} \cup (\ctxtwo) \}$ the set of all contexts that can have been obtained at this step.
For $\ctxone \in \textbf C$, we take $p'_\ctxone$ his total probability defined as $p'_\ctxone = \sum_{j \neq i} \kronecker{\ctxone} {\ctxone_j} \cdot p_j + \kronecker{\ctxone}{ \ctxtwo} \cdot p_i$, 
and we obtain: 
\begin{align*}
\distrone'_\ctxone = &\sum_{j \neq i} \kronecker{\ctxone}{\ctxone_j}\cdot\frac{p_j}{p'_\ctxone}\cdot\distrone_j \\& + \kronecker{\ctxone}{\ctxtwo}\cdot\frac{p_i}{p'_\ctxone} \cdot \sum_{\abstr \varone \termtwo}\distrone_i(\abstr \varone \termtwo) \cdot \dirac{\subst \termtwo \varone \valone}
\end{align*}
and 
\begin{align*}  
\distrtwo'_\ctxone = & \sum_{j \neq i} \kronecker{\ctxone}{ \ctxone_j}\cdot\frac{p_j}{p'_\ctxone}\cdot\distrtwo_j \\&+ \kronecker{\ctxone}{ \ctxtwo}\cdot\frac{p_i}{p'_\ctxone} \cdot \sum_{\abstr \varone \termtwo}\distrtwo_i(\abstr \varone \termtwo) \cdot \dirac{\subst \termtwo \varone \valone}
\end{align*}
  We have that: $\forall j \neq i$,
 $\appl \metrtr {\distrone'_j = \distrone_j}{\distrtwo'_j = \distrtwo_j} \leq \epsone$ by hypothesis, and moreover,  for every trace $\traceone$:  
\begin{align*}
& |{\probtr{\left(\sum_{\abstr \varone \termtwo}\distrone_i(\abstr \varone \termtwo) \cdot \dirac{\subst \termtwo \varone \valone}\right)}{\traceone}} \\ 
& -  \probtr{\left(\sum_{\abstr \varone \termtwo}\distrtwo_i(\abstr \varone \termtwo) \cdot \dirac{\subst \termtwo \varone \valone}\right)}{\traceone}| \\
= &| \sum_{\abstr \varone \termtwo}\distrone_i(\abstr \varone \termtwo) \cdot \probtr{\subst \termtwo \varone \valone}{\traceone}\\
& - \sum_{\abstr \varone \termtwo}\distrtwo_i(\abstr \varone \termtwo) \cdot \probtr{\subst \termtwo \varone \valone}{\traceone}|
\\=& {\sum_{\abstr \varone \termtwo}\distrone_i(\abstr \varone \termtwo) \cdot \probtr{{\abstr \varone \termtwo})}{\concat {\app \valone}\traceone}}\\& -   {\sum_{\abstr \varone \termtwo}\distrone_i(\abstr \varone \termtwo) \cdot \probtr{{\abstr \varone \termtwo})}{\concat {\app \valone}\traceone}}|\\
  =| &\probtr {\distrone_i}{\concat {\app \valone}\traceone}\\
 & -  \probtr {\distrtwo_i}{\concat {\app \valone}\traceone}|\\
   \leq & \epsone
  \end{align*}
Since the relation $\appl \metrtr \cdot \cdot \leq \epsone$ on terms distribution is stable by convex summations, the result holds. 
 \item Or $\ctxone_i = \fillc\ctxtwo{{\abstr\varone\termtwo}\hole}$ 
 Then the rules of $\wtrct {} \internact {}$ allows us to show that there exist $k_1$ and $k_2$ with $k_1 + k_2 > 0$, and such that:
 $$\wtrct{\distrone}{ \internact^{k_1}}{\distrone' = \sum_{j\neq i}p_j\cdot{(\ctxone_j, \distrone_j)} +p_i \cdot {(\subst{\termtwo}{\varone}{\hole}, \distrone_i)}}$$
  and:
   $$\wtrct{\distrtwo}{ \internact^{k_2}}{\distrtwo' = \sum_{j\neq i}p_j\cdot{(\ctxone_j, \distrtwo_j)} +p_i \cdot {(\subst{\termtwo}{\varone}{\hole}, \distrtwo_i)}}$$
  , and we can easily see that $\mleqpar{\distrone'}{\distrtwo'} \epsone$.
 \end{itemize}
 \end{itemize}
 \end{itemize}
\end{proof}
\condskip
 \begin{lemma}\label{lemmehuitauxdeux}
Let be $\epsone > 0$.
  If $\distrone$ and $\distrtwo$ are two value distributions over $\ensct$ (and consequently, in normal form for $\wtr \cdot \internact \cdot$) with $\mleqpar \distrone \distrtwo  \epsone$, then for every $\valone$, there exists $\distrthree$, $\distrfour$ with $ \mleqpar {\distrthree}{\distrfour} \epsone$ such that $\wtrct \distrone {\app \valone} {\distrthree}$, and $\wtrct \distrtwo {\app \valone}{\distrfour}$. 
  \end{lemma}
\begin{proof}
By hypothesis, we know that $\mleqpar \distrone \distrtwo  \epsone$, and so we can write $\distrone$ and $\distrtwo$ as:
\begin{align*}
\distrone &= \sum_{i}p_i \cdot {(\ctxone_i, \distrone_i)} 
&& \text{ and }\distrtwo = \sum_{i}p_i \cdot {(\ctxone_i, \distrtwo_i)}\\
\text{ and } \forall i ,\, & \appl \metrtr {\distrone_i}{\distrtwo_i} \leq \epsone &&
\end{align*}
When $\distrone_i$  is a term distribution in normal form (i.e with value or non-reducible terms), we note 
\begin{align*}
\distrone'_i &= \sum_{\abstr \varone \termone}\distrone(\abstr \varone \termone)\cdot \dirac{\subst \termone \varone \valone}\\
  \text{ and }  \distrtwo'_i &= \sum_{\abstr \varone \termone}\distrtwo(\abstr \varone \termone)\cdot \dirac{\subst \termone \varone \valone}
\end{align*}
And we have $$\wtrct {\distrone}{\app \valone}{\distrthree = \sum_{i \mid \ctxone_i = \hole}} p_i \cdot (\hole, \distrone'_i) + \sum_{i \mid \ctxone_i = \abstr \varone \ctxtwo_i} p_i \cdot (\subst{\ctxtwo_i}{\varone}{\valone}, \distrone_i)$$,
and similarly: 
$$\wtrct {\distrtwo}{\app \valone}{\distrfour = \sum_{i \mid \ctxone_i = \hole}} p_i \cdot (\hole, \distrtwo'_i) + \sum_{i \mid \ctxone_i = \abstr \varone \ctxtwo_i} p_i \cdot (\subst{\ctxtwo_i}{\varone}{\valone}, \distrtwo_i)$$,
and we can see that $\appl \metrtr \distrthree \distrfour \leq \epsone$.
\end{proof}
\condskip
We can now use these two auxiliary lemma in order to prove Lemma \ref{lemmehuit}.
The proof is by induction on the length of $\traceone$:
\begin{itemize}
\item: if $\traceone = \emptytr$, then we have that $\distrthree = \normal \distrone$, $\distrfour = \normal \distrtwo$ and we have that $\mleqpar \distrthree \distrfour \epsone$ by Lemma \ref{lemmehuitauxmetr}.
\item if $\traceone = \concat{\app \valone}{\tracetwo}$. 
 Let be $\distrfive$, $\distrsix$ such that $\wtrct {\normal \distrone} {\app \valone}{\distrfive}$ and $\wtrct {\normal \distrtwo}{\app \valone}{\distrfive}$. We have (since $\wtrc{}{\traceone}{}$ is confluent ):

$$\wtrcct{\distrone}{\emptytr}{\wtrct{\normal \distrone}{\app \valone}{\wtrc\distrfive\tracetwo\distrthree}}.$$
and 
$$\wtrcct{\distrtwo}{\emptytr}{\wtrct{\normal \distrtwo}{\app \valone}{\wtrc\distrsix\tracetwo\distrfour}}.$$
Then by Lemma \ref{lemmehuitauxmetr} we have:
$\mleqpar{\normal \distrone}{\normal\distrtwo} \epsone$.
 Now we can apply Lemma \ref{lemmehuitauxdeux}, and we obtain that $\mleqpar {\distrfive}{\distrsix}{\epsone}$.
And now we apply the induction hypothesis to  $\tracetwo$, and we obtain that $\mleqpar \distrthree \distrfour \epsone$.
\end{itemize}
\end{proof}
\condskip

\subsubsection{Proof of Theorem \ref{theo:tracenonexpansive}.}

We can now see that Theorem \ref{theo:tracenonexpansive} is a direct
consequence of Lemma \ref{lemmehuit}. Indeed, let $\termone$ and
$\termtwo$ be two programs at distance at most $\epsone$ for the
trace metric, and let $\distrone$ and $\distrtwo$ be such that
$\wtrcct {\dirac{\ctxone, \termone}}{\traceone}{\distrone}$, and
$\wtrcct {\dirac{\ctxone, \termtwo}}{\traceone}{\distrtwo}$. Then, as
we have already observed, $ {\dirac{\ctxone, \termone}}$ and
${\dirac{\ctxone, \termtwo}}$ are $\epsone$-related. By Lemma
\ref{lemmehuit}, we can deduce that $\distrone$ and $\distrtwo$ are
$\epsone$-related. And it is easy to see that it implies that
$\abs{\sumdistr \distrone - \sumdistr \distrtwo} \leq \epsone$.

%%%%%%%%%%%%%%%%%%%%%%%%%%%%%%%%%%%%%%%%%
\subsection{Adding Pairs to the Calculus}\label{sect:pairs}
%%%%%%%%%%%%%%%%%%%%%%%%%%%%%%%%%%%%%%%%%
The trace distance and the results we have just presented about
it can be extended to an affine $\lambda$-calculus \emph{with pairs},
namely a calculus whose language of terms also includes the
following two constructs:
$$
\termone\bnf\pair{\termone}{\termtwo}\midd\letin{\varone}{\vartwo}{\termone}{\termtwo}.
$$ 
We assume that terms are typed in any linear type system guaranteeing
the absence of deadlocks (e.g., simple recursive types), and
 we add the following rules to the big-step semantics:
$$\AxiomC{} 
\UnaryInfC{$\bssp {\pair \termone \termtwo}{\dirac{\pair \termone \termtwo}} $}
\DisplayProof
$$
$$\AxiomC{$\bssp \termone \distrone $}
  \AxiomC{
$ (\bssp \termthree {\distrthree_{\termthree}},\, \bssp \termfour {\distrfour_\termfour})_{\pair \termthree \termfour\in \supp \distrone}  $}
\AxiomC{$\bssp {\subst{\subst \termtwo \varone \valone} \vartwo \valtwo} \distrtwo_{\valone, \valtwo} $}
  \TrinaryInfC{$\bssp {\letin \varone \vartwo \termone \termtwo}{\sum \distrone(\pair \termthree \termfour)\cdot {\distrthree_\termthree}(\valone) \cdot {\distrfour_\termfour(\valtwo) \cdot \distrtwo_{\valone, \valtwo}} }$} 
  \DisplayProof$$ 

We would now like to extend the definition of a trace to pairs
accordingly: which action should we perform on a term in the form
$\pair{\termone}{\termtwo}$? The na\"ive solution would be to add
projections to the trace language: $\traceone \bnf \concat {\proj
  1}{\traceone} \midd \concat{\proj 2}{\traceone}$, with trace
interpretation extended in the expected way:
\begin{align*}
\probtr {\pair \termone \termtwo}{\concat {\proj 1} \tracetwo} &= \probtr \termone \tracetwo \\
 \probtr {\pair \termone \termtwo}{\concat {\proj 2} \tracetwo} &= \probtr \termtwo \tracetwo 
\end{align*} 
However, this way the trace distance would not coincide 
with the context distance, anymore. Indeed, let us consider 
the following example:
\condskip
\begin{example}\label{expair}
We are going to compare the following terms:
$$
\termone\defi \pair {\abstr \varthree (\psum \identity \diver)}{\abstr \varthree (\psum \identity \diver)};\qquad
\termtwo\defi \pair {\abstr \varthree\identity} {\abstr \varthree\identity}.
$$
These two terms are at context distance at least $\frac 3 4$, since
we can consider the context $\ctxone \defi \letin \varone \vartwo
\hole { {\left(\varone \identity\right)} {\left(\vartwo
    \identity\right)}}$, and we can see that $\sumdistr{\sem {\fillc
    \ctxone \termone}} = \frac 1 4$, while $\sumdistr{\sem {\fillc
    \ctxone \termtwo}} = 1$. But we cannot find any trace that separates
them more than $\frac 1 2$. The interesting case is when $\traceone = \concat
{\proj i} \tracetwo$.  But then:
\begin{align*}
\abs{\probtr \termone \traceone - \probtr \termtwo \traceone}  & 
= \abs{\probtr {\abstr \varthree (\psum \diver \identity)} \tracetwo - \probtr {\abstr \varthree \identity} \tracetwo} \\
& \leq \appl \metrtr {\abstr\varthree (\psum \diver \identity)}{\abstr \varthree \identity}.   
\end{align*}
And it is easy to see that in the calculus with pairs we still have $\appl \metrtr {\abstr
  \varthree (\psum \diver \identity)}{\abstr \varthree \identity} = \frac
1 2 $.
\end{example}
\condskip
The reason why we cannot recover the context distance by way of
projections is that the $\texttt{let}$ construct above allows us to
access \emph{both} components of a pair, and the distances each of
them induce can \emph{add up}. A way out consists in extending the
trace language to pairs really following linearity, and considering a
new action in the form $\tenseur{\termone}$ with the following
extension of trace interpretation:

{\footnotesize
$$%\begin{align*}
\probtr {\pair \termone \termtwo}{\concat {\tenseur \termthree} \tracetwo} = \sum_{\valone, \valtwo} \sem \termone( \valone) \cdot \sem \termtwo (\valtwo) \cdot \probtr { {\subst \termthree {\varone, \vartwo} {\valone, \valtwo}} }\tracetwo 
$$}%\end{align*}

Please observe that we could in fact express the pairs in the original
language~\cite{Asperti2002TOCL}: let us consider $\embone: \calcwpair\rightarrow \terms$
defined by
\begin{align*}
\embone(\pair \termone \termtwo) & \defi \abstr \varone \varone \embone(\termone) \embone(\termtwo)  \\
\embone(\letin \varone \vartwo \termone \termtwo ) & \defi \embone(\termone)\, (\abstr \varone (\abstr \vartwo \embone(\termtwo)))\\
\embone (\abstr \varone \termone) & \defi \abstr \varone \embone( \termone)  \cdots 
\end{align*}
Moreover, we could see that every trace for the language $\calcwpair$ can be seen as a trace in the original language:
We can extend  $\embone: \words(\calcwpair) \rightarrow \words$, by:
\begin{align*}
\embone(\emptytr) &= \emptytr \\
\embone(\concat {\app \valone} \traceone) &= \concat {\app\valone} {\embone (\traceone)}\\
\embone(\concat{\tenseur \termone} \traceone) & = \concat{\abstr \varone \abstr \vartwo \termone}{\embone (\traceone)}
\end{align*}
and we have for every term $\termone \in \calcwpair$, and for every trace $\traceone \in \words(\calcwpair)$, 
$$\probtr \termone \traceone = \probtr {\embone (\termone)}{\embone(\traceone)} $$

This way of handling pairs allows the trace distance and the context
distance to coincide, again. However, the trace distance loses its
grip with respect to the context distance. Consider, for instance, the terms
$\termone$ and $\termtwo$ from Example \ref{expair}. Showing an
upper bound on the distance between $\termone$ and $\termtwo$ is the
same thing as showing an upper bound on $\appl \metrtr {\subst
  \termthree {\varone, \vartwo}{\abstr \varthree(\psum \diver
    \identity), \abstr \varthree (\psum \diver \identity)}}{\subst
  \termthree {\varone, \vartwo}{\abstr \varthree\identity, \abstr
    \varthree\identity}} $ \emph{for all terms} $\termthree$ such that $\wfj
{\varone, \vartwo} \termthree$, which is in fact not far away from what we
should show if we were considering the context distance directly.
%%%%%%%%%%%%%%%%%%%%%%%%%%%%%%
\section{The Bisimulation Distance}\label{sect:bisimulationdistance}
%%%%%%%%%%%%%%%%%%%%%%%%%%%%%
As we realised in the last section, the trace metric can be a
way to alleviate the burden of evaluating the context distance
between terms but, in particular in presence of pairs, its
usefulness can be limited. In this section, we will look at
another way to define the distance between programs which is
genuinely coinductive, and based on the Kantorovich metric
for distributions.
%%%%%%%%%%%%%%%%%%%%%%%
\subsection{Definition}
%%%%%%%%%%%%%%%%%%%%%%%
A labelled Markov chain (LMC) is a triple $\markovone = (\setone,
\labelsone, \probmatrone)$, where $\setone$ is a countable set of
states, $\labelsone$ is a countable set of labels, and $\probmatrone$
is a transition probability matrix, that is a function: $\probmatrone:
\setone \times \labelsone \rightarrow \distrs \setone$. Moreover, if
the image of $\probmatrone$ only consists of distributions with
\emph{finite} support, we call $\markovone$ an \emph{image-finite}
LMC. We are now going to define, in a similar way to
\cite{DesharnaisLICS02} (but in absence of non-determinism), the
metric analog to bisimulation.
The idea is to define a metric on the set $\setone$ of states of the
LMC as the greatest fixed point of some monotone operator on
metrics. Please recall that $(\metrs
\setone,\leqmetr{}{})$ is a complete lattice, and so any monotone
operator has indeed a greatest fixed point. 
%%%%%%%%%%%%%%%%%%%%%%%%%%%%%%%%%%%%%%%%%%%%%%%%%
\subsubsection*{Lifting Metrics to Distributions}
%%%%%%%%%%%%%%%%%%%%%%%%%%%%%%%%%%%%%%%%%%%%%%%%%
We are going to define a way to turn any premetric over a set
$\setone$ into a metric over finite distribution over $\setone$.
\condskip
\begin{definition}
Let $\metrone$ be a premetric on a set $\setone$. We define the
\emph{lifting of $\metrone$} as the metric on the set of finite
distributions over $\setone$ defined by: for every $\distrone$,
$\distrtwo$ finite distributions over $\setone$,
$\appl{\metrone}{\distrone}{\distrtwo}$ is the optimum solution to the
following linear program:

{\footnotesize
$$
\min\sum_{i, j}\coeffdtwo_{i, j} \cdot \appl{\metrone}{\stateone_i}{\stateone_j}  + \sum_{i}\coeffdxtwo_i + \sum_{j} \coeffdytwo_j 
$$
\vspace{-10pt}
\begin{align*}
\mbox{subject to}\qquad&\mbox{$\sum_{i}\coeffdtwo_{i, j} + \coeffdytwo_j= \distrtwo({\stateone_j});$}\\
&\mbox{$\sum_{j}\coeffdtwo_{i, j} + \coeffdxtwo_i= \distrone({\stateone_i});$}\\
&\forall i, j, \coeffdtwo_{i,j}, \, \coeffdytwo_j, \, \coeffdxtwo_i \geq 0.
\end{align*}}
\end{definition}
\condskip
Please observe that this linear program has an optimal solution. We
can make use of the notion of duality from linear programming, and
obtain an alternative characterisation of lifting:
\condskip
\begin{theorem}\label{theo:duality}
Let $\metrone$ be a premetric on $\setone$ and Let $\distrone$,
$\distrtwo$ be finite distributions over $\setone$. Then:

{\footnotesize
$$
\appl{\metrone}{\distrone}{\distrtwo}=\max \sum_{\stateone}\coeffone_\stateone \distrone(\stateone)+ \coefftwo_\stateone \distrtwo(\stateone) 
$$
\vspace{-10pt}
\begin{align*}
\mbox{subject to}\qquad&\forall \stateone \in \setone,\, \coeffone_\stateone \leq 1;\\
 &\forall \stateone\in \setone, \, \coefftwo_\stateone \leq 1;\\
 &\forall \stateone,\, \statetwo \in \setone, \, \coeffone_\stateone + \coefftwo_\statetwo \leq \appl{\metrone}{\stateone}{\statetwo}.
\end{align*}}
\end{theorem}
\begin{proof}
By strong duality theorem in linear programming.
\end{proof}
\condskip

We would like to have the lifting of a metric $\metrone$ behaving
coherently with $\metrone$ itself. If we know the lifting of
$\metrone$, we should first of all be able to recover $\metrone$
 by considering Dirac distributions:
\condskip
\begin{lemma}
\label{distdirac}
Let $\metrone$ be a premetric on $\setone$, and let $\stateone,\statetwo \in \setone$. Then $\appl{\metrone}{\dirac \stateone}{\dirac \statetwo} = \appl{\metrone}{\stateone}{\statetwo}$. 
\end{lemma}
\begin{proof}
Let be $\stateone, \statetwo \in \setone$, and 
let be $\distrone = \dirac \stateone$, $\distrtwo = \dirac \statetwo$. Then we can see that:
\begin{align*}
\appl{\metrone}{\distrone}{\distrtwo} &=\max\left\{
\begin{array}{l}
 \sum_{\statethree}\coeffone_\statethree \cdot\distrone(\statethree)+ \coefftwo_\statethree\cdot \distrtwo(\statethree) \\ \text{ with } \forall \statethree \in \setone,  \coeffone_\statethree \leq 1 \wedge \coefftwo_\statethree \leq 1 \\ \text{ and } \forall \statethree_1, \statethree_2 \in \setone, \,  \coeffone_{\statethree_1} + \coefftwo_{\statethree_2} \leq \appl{\metrone}{\statethree_1}{\statethree_2}
\end{array}
\right\}\\
& = \max \left\{
\begin{array}{l} 
  \coeffone_\stateone + \coefftwo_{\statetwo} \\ 
\text{with }  \coeffone_\stateone \leq 1 \wedge \coefftwo_\statetwo \leq 1
\\ \text{ and } \coeffone_{\stateone} + \coefftwo_{\statetwo} \leq \appl{\metrone}{\stateone}{\statetwo}  
\end{array}
\right\}\\
& = \appl \metrone \stateone \statetwo
\end{align*}
\end{proof}
\condskip
If a premetric on states verifies the triangular inequality, its
lifting verifies the triangular inequality too, which is a consequence
of the following lemma:
\condskip
\begin{lemma}\label{lifting}
Let $\metrone$, $\metrtwo$, $\metrthree$ be three premetrics on
$\setone$, such that $\forall \stateone, \statetwo, \statethree \in
\setone,\, \appl\metrone\stateone\statetwo \leq
\appl\metrtwo\stateone\statethree +
\appl\metrthree\statethree\statetwo $.  Let $\distrone$,
$\distrtwo$, $\distrthree$ be finite distributions over $\setone$. Then $
\appl{\metrone}{\distrone}{\distrthree} \leq
\appl{\metrtwo}{\distrone}{\distrtwo} +
\appl{\metrthree}{\distrtwo}{\distrthree}$.
\end{lemma} 
\begin{proof}
Let be $\distrone, \distrtwo, \distrthree$ finite distributions over $\setone$.
We're going to use the minimum-based definition of lifting:
Let be $\epsone$, $\epstwo$ such that:
$\appl{\metrtwo}{\distrone}{\distrtwo} = \epsone$ and
$\appl{\metrthree}{\distrtwo}{\distrthree} = \epstwo$.
By assumption, there is a finite number of states which appear in the union of the support of every considered distributions. We numerate these states between $1$ and $n$.\\
Then let be $(\coeffdone_{i,j})_{1 \leq i ,j \leq n}$, $(\coeffdxone_i)_{1 \leq i \leq n}$, $(\coeffdyone_j)_{1 \leq j \leq n}$ the coefficients for which the minimum of the optimisation problem associated with $\appl{\metrtwo}{\distrone}{\distrtwo}$ is reached. They verify the following equations:
$$\begin{cases}
\epsone = \sum\nolimits_{i, j}\coeffdone_{i, j} \cdot \appl{\metrtwo}{\stateone_i}{\stateone_j}  + \sum\nolimits_{i}\coeffdxone_i + \sum\nolimits_{j} \coeffdyone_j \\
\forall  j ,\, \sum\nolimits_{i}\coeffdone_{i, j} + \coeffdyone_j = \distrtwo({\stateone_j})\\
\forall  i, \, \sum\nolimits_{j}\coeffdone_{i, j} + \coeffdxone_i = \distrone({\stateone_i})\\
\forall  i, j \,:  \, \coeffdone_{i,j}, \coeffdxone_i, \coeffdyone_j  \geq 0
\end{cases}$$

Similarly, let be $(\coeffdtwo_{i,j})_{1 \leq i ,j \leq n}$, $(\coeffdxtwo_i)_{1 \leq i \leq n}$, $(\coeffdytwo_j)_{1 \leq j \leq n}$ the coefficients that reach the minimum for the optimisation problem associated to $\appl{\metrthree}{\distrtwo}{\distrthree}$. They verify the following equations: 
$$
\begin{cases}
&\epstwo = \sum\nolimits_{j, k}\coeffdtwo_{j, k} \cdot \appl{\metrthree}{\stateone_j}{\stateone_k}  + \sum\nolimits_{j}\coeffdxtwo_j + \sum\nolimits_{k} \coeffdytwo_k \\
&\forall  k , \, \sum\nolimits_{j}\coeffdtwo_{j, k} + \coeffdytwo_k = \distrthree({\stateone_k}) \\
&\forall  j , \, \sum\nolimits_{k}\coeffdtwo_{j, k} + \coeffdxtwo_j = \distrtwo({\stateone_j})\\
&\forall  j, k \,: \, \coeffdtwo_{j,k}, \coeffdxtwo_j, \coeffdytwo_k  \geq 0
\end{cases}
$$

We want to show that $\appl{\metrone}{\distrone}{\distrthree} \leq \epsone + \epstwo$. In order to do that, we would like to have coefficients 
 $\coeffdthree_{i,k}$, $\coeffdxthree_i$, $\coeffdythree_k$ which verifies the constraints of the optimisation problem associated with $\appl \metrone \distrone \distrthree$, and such that the objective function is bounded by $\epsone + \epstwo$. That is, we would like to have:
$$\left\{\begin{minipage}{0.45 \textwidth} 
\begin{align}
 &\sum\nolimits_{i}\coeffdthree_{i, k} + \coeffdythree_k = \distrthree({\stateone_k}) \label{eq:eqn1}\\
&\sum\nolimits_{k}\coeffdthree_{i, k} + \coeffdxthree_i = \distrone({\stateone_i}) \label{eq:eqn2}\\
&\sum\nolimits_{i, k}\coeffdthree_{i, k} \cdot \appl{\metrone}{\stateone_i}{\stateone_k}  + \sum_{i}\coeffdxthree_i + \sum_{k} \coeffdythree_k  \leq \epsone + \epstwo \label{eq:eqn3}
\end{align}
\end{minipage}
\right.
$$
In order to achieve that, we define the $\coeffdthree_{i,k}$, $\coeffdxthree_i$, $\coeffdythree_k$ on the following way: 
\begin{align*}
\coeffdthree_{i,k} &=  \sum_{j}\frac{\coeffdone_{i,j}\cdot\coeffdtwo{j,k}}{\distrtwo(\stateone_j)}\\
\coeffdxthree_i & = \coeffdxone_i + \sum_{j}\frac{\coeffdone_{i,j}\cdot\coeffdxtwo{j}}{\distrtwo(\stateone_j)}\\
\coeffdythree_k & = \coeffdytwo_k + \sum_{j}\frac{\coeffdtwo_{j,k}\cdot\coeffdytwo{j}}{\distrtwo(\stateone_j)}
\end{align*}
where we have adopted the following notation:  if $\distrtwo(\stateone_j) = 0$, then $\coeffdtwo_{j,k} = 0 = \coeffdone_{i,j}$, and then the meaning of $\frac{\coeffdone_{i,j}}{\distrtwo(\stateone_j)}$ is 0.\\
Now we are going to show that this choice of coefficients gives us what we wanted to have. We first verify that equation \eqref{eq:eqn2} holds. Indeed, we have that:
\begin{align*}
\sum_{k}\coeffdthree_{i,k} + \coeffdxthree_i & =
\sum_k\sum_{j}\frac{\coeffdone_{i,j}\cdot\coeffdtwo{j,k}}{\distrtwo(\stateone_j)} + (\coeffdxone_i + \sum_{j}\frac{\coeffdone_{i,j}\cdot\coeffdxtwo_{j}}{\distrtwo(\stateone_j)})\\
& = \coeffdxone_i + \sum_j\coeffdone_{i,j}\frac{\sum_k
\coeffdtwo_{j,k} + \coeffdxtwo_{j}}{\distrtwo(\stateone_j)}\\
& = \coeffdxone_i + \sum_j\coeffdone_{i,j}\frac{\distrtwo(\stateone_j)}{\distrtwo(\stateone_j)}
 = \distrone(\stateone_i)
\end{align*}

We can verify in a very similar way that equation \eqref{eq:eqn1} holds, that is: $\sum_{i}\coeffdthree_{i,k} + \coeffdythree_k  = \distrthree(\stateone_k)$.
We are now going to verify that equation $\eqref{eq:eqn3}$ holds. Indeed, we have that:
\begin{align*}
\sum_{i,k}& \coeffdthree_{i,k} \appl\metrone{\stateone_i} {\stateone_k} + \sum_i\coeffdxthree_i  + \sum_k \coeffdythree_k \\
& = \left[
\begin{array}{l}
 \sum_{i,k}\sum_{j}\frac{\coeffdone_{i,j}\cdot\coeffdtwo{j,k}}{\distrtwo(\stateone_j)} \cdot \appl\metrone{\stateone_i} {\stateone_k}\\ +
\sum_i (\coeffdxone_i + \sum_{j}\frac{\coeffdone_{i,j}\cdot\coeffdxtwo{j}}{\distrtwo(\stateone_j)}) \\+  \sum_k{(\coeffdytwo_k + \sum_{j}\frac{\coeffdtwo_{j,k}\cdot\coeffdytwo{j}}{\distrtwo(\stateone_j)})}
\end{array}\right.\\
& \leq \left[
\begin{array}{l}
\sum_{i,k}\sum_{j}\left(\frac{\coeffdone_{i,j}\cdot\coeffdtwo{j,k}}{\distrtwo(\stateone_j)} \cdot \appl\metrone{\stateone_i} {\stateone_k}\right)\\ + \sum_i(\coeffdxone_i + \sum_{j}\frac{\coeffdone_{i,j}\cdot\coeffdxtwo{j}}{\distrtwo(\stateone_j)})\\ + \sum_k(\coeffdytwo_k + \sum_{j}\frac{\coeffdtwo_{j,k}\cdot\coeffdytwo{j}}{\distrtwo(\stateone_j)})
\end{array}
\right.\\
& \leq \left[ 
\begin{array}{l}
\sum_{i,k}\sum_{j}\left(\frac{\coeffdone_{i,j}\cdot\coeffdtwo{j,k}}{\distrtwo(\stateone_j)}\cdot (\appl\metrtwo{\stateone_i}{\stateone_j}+ \appl \metrthree{\stateone_i} {\stateone_k})\right) \\
+ \sum_i(\coeffdxone_i + \sum_{j}\frac{\coeffdone_{i,j}\cdot\coeffdxtwo{j}}{\distrtwo(\stateone_j)}) \\ 
+ \sum_k( \coeffdytwo_k + \sum_{j}\frac{\coeffdtwo_{j,k}\cdot\coeffdytwo{j}}{\distrtwo(\stateone_j)})
\end{array}\right. \\
& \leq 
\left[ \begin{array}{l}
\sum_{i,j}\coeffdone_{i,j}\cdot\appl{\metrtwo}{\stateone_i}{\stateone_j}\left( \frac{\sum_k\coeffdtwo{j,k}}{\distrtwo(\stateone_j)}
 \right)\\ +
  \sum_{j,k}\coeffdtwo{j,k}\cdot\appl{\metrthree}{\stateone_j}{\stateone_k}\left( \frac{\sum_i\coeffdone_{i,j}}{\distrtwo(\stateone_j)}
 \right) \\
  + \sum_i \coeffdxone_i  + \sum_{j}\coeffdxtwo_j \left(  \frac{\sum_i\coeffdone_{i,j}}{\distrtwo(\stateone_j)}\right) \\ + \sum_k \coeffdytwo_k + \sum_j\coeffdytwo{j} \cdot\left(\frac{\sum_k\coeffdtwo_{j,k}}{\distrtwo(\stateone_j)}\right)
\end{array}\right.\\
  & \leq \left[
\begin{array}{l}
\sum_{i,j}\coeffdone_{i,j}\cdot\appl{\metrtwo}{\stateone_i}{\stateone_j}\\+
  \sum_{j,k}\coeffdtwo{j,k}\cdot\appl{\metrthree}{\stateone_j}{\stateone_k} \\
  + \sum_i \coeffdxone_i  + \sum_{j}\coeffdxtwo_j  + \sum_k \coeffdytwo_k + \sum_j\coeffdytwo{j}  
\end{array}\right.
\\& \leq \epsone + \epstwo
\end{align*}

\end{proof}
\condskip
%%%%%%%%%%%%%%%%%%%%%%%%%%%%%%%%%%%%%
\subsubsection*{Metrics as Fixpoints}
%%%%%%%%%%%%%%%%%%%%%%%%%%%%%%%%%%%%%
In a non-probabilistic setting, a relation $\relone$ is a
bisimulation if every pair of states $\stateone,\statetwo$ such that
$\relate \relone\stateone \statetwo$ can do the same actions and
end up into states which are still bisimilar.
More precisely,for every
action $\actone$, and for every state $\statethree$ such that $\doact
\stateone \actone \statethree$, there exists $\statefour$ such
that $\doact \statetwo \actone \statefour$, and $\relate \relone
\statethree \statefour$.

In order to obtain a quantitative counterpart
of the scheme above, we define an operator $\functionnal$ on the set
of metrics over the states of a LMC: intuitively, given
a metric $\metrone$, we define a new metric $\functionnal (\metrone)$
which corresponds to the distance obtained by first doing a step of
the transition relation, and then applying the lifting of $\metrone$
to the resulting distributions. 
More precisely, let be two states
$\stateone$ and $\statetwo$: $\appl {\functionnal
  (\metrone)}{\stateone}{\statetwo}$ is computed in the following way:
for every action $\actone$, we consider the distance (with respect to
$\metrone$) between the behaviour obtained from $\stateone$ after
doing the action $\actone$, and the behaviour obtained from
$\statetwo$ after doing the same action $\actone$, and then we take
the maximum over all action $\actone$ of those quantity.
\condskip
\begin{definition}
Let $\markovone = (\setone, \labelsone, \probmatrone)$ be an image-finite LMC.
We define an operator $\functionnal$ on $\metrs \setone$ as 
\begin{align*}\appl{\functionnal(\metrone)}{\stateone}{ \statetwo} = & \sup \{\appl \metrone  {\probmatrone(\stateone)(\actone)}{ \probmatrone(\statetwo)(\actone)} \mid \actone \in \labelsone \}.
\end{align*}
\end{definition}
\condskip
\begin{theorem}
For any image-finite LMC $\markovone$, $\functionnal$ has a
maximum fixpoint. We call it \emph{the bisimulation metric} for the
LMC $\markovone$, and we note it $\bisim_\markovone$
\end{theorem}

%%%%%%%%%%%%%%%%%%%%%%%%%%%%%%%%%%%%%%%%%%%%%%%%%%%%%%%%%%%%%%%%%%%%%%
\subsubsection*{Bisimulation Metric and the Affine $\lambda$-Calculus}
%%%%%%%%%%%%%%%%%%%%%%%%%%%%%%%%%%%%%%%%%%%%%%%%%%%%%%%%%%%%%%%%%%%%%%
We are now going to consider a specific LMC $\markovterm$,
which captures the interactive behaviour of our calculus.
\condskip
\begin{definition}
We define the LMC $\markovterm = (\setterm ,\labelsterm
,\probmatrterm )$ where:
\begin{varitemize}
\item 
  The set of states $\setterm$ is defined as follows:
  $$
  \setterm= \programs\uplus\valset,
  $$ 
  A value $\valone$ in the second component of $\setterm$ is
  distinguished from one in the first by using the notation
  $\hat\valone$.
\item 
  The set of labels $\labelsterm$ is taken to be
  $$
  \labelsterm = \{ \app \valone \mid \valone \text{a value} \} \bigcup \{\eval\}.
  $$
\item 
  The transition probability matrix $\probmatrterm$ is such that:
  for every  
  $\termone \in \programs$, and any value 
  $\valone\in\supp{\sem{\termone}}$, it holds that
  $\probmatrterm(\termone,\eval)(\dval \valone)=\sem{\termone}(\valone)$, and that for every term $\termone$ such that $\abstr \varone \termone \in \programs$, and $\valone \in \valset$, it holds that $\probmatrterm(\dval{\abstr \varone \termone},\app\valone)(\subst{\termone}{\valone}{\varone})=1$.
\end{varitemize}
\end{definition}
\condskip
The results we have proved previously in this section apply to
$\markovterm$. In particular, one can define the bisimulation metric
on $\markovterm$. The \emph{bisimulation distance} on programs, which
we indicate $\bisim$, is defined to be the restriction of
$\bisim_{\markovterm}$ to programs.
\condskip
\begin{definition}
We define a metric $\bisim$ on the set of closed terms, by: for every $\termone$, $\termtwo$,
$$ \appl \bisim \termone \termtwo = \appl {\bisim_\markovterm} \termone \termtwo$$ 
\end{definition}
\condskip
\begin{lemma}
For this particular LMC, we have that:
\begin{align*}
&\appl{\functionnal(\metrone)}{\dval{\abstr\varone\termone}}{\dval{\abstr\varone\termtwo}} = \sup\{\appl \metrone {\subst \termone \varone \valone} {\subst \termtwo \varone \valone} \mid {\valone\text{ a value }}\}\\
&\appl{\functionnal(\metrone)}{\termone}{ \termtwo} = \appl \metrone{\sem \termone}{\sem\termtwo}\}\\
&\appl{\functionnal(\metrone)}{\termone}{ \dval V} =0
\end{align*}
\end{lemma}
\condskip
We can see easily that $\bisim$ is an adequate metric.
\condskip
\begin{lemma}
$\bisim$ is an adequate metric on programs.
\end{lemma}
\begin{proof}
We have to show: for every $\termone$, $\termtwo$:
 $$\abs{\sumdistr {\sem \termone} - \sumdistr{\sem \termtwo}} \leq \appl \bisim \termone \termtwo$$
\end{proof}
\condskip 
But there is more, since the bisimulation metric is well-known to be a
lower bound on the trace distance: the bisimulation distance is a
\emph{sound} metric. In the next
section, we anyway show non-expansiveness for it, which is stronger.
%%%%%%%%%%%%%%%%%%%%%%%%%%%%%%
\subsection{Non-Expansiveness}
%%%%%%%%%%%%%%%%%%%%%%%%%%%%%%
Proving the non-expansiveness of $\bisim$ cannot be done directly,
by a plain induction on contexts. Our strategy towards the result
is the Howe's technique~\cite{Howe}, a way of proving
congruence of coinductively-defined equivalences which has been widely
used for deterministic and non-deterministic languages, and that we
here adapt to metrics.
 
The idea, then, is to start from $\bisim$, to construct another metric
$\howe \bisim$ on top of $\bisim$ (which turns out to be non-expansive
by construction), and to show that $\howe \bisim=\bisim$.  We first
need to transform our metric $\bisim$ on programs into a metric on
(potentially open) terms. Any metric $\metrone$ on programs can be
extended into a metric on open terms, which by abuse of notation we
continue to call $\metrone$ and which is defined as follows
\begin{align*}
\appl\metrone{\termone}{\termtwo} = \sup_{\valone_1,\ldots,\valone_n}\appl \metrone{&\subst\termone{\varone_1,\ldots,\varone_n}{\valone_1,\ldots,\valone_n}}{\\ &\subst\termtwo{\varone_1,\ldots,\varone_n}{\valone_1,\ldots,\valone_n}},
\end{align*}
where $\varone_1,\ldots,\varone_n$ are the variables
occurring free in either $\termone$ or $\termtwo$.
\condskip
\begin{definition}
Let be $\metrone$ a metric on terms. An \emph{Howe judgement} is a
element of the form $(\contone, (\termone, \termtwo), \epsone)$, where
$\contone$ is a typing context, $\termone$ and $\termtwo$ are two
terms, and $\epsone \in [0,1]$. We say that an Howe judgement is
\emph{valid}, and we note
\emph{$\judghowe{\contone}{\metrone}{\termone}{\termtwo}{\epsone}$} ,
if it can be derived by the rules of figure \ref{Howe}.
\end{definition}
\condskip
Please observe that, potentially, there are several different
$\epsone$ such that $\judghowe \contone \metrone\termone\termtwo
\epsone$. 
\begin{figure}[!h]
 $$\AxiomC{${\appl{\metrone}{\varone}{\termone}\leq \epsone} $}
\AxiomC{$\wfj{\varone, \contone}{\termone} $}
  \BinaryInfC{$\judghowe{\varone,\contone}\metrone\varone\termone\epsone$}
  \DisplayProof
 % \qquad
 % \AxiomC{$\varone \neq \vartwo$}
 % \UnaryInfC{$\judg{\appl{\howe\metrone}\varone\vartwo\leq 1}$}
%\DisplayProof
$$

$$ 
\AxiomC{$\judghowe \contone {\metrone}\termone\termfour\epsone$}
\AxiomC{$\judghowe \conttwo \metrone\termtwo\termfive \epsthree $}
\AxiomC{$\appl\metrone{\termfour\termfive}{\termthree} \leq \epstwo $}
\AxiomC{$\wfj{\contone, \conttwo}{\termthree} $}
  \QuaternaryInfC{$\judghowe {\contone, \conttwo} \metrone{\termone\termtwo}\termthree{\epsone +\epsthree + \epstwo}$}
  \DisplayProof$$
  
$$
\AxiomC{$\judghowe \contone \metrone \termone\termfour \epsone $}
\AxiomC{$\judghowe \contone \metrone\termtwo\termfive \epstwo $}
\AxiomC{$\appl\metrone{\psum\termfour\termfive}{\termthree} \leq \epsthree $}
\AxiomC{$\wfj\contone\termthree $}
  \QuaternaryInfC{$\judghowe\contone\metrone{\psum\termone\termtwo}\termthree{\frac{\epsone +\epstwo}{2} + \epsthree} $}
  \DisplayProof
$$ $$
\AxiomC{$\judghowe{\varone,\contone}\metrone\termone \termfour \epsone $}
\AxiomC{${\appl\metrone{\abstr\varone\termfour}{\termthree} \leq \epstwo} $}
\AxiomC{$\wfj \contone\termthree $}
  \TrinaryInfC{$\judghowe{\contone}{\metrone}{\abstr\varone\termone}\termthree{\epsone + \epstwo} $}
\DisplayProof
$$
\caption{Rule for Howe's constructor on metrics}\label{Howe}
\end{figure}

We are finally in a position to define the Howe's lifting of $\metrone$:
\condskip
\begin{definition}
Let be $\metrone$ a metric on terms.  We define a premetric $\howe
\metrone$ on terms by:
$$ 
\appl {\howe \metrone}\termone\termtwo = \inf\left(\{\epsone\mid
\exists \contone, \judghowe\contone\metrone\termone\termtwo \epsone\}
\bigcup \{1\}\right).
$$ 
\end{definition}
\condskip

The following lemma says that the optimum value of $\epsone$ can be reached with any typing context $\contone$ which contains the free variables of $\termone$ and $\termtwo$.
\begin{lemma}
\label{typehowe}
For every terms $\termone$, $\termtwo$, for every typing contexts $\contone$, and every real $\epsone$ such that ${\judghowe \contone \metrone \termone \termtwo  \epsone}$, we have that: $\freevar \termone \cup \freevar \termtwo \subseteq \contone$. Moreover, for any context $\conttwo$ such that $\{\epstwo \mid {\judghowe \conttwo\metrone \termone \termtwo \epstwo} \} \neq \emptyset$, then $\inf \{\epstwo \mid {\judghowe \conttwo \metrone \termone \termtwo \epsone} \} \leq \epsone$.
\end{lemma}

We can see that $\howe \bisim$ is a premetric on open terms. Please
observe that it is not necessarily a metric, since its construction
entails neither symmetry nor the triangular inequality.
\begin{lemma}\label{reflhowe}
If $\metrone$ is any premetric on closed terms, then 
$\howe\metrone$ is a premetric on (potentially open) terms.
\end{lemma}
\begin{lemma}
For every terms $\termone$, $\termtwo$:
 $$ \appl {\howe\metrone}{\termone}{\termtwo}\leq \appl \metrone \termone \termtwo$$
\end{lemma}
The construction of Howe's lifting allows us to have the two following properties: 
\begin{lemma}[Pseudo-Transitivity]\label{pseudotrans}
Let be $\metrone$ a metric on terms.
For every terms $\termone,\, \termtwo,\, \termthree $:
$$\appl{\howe\metrone}{\termone}{\termtwo}\leq \appl{\howe\metrone}{\termone}{\termthree} + \appl{\metrone}{\termthree}{\termtwo}$$
\end{lemma}
 \begin{proof}
Let be $\epsone$ such that $\judghowe \contone \metrone \termone \termthree \epsone$ is a valid judgement. It is enough to show  that  $$\judghowe \contone\metrone\termone\termtwo{\epsone + \appl \metrone \termthree \termtwo}$$ is a valid judgement.The proof is by induction on the rules of the construction of valid judgements.
\end{proof}

\begin{lemma}[Pseudo-substitutivity]\label{pseudosubst}
If $\metrone$ verifies that, for every terms $\termone, \termtwo$, for every values $\valone$: 
$\appl{\metrone} { \subst\termone\varone\valone}{\subst \termtwo\varone\valone} \leq \appl{\metrone}{\termone}{\termtwo}$.
Then for every terms $\termone, \termtwo$, for every values $\valone$, $\valtwo$: 
$$\appl{\howe\metrone} { \subst\termone\varone\valone}{\subst \termtwo\varone\valtwo} \leq \appl{\howe\metrone}{\termone}{\termtwo} + \appl{\howe\metrone}{\valone}{\valtwo}$$
\end{lemma}
Please observe that, the open extension of a metric on closed term verifies the hypothesis.

\begin{proof}
Let be $\epsone$ such that: $\judghowe \contone\metrone{\termone}{\termtwo} \epsone$.
The proof is by induction on the structure of the derivation of $\judghowe \contone \metrone{\termone}{\termtwo} \epsone$.
\begin{itemize}
\item If the derivation is:
\AxiomC{$\appl{\metrone}{\varone}{\termtwo}\leq \epsone $}
  \UnaryInfC{$\judghowe {\contone, \varone}\metrone\varone\termtwo\epsone$}
  \DisplayProof 
%(or the symmetric rule)
: Then   ${ \subst\termone\varone\valone} = \valone$. Then, since $\metrone$ is pseudo substitutive: $\appl{\metrone} { \valtwo}{\subst \termtwo\varone\valtwo} \leq \appl{\metrone}{\termone}{\termtwo}\leq \epsone$. Now by pseudo-transitivity of $\howe\metrone$, we have that: 
$\appl{\howe\metrone}\valone {\subst \termtwo\varone\valtwo} \leq \appl {\howe\metrone} \valone \valtwo + \appl{\metrone} { \valtwo}{\subst \termtwo\varone\valtwo} \leq \appl {\howe\metrone} \valone \valtwo + \epsone$.
\item If the derivation is:
$$
\AxiomC{$\judghowe\contone \metrone\termfive\termfour  \epsone $}
\AxiomC{$ \judghowe\conttwo \metrone{\termsix}{\termseven} \epsthree$}
\AxiomC{$\wfj {\contone, \conttwo}{\termthree}$}
\AxiomC{$\appl\metrone{\termfour\termseven}{\termthree} \leq \epstwo $}
  \QuaternaryInfC{$\judghowe {\contone, \conttwo}{\metrone}{\termfive\termsix}\termthree{ \epsone + \epstwo + \epsthree}$}
  \DisplayProof$$
We know that $x$ cannot appear  both in $\contone$ and in $\conttwo$.
Suppose for example that $x$ doesn't appear in $\conttwo$ . Then (by Lemma \ref{typehowe}) $x$ doesn't appear in $\freevar\termsix\cup\freevar \termseven$. Then:
We apply the induction hypothesis to: $\appl{\howe\metrone}\termfive\termfour \leq \epsone $. We have:
$$\appl{\howe\metrone}{\subst\termfive\varone\valone}{\subst\termfour\varone\valtwo} \leq {\appl{\howe\metrone}{\termfive}{\termfour}}+{\appl{\howe\metrone}{\valone}{\valtwo}} \leq \epsone + {\appl{\howe\metrone}{\valone}{\valtwo}} .$$
Moreover, since $\appl\metrone{\termfour\termseven}{\termthree} \leq \epstwo$, we have that (since $\metrone$ is value substitutive):
$$\appl\metrone{\subst{(\termfour\termseven)}{\varone}{\valtwo}}{\subst\termthree\varone\valtwo} \leq \epstwo$$ 
So now, we have that:
$$
\AxiomC{\small$\begin{array}{c}\judghowe{\contone \setminus \varone}{\metrone}{\subst\termfive\varone\valone}{\subst\termfour\varone\valtwo} { \epsone + \appl{\howe\metrone}{\valone}{\valtwo}}\\ \judghowe \conttwo\metrone{\termsix}{\termseven}\epsthree \end{array}$}
\AxiomC{\small$\appl\metrone{{\subst\termfour\varone\valtwo}{\termseven}}{\subst{\termthree}\varone\valtwo} \leq \epstwo $}
\AxiomC{\small$\wfj {\contone \setminus\varone, \conttwo}{{\subst{\termthree}\varone\valtwo}}$}
  \TrinaryInfC{$\judghowe{\contone \setminus \varone, \conttwo}\metrone{\subst{(\termfive\termsix)}{\varone}{\valone}}\termthree{\epsone + \epstwo + \epsthree + \appl{\howe\metrone}{\valone}{\valtwo}}$}
  \DisplayProof$$
  \item Other cases are similar.
\end{itemize}
\end{proof}
\condskip
The interest of this construction is that the metric $\howe \bisim$
is (more or less by construction) non-expansive:
\condskip
\begin{lemma}[Non-expansiveness of $\howe\bisim$] 
For every  context $\ctxone$ and for every terms $\termone$, $\termtwo$
it holds that
$\appl {\howe\bisim} {\fillc\ctxone\termone}{\fillc\ctxone\termtwo} \leq \appl {\howe\bisim} \termone \termtwo $.
\end{lemma}
\begin{proof}
The proof is by induction on the structure of the context $\ctxone$.
\end{proof}
\condskip
The goal now is to show that $\leqmetr{\howe \bisim} \bisim$. Since
$\bisim$ is the greatest fixed point of
$\functionnal$ for our LMC $\markovterm$, we are going to show that
$\howe\bisim$ can be extended into a metric on the states of
$\markovterm$, obtaining a fixed point for the operator
$\functionnal$.  First we extend $\howe\bisim$ to a premetric on
$\setterm$:
\condskip
\begin{definition}
We define the extension of $\howe \bisim$ to $\setterm$ (that we note
still $\howe \bisim$ by abuse of notation), by:
{\footnotesize \begin{align*}
\appl{\howe\bisim}{ \termone}{ \termtwo}&=\appl{\howe\bisim}{\termone}{\termtwo};\\
\appl{\howe\bisim}{\dval\valone}{\dval\valtwo}&=\appl{\howe\bisim}{\valone}{\valtwo};\\
\appl{\howe\bisim}{\termone}{ \dval\valtwo}&=1.
\end{align*}}
\end{definition}
\condskip
Since $\howe \bisim$ isn't guaranteed to be a metric, we are forced to further refine it, by adding rules corresponding to symmetry
and to the triangular inequality: we define $\closure {\howe \bisim}$
over $\setterm$ by the rules of Figure \ref{closure}. 
\begin{definition}
We define a valid $\closure{\howe\bisim}$-judgement $\judgclh \stateone \statetwo \epsone$, where $\stateone, \statetwo \in \setterm$, $\epsone \in [0,1]$, as the judgements which have a finite proof-tree by using the rules of Figure \ref{closure}.

We define $\closure {\howe \bisim}$ a metric over $\setterm$ by:

$${\appl {\closure {\howe \bisim}} \stateone \statetwo }= 
\inf { \{ \epsone \mid {\judgclh \stateone \statetwo \epsone }\}}$$  

\end{definition}

\begin{figure}[!h]
\begin{center}
\fbox{\footnotesize
\begin{minipage}{\figurewidth}
$$
\AxiomC{$\appl{\howe\bisim}{\stateone}{\statetwo}\leq\epsone$}
\UnaryInfC{$\judgclh{\stateone}{\statetwo} \epsone $}
\DisplayProof
\qquad
\AxiomC{$\judgclh{\stateone}{\statetwo}\epsone$}
\AxiomC{$\judgclh{\statetwo}{\stateone}\epstwo$}
\BinaryInfC{$\judgclh{\stateone}{\statetwo} { \min(\epsone, \epstwo)} $}
\DisplayProof$$
$$
\AxiomC{$\judgclh{\stateone}{\statetwo}\epsone$}
\AxiomC{$\judgclh{\statetwo}{\statethree}\epstwo$}
\BinaryInfC{$\judgclh{\stateone}{\statethree} {\epsone + \epstwo} $}
\DisplayProof
 $$
\end{minipage}}
\end{center}
\caption{Metric Closure of $\howe\bisim$}\label{closure}
\end{figure}
\begin{lemma}
$\closure {\howe\bisim}$ is a metric.
\end{lemma}

We can see easily that $\leqmetr {\bisim}{\leqmetr{\howe
    \bisim}{\closure {\howe \bisim}}}$ with respect to the preorder on
terms. We want to show that $\howe \bisim = \bisim$. In order to have
that, we will show that $\leqmetr{\closure {\howe
    \bisim}}{\bisim}$. That is a direct consequence of the following
theorem:
\condskip
\begin{theorem}
$\closure {\howe \bisim}$ is a pre-fixpoint of $\functionnal$.
\end{theorem}
\condskip
\begin{proof}
We need to show that $\leqmetr{\closure{\howe
    \bisim}}{\functionnal {(\closure{\howe \bisim})}}$. Please
remember that the preorder on metrics corresponds to the reverse of
the point-wise preorder for states. So if we read this inequality on
metrics as an inequality on the states of $\markovterm$, we see that
it is equivalent to: for every $\stateone, \statetwo \in \setterm$,
$\appl{\functionnal{(\closure {\howe \bisim}})} \stateone \statetwo
\leq \appl{\closure {\howe \bisim}}{\stateone}{\statetwo}$.  If we
unfold the definition of the operator $\functionnal$ on metrics, we
can see that it means that for every $\actone \in \labelsterm$, $\appl
{\closure {\howe
    \bisim}}{\probmatrterm(\stateone,\actone)}{\probmatrterm(\statetwo,\actone)}
\leq \appl {\closure {\howe \bisim}} \stateone \statetwo$.  Please
remember that there are two kinds of actions in our LMC: the
action $\eval$ of evaluating a program to obtain a value
distribution, and the action $\app \valone$, which corresponds to
passing the value $\valone$ to a distinguished value. If we consider
separately each of these actions, we see that the result we want to
have is equivalent to:
\begin{varitemize}
\item 
  Let be $\termone$, $\termtwo$ closed terms.  Then
  $\appl{\closure{\howe\bisim}}{\dval{\sem\termone}}{\dval{\sem\termtwo}}
  \leq \appl{\closure{\howe\bisim}}{\termone}{\termtwo}$
\item 
  Let be $\termone$, $\termtwo$ such that $\wfj \varone \termone$ and
  $\wfj \varone \termtwo$, and let $\valone$ be a value. Then it holds
  that: $
  \appl{\closure{\howe\bisim}}{\subst\termone\varone\valone}{\subst\termtwo\varone\valone}
  \leq \appl{\closure {\howe\bisim}}{\dval{\abstr \varone
      \termone}}{\dval{\abstr\varone\termtwo}}$
\end{varitemize}

We are first going to show these two result to the original premetric
on terms $\howe \bisim$, and we will extend them later to $\closure
{\howe \bisim}$.
%%%%%%%%%%%%%%%%%%%%%%%%%%%%%%%%%%%%%%%%%%%%%%%%
% Key-Lemma
%%%%%%%%%%%%%%%%%%%%%%%%%%%%%%%%%%%%%%%%%%%%%%%%%%
 \begin{lemma}[Key-Lemma]\label{KL}
Let be $\termone$ and $\termtwo$ two closed terms. Then: $$\appl{\howe\bisim}{\sem\termone}{\sem\termtwo} \leq \appl{\howe\bisim}{\termone}{\termtwo} $$
\end{lemma}

\begin{proof}
We show in fact that,  
for every $\epsone$ such that the judgement $\judghowe{} \bisim {\termone}{\termtwo} \epsone$ is a valid one, it holds that $\appl{\howe\bisim}{\sem\termone}{\sem\termtwo} \leq \epsone$.
We show that by induction on the structure of the derivation of: $\bssp \termone {\sem \termone}$
\begin{itemize}
\item If $\termone$ is a value: then $\termone = \abstr \varone \termfour$, and the derivation of $\bssp \termone {\sem \termone}$ is of the following form:
$$\AxiomC{}
\UnaryInfC{$\bssp {\abstr \varone \termfour}{\dirac{\abstr \varone \termfour}} $}
\DisplayProof $$
Then the proof tree allowing to certify the validity of  $\judghowe{} \bisim {\termone}{\termtwo} \epsone$ should be of the form:
$$
\AxiomC{$\judghowe {\varone}{\bisim}\termfive \termfour \epsthree $}
\AxiomC{$\appl\bisim{\abstr\varone\termfour}{\termtwo} \leq \epstwo $}
  \BinaryInfC{$\judghowe{}{\bisim}{\abstr\varone\termfive}\termtwo{\epsone = \epsthree + \epstwo} $}
\DisplayProof
$$
Since $\bisim$ is a fixpoint of $\functionnal$, we have that:
$\appl{\bisim}{\sem{\abstr\varone\termfour}}{\sem\termtwo} \leq\appl{\bisim}{\abstr\varone\termfour}{\termtwo} \leq \epstwo $.
And so:
\begin{align*}
\appl{\howe\bisim}{\sem{\abstr\varone\termfive}}{\sem\termtwo}& \leq {\appl{\howe\bisim}{\sem {\abstr\varone\termfive}}{\sem{\abstr\varone\termfour}}} + \appl{\bisim}{\sem{\abstr\varone\termfour}}{\sem\termtwo} && \text{ by lemma \ref{pseudotrans}}\\
&\leq {\appl{\howe\bisim}{\sem{\abstr\varone\termfive}}{\sem{\abstr\varone\termfour}}} + \delta
\end{align*}
Moreover, since $\abstr \varone \termfive$ and $\abstr \varone \termfour$ are values, we know that: $\sem {\abstr \varone \termfive} = \dirac{\abstr \varone \termfive}$, and $\sem{\abstr \varone \termfour} = \dirac{\abstr \varone \termfour}$. By Lemma \ref{distdirac}, we can see that:
${\appl{\howe\bisim}{\sem{\abstr\varone\termfive}}{\sem{\abstr\varone\termfour}}}  = {\appl{\howe\bisim}{{\abstr\varone\termfive}}{{\abstr\varone\termfour}}}$. It follows that: $\appl{\howe\bisim}{\sem{\abstr\varone\termfive}}{\sem\termtwo} \leq {\appl{\howe\bisim}{{\abstr\varone\termfive}}{{\abstr\varone\termfour}}}+ \delta$
and since we have the following proof tree, it allows us to conclude.
$$
\AxiomC{$\judghowe{\varone}{\bisim}\termfive \termfour \epsthree $}
\AxiomC{$\appl\bisim{\abstr\varone\termfour}{\abstr\varone\termfour} \leq 0 $}
  \BinaryInfC{$\judghowe{}\bisim{\abstr\varone\termfive}{\abstr\varone\termfour} \epsthree $}
\DisplayProof
$$

\item If $\termone = \termsix\termthree$.
Then the derivation of $\bssp \termone {\sem \termone}$ is the following:
 $$
  \AxiomC{$
    \begin{array}{c}
      \bssp{\termsix}{\sem\termsix}\qquad\bssp{\termthree}{\sem\termthree}\\
      \{\bssp{\subst{\termseven}{\varone}{\valone}}{\sem{\subst\termseven\varone\valone}}\}_{\abstr{\varone}{\termseven}\in\supp{\sem\termsix},\valone\in\supp{\sem\termthree}}
    \end{array}$}
  \UnaryInfC{$\bssp{\termone\termtwo}{\sum{\sem\termsix}(\abstr{\varone}{\termthree})\cdot{\sem\termthree}(\valone)\cdot\sem{\subst\termseven\varone\valone}}$}
  \DisplayProof
  $$

And the proof tree corresponding to the validity of $\judghowe{} {\bisim} \termone \termtwo  \epsone$ has the following form:
$$ 
\AxiomC{$\judghowe{}\bisim\termsix\termfour  \epsfour $}
\AxiomC{$\judghowe{}\bisim\termthree\termfive \epsthree $}
\AxiomC{$\appl\bisim{\termfour\termfive}{\termtwo} \leq \epstwo $}
  \TrinaryInfC{$\judghowe{}{\bisim}{\termsix\termthree}\termtwo{\epsone =  \epsfour +\epsthree + \epstwo}$}
  \DisplayProof$$
  We have: 
  \begin{align*}
  \appl{\howe\bisim}{\sem{\termsix\termthree}}{\sem\termtwo} & \leq \appl{\howe\bisim}{\sem{\termsix\termthree}}{\sem{\termfour\termfive}} + \appl{\bisim}{\sem{\termfour\termfive}}{\termtwo}\\
  &\leq \appl{\howe\bisim}{\sem{\termsix\termthree}}{\sem{\termfour\termfive}} + \epstwo \\
  \end{align*}
  So it is enough to show that: $\appl{\howe\bisim}{\sem{\termsix\termthree}}{\sem{\termfour\termfive}} \leq \epsone + \epsthree$.

So we have that:
\begin{align*}
\appl{\howe\bisim}{\sem{\termsix\termthree}}{\sem{\termfour\termfive}} &=\max\left\{ \sum_{\stateone}\coeffone_\stateone\cdot  {\sem{\termsix\termthree}}(\stateone)+ \coefftwo_\stateone{\sem{\termfour\termfive}}(\stateone) \mid  \coeffone_\stateone \leq 1, \coefftwo_\stateone \leq 1 \wedge \coeffone_\stateone + \coefftwo_\statetwo \leq \appl{\howe\bisim}{\stateone}{\statetwo}\right\}\\
&= \max\{ \sum_{\stateone}(\coeffone_\stateone \cdot {\sum_{\termseven}\sum_{\valone}{\sem{\termsix}(\abstr\varone\termseven)\cdot\sem{\termthree}(\valone)\cdot\sem{\subst\termseven \varone\valone}(\stateone)}}+ \\& \qquad\qquad\qquad
\coefftwo_\stateone\cdot{\sum_{\termeight}\sum_{\valtwo} 
{\sem{\termfour}(\abstr\varone\termeight)\cdot\sem{\termfive}(\valtwo)\cdot\sem{\subst\termeight \varone\valtwo}(\stateone)}}) \\ & \qquad \mid   \coeffone_\stateone \leq 1, \coefftwo_\stateone \leq 1 \wedge \coeffone_\stateone + \coefftwo_\statetwo \leq \appl{\howe\bisim}{\stateone}{\statetwo}\}
\end{align*}

We are now going to use the dual characterisation of the lifting of a metric to a distribution:
We know that: $\appl{\howe\bisim}\termsix\termfour \leq \epsone $.\\
So there exist $(\coeffdone_{\termseven, \termeight})_{{\abstr\termseven \in \supp{\sem\termsix}}, {\abstr\termeight \in \supp{\sem\termfour}}}$, and $({\coeffdxone_\termseven})_{\abstr\varone\termseven \in \supp{\sem\termsix}}$, and $(\coeffdyone_\termeight)_{\abstr\varone\termeight \in \supp{\sem\termfour}}$, such that:
\begin{align}
\label{eqn1}
\sum_{\termseven, \termeight}\coeffdone_{\termseven, \termeight} \cdot \appl{\howe\bisim}{\abstr\varone\termseven}{\abstr\varone\termeight} & + \sum_{\termseven}\coeffdxone_\termseven + \sum_{\termeight} \coeffdyone_\termeight = \appl{\howe\bisim}{\termsix}{\termfour}\\
\label{eqn2}
\sum_{\termseven}\coeffdone_{\termseven, \termeight} + \coeffdyone_\termeight &= \sem{\termfour}({\abstr\varone\termeight})\\
\label{eqn3}
\sum_{\termeight}\coeffdone_{\termseven, \termeight} + \coeffdxone_\termseven &= \sem{\termsix}({\abstr\varone\termseven})
\end{align}
Please observe that the equation (\ref{eqn1}) implies in particular that: $\sum_{\termseven, \termeight}\coeffdone_{\termseven, \termeight}\leq 1$. Similarly, the equations (\ref{eqn2}) and (\ref{eqn3}) implies that $\sum_{\termseven}\coeffdxone_{\termseven,}\leq 1$
and $\sum_{\termeight}\coeffdyone_{ \termeight}\leq 1$

and
\begin{align*}
\appl{\howe\bisim}{\sem{\termsix\termthree}}{\sem{\termfour\termfive}} 
&= \max\{ \sum_{\stateone}(\coeffone_\stateone \cdot \sum_{\termseven}\sum_{\valone}\left( \sum_{\termeight}\coeffdone_{\termseven, \termeight} + \coeffdxone_\termseven\right)\cdot\sem{\termthree}(\valone)\cdot\sem{\subst\termseven \varone\valone}(\stateone)\\
& \qquad\qquad + \coefftwo_\stateone \sum_{\termeight}\sum_{\valtwo} 
\left( \sum_{\termseven}\coeffdone_{\termseven, \termeight} + \coeffdyone_\termeight \right)\cdot\sem{\termfive}(\valtwo)\cdot\sem{\subst\termeight \varone\valtwo}(\stateone)) \\
 & \qquad \qquad \mid  \coeffone_\stateone \leq 1 \wedge \coefftwo_\stateone \leq 1 \wedge \coeffone_\stateone + \coefftwo_\statetwo \leq \appl{\howe\bisim}{\stateone}{\statetwo}\}\\ \ \\
& = \max\{ \sum_{\stateone}  \sum_{\termseven, \termeight}( \coeffdone_{\termseven, \termeight} \sum_{\valone}\coeffone_\stateone
 \cdot\sem{\termthree}(\valone)\cdot\sem{\subst\termseven \varone\valone}(\stateone)+\coefftwo_\stateone \cdot\sem{\termfive}(\valtwo)\cdot\sem{\subst\termeight \varone\valtwo}(\stateone)\\
 &\qquad + \sum_{\stateone}\coeffone_\stateone \cdot\sum_{\termseven} \coeffdxone_{\termseven} \sum_{\valone}
 \cdot\sem{\termthree}(\valone)\cdot\sem{\subst\termseven \varone\valone}(\stateone)\\
  & \qquad  + \sum_{\stateone} \coefftwo_\stateone \sum_{\termeight}\coeffdyone_{\termeight}\sum_{\valtwo} \cdot\sem{\termfive}(\valtwo)\cdot\sem{\subst\termeight \varone\valtwo}(\stateone)\\
 & \qquad \mid  \coeffone_\stateone\leq 1\wedge \coefftwo_\stateone \leq 1 \wedge \coeffone_\stateone + \coefftwo_\statetwo \leq \appl{\howe\bisim}{\stateone}{\statetwo}\}\\ \ \\
& \leq
  \max\{\sum_{\stateone}\sum_{\termseven, \termeight}( \coeffdone_{\termseven, \termeight} \sum_{\valone}\coeffone_\stateone
 \cdot\sem{\termthree}(\valone)\cdot\sem{\subst\termseven \varone\valone}(\stateone)+\coefftwo_\stateone \cdot\sem{\termfive}(\valtwo)\cdot\sem{\subst\termeight \varone\valtwo}(\stateone) \\ &\qquad  \mid \coeffone_\stateone \leq 1 \wedge \coefftwo_\stateone \leq 1 \wedge \coeffone_\stateone + \coefftwo_\statetwo \leq \appl{\howe\bisim}{\stateone}{\statetwo}\}\\
 &+ \max \{\sum_{\stateone}\coeffone_\stateone\cdot \sum_{\termseven} \coeffdxone_{\termseven} \sum_{\valone}
 \cdot\sem{\termthree}(\valone)\cdot\sem{\subst\termseven \varone\valone}(\stateone)\\ &\qquad \mid  \coeffone_\stateone \leq 1, \coefftwo_\stateone \leq 1 \wedge \coeffone_\stateone + \coefftwo_\statetwo \leq \appl{\howe\bisim}{\stateone}{\statetwo}\}\\
& +\max \{  \sum_{\stateone} \sum_{\termeight}\coeffdyone_{\termeight}\sum_{\valtwo} \coefftwo_\stateone\cdot\sem{\termfive}(\valtwo)\cdot\sem{\subst\termeight \varone\valtwo}(\stateone)\\
 & \qquad \mid  \coeffone_\stateone \leq 1, \coefftwo_\stateone \leq 1 \wedge \coeffone_\stateone + \coefftwo_\statetwo \leq \appl{\howe\bisim}{\stateone}{\statetwo}\}\\ \ \\
  & \leq
 \max\{ \sum_{\stateone}\sum_{\termseven, \termeight}( \coeffdone_{\termseven, \termeight} \sum_{\valone}
 \coeffone_\stateone\cdot\sem{\termthree}(\valone)\cdot\sem{\subst\termseven \varone\valone}(\stateone)+\coefftwo_\stateone \cdot\sem{\termfive}(\valtwo)\cdot\sem{\subst\termeight \varone\valtwo}(\stateone) \\ &\qquad \mid  \coeffone_\stateone \leq 1 \wedge \coefftwo_\stateone \leq 1 \wedge \coeffone_\stateone + \coefftwo_\statetwo \leq \appl{\howe\bisim}{\stateone}{\statetwo}\}\\
& + \sum_{\termseven} \coeffdxone_{\termseven} 
 + \sum_{\termeight}\coeffdyone_{\termeight}
\end{align*}

We can now apply the induction hypothesis to 
$\appl{\howe\bisim}\termthree\termfive \leq \epsthree $.
We obtain that: $\appl{\howe\bisim}{\sem\termthree}{\sem\termfive} \leq \epsthree $.
So there exist $(\coeffdtwo_{\valone, \valtwo})_{{\valone \in \supp{\sem\termthree}}, {\valtwo \in \supp{\sem\termfive}}}$, and $({\coeffdxtwo_\valone})_{\valone \in \supp{\sem\termthree}}$, and $(\coeffdytwo_\valtwo)_{\valtwo \in \supp{\sem\termfive}}$, such that:
\begin{align}
\label{eqn1b}
\sum_{\valone, \valtwo}\coeffdtwo_{\valone, \valtwo} \cdot \appl{\howe\bisim}{\valone}{\valtwo} & + \sum_{\valone}\coeffdxtwo_\valone + \sum_{\valtwo} \coeffdytwo_\valtwo = \appl{\howe\bisim}{\termthree}{\termfive} \leq \epsthree \\
\sum_{\valone}\coeffdtwo_{\valone, \valtwo} + \coeffdytwo_\valtwo &= \sem{\termfive}({\valtwo})\\
\sum_{\valtwo}\coeffdtwo_{\valone, \valtwo} + \coeffdxtwo_\valone &= \sem{\termthree}({\valone})
\end{align}

And now we have:
\begin{align*}
& \appl{\howe\bisim }{\sem{\termsix\termthree}}{\sem{\termfour\termfive}} \\
 & \leq \max\{ \sum_{\stateone}\sum_{\termseven, \termeight} \coeffdone_{\termseven, \termeight}( \sum_{\valone}
 (\sum_{\valtwo}\coeffdtwo_{\valone, \valtwo} + \coeffdxtwo_\valone)\cdot\sem{\subst\termseven \varone\valone}(\stateone) \cdot \coeffone_\stateone \\
 & \qquad\qquad\qquad\qquad \quad + \sum_{\valtwo} (\sum_{\valone}\coeffdtwo_{\valone, \valtwo} + \coeffdytwo_\valtwo )\cdot\sem{\subst\termeight \varone\valtwo}(\stateone)\cdot \coefftwo_\stateone ) \\ &\qquad \qquad \mid  \coeffone_\stateone \leq 1, \coefftwo_\stateone \leq 1 \wedge \coeffone_\stateone + \coefftwo_\statetwo \leq \appl{\howe\bisim}{\stateone}{\statetwo}\}\\
& \qquad + \sum_{\termseven} \coeffdxone_{\termseven} 
 + \sum_{\termeight}\coeffdyone_{\termeight}\\ 
\ \\
 & \leq
 \max \{\sum_{\termseven, \termeight}( \coeffdone_{\termseven, \termeight}( \sum_{\valone,\valtwo}\coeffdtwo_{\valone, \valtwo}\sum_{\stateone}
 \sem{\subst\termseven \varone\valone}(\stateone)\cdot \coeffone_\stateone + 
 \sem{\subst\termeight \varone\valtwo}(\stateone) \cdot \coefftwo_\stateone \\
 & \qquad\qquad\qquad\quad + \sum_{\valone}\coeffdxtwo_\valone\cdot \sum_{\stateone}
\sem{\subst\termseven \varone\valone}(\stateone)\cdot \coeffone_\stateone \\
 & \qquad \qquad\qquad\quad +\sum_{\valtwo}\coeffdytwo_\valtwo\cdot\sum_{\stateone} \sem{\subst\termeight \varone\valtwo}(\stateone))\cdot \coefftwo_\stateone
   \\ &\qquad \qquad\mid  \coeffone_\stateone \leq 1 \wedge \coefftwo_\stateone \leq 1 \wedge \coeffone_\stateone + \coefftwo_\statetwo \leq \appl{\howe\bisim}{\stateone}{\statetwo}\}\\
& + \sum_{\termseven} \coeffdxone_{\termseven} 
 + \sum_{\termeight}\coeffdyone_{\termeight}
 \\ \ \\
 & \leq
   \sum_{\termseven, \termeight} \coeffdone_{\termseven, \termeight}( \sum_{\valone,\valtwo}\coeffdtwo_{\valone, \valtwo}\max\{\sum_{\stateone}
 \sem{\subst\termseven \varone\valone}(\stateone)\cdot \coeffone_\stateone + 
 \sem{\subst\termeight \varone\valtwo}(\stateone) \cdot  \coefftwo_\stateone \mid  \coeffone_\stateone \leq 1\wedge \coefftwo_\stateone \leq 1 \wedge \coeffone_\stateone + \coefftwo_\statetwo \leq \appl{\howe\bisim}{\stateone}{\statetwo}\}\\
  & \qquad\qquad\quad +\sum_{\valone}\coeffdxtwo_\valone\cdot\max\{ \sum_{\stateone}
\sem{\subst\termseven \varone\valone}(\stateone)\cdot \coeffone_\stateone 
 \mid  \coeffone_\stateone \leq 1 \wedge \coefftwo_\stateone \leq 1 \wedge \coeffone_\stateone + \coefftwo_\statetwo \leq \appl{\howe\bisim}{\stateone}{\statetwo}\}\\
  & \qquad\qquad\quad +\sum_{\valtwo}\coeffdytwo_\valtwo\cdot \max\{ \sum_{\stateone} \sem{\subst\termeight \varone\valtwo}(\stateone)\cdot \coefftwo_\stateone
\mid \coeffone_\stateone\leq 1 \wedge \coefftwo_\stateone \leq 1 \wedge \coeffone_\stateone + \coefftwo_\statetwo \leq \appl{\howe\bisim}{\stateone}{\statetwo}\})\\
& + \sum_{\termseven} \coeffdxone_{\termseven} 
 + \sum_{\termeight}\coeffdyone_{\termeight}\\ 
\end{align*}

Now, we can use equation (\ref{eqn1b}), and the fact that the sum of a distribution is always lesser or equal to 1: 

\begin{align*} 
 \appl{\howe\bisim }{\sem{\termsix\termthree}}{\sem{\termfour\termfive}} \leq
  &\sum_{\termseven, \termeight} \coeffdone_{\termseven, \termeight}\left( \sum_{\valone,\valtwo}\coeffdtwo_{\valone, \valtwo}
 \appl{\howe\bisim}{\subst\termseven\varone\valone}{\subst\termeight\varone\valtwo}
  + \sum_{\valone}\coeffdxtwo_\valone
  +\sum_{\valtwo}\coeffdytwo_\valtwo\right)\\
& + \sum_{\termseven} \coeffdxone_{\termseven} 
 + \sum_{\termeight}\coeffdyone_{\termeight}\\
 \end{align*}

 We can here use lemma \ref{pseudosubst}, which states that $\howe \bisim$ is pseudo-substitutive:
 \begin{align*}
 \appl{\howe\bisim }{\sem{\termsix\termthree}}{\sem{\termfour\termfive}} \leq
  &\sum_{\termseven, \termeight} \coeffdone_{\termseven, \termeight}\left( \sum_{\valone,\valtwo}\coeffdtwo_{\valone, \valtwo} (\appl{\howe\bisim}{\termseven}{\termeight} +
 \appl{\howe\bisim}{\valone}{\valtwo})
  + \sum_{\valone}\coeffdxtwo_\valone
  +\sum_{\valtwo}\coeffdytwo_\valtwo\right)\\
& + \sum_{\termseven} \coeffdxone_{\termseven} 
 + \sum_{\termeight}\coeffdyone_{\termeight}\\
 & \leq
  \sum_{\termseven, \termeight} \coeffdone_{\termseven, \termeight} (\sum_{\valone,\valtwo}\coeffdtwo_{\valone, \valtwo}) \cdot \appl{\howe\bisim}{\termseven}{\termeight} + \sum_{\termseven} \coeffdxone_{\termseven} 
 + \sum_{\termeight}\coeffdyone_{\termeight} \\
  & + \sum_{\termseven, \termeight} \coeffdone_{\termseven, \termeight}  \sum_{\valone,\valtwo}\coeffdtwo_{\valone, \valtwo}\cdot  
 \appl{\howe\bisim}{\valone}{\valtwo}
  + \sum_{\valone}\coeffdxtwo_\valone
  +\sum_{\valtwo}\coeffdytwo_\valtwo)\\
\end{align*}
and, since $\sum_{\valone, \valtwo}\coeffdtwo_{\valone, \valtwo} \leq 1$, and similarly $\sum_{\termseven, \termeight}\coeffdone_{\termseven, \termeight} \leq 1$, we have that:

\begin{align*}
\appl{\howe\bisim }{\sem{\termsix\termthree}}{\sem{\termfour\termfive}} & \leq 
\sum_{\termseven, \termeight} \coeffdone_{\termseven, \termeight} \cdot \appl{\howe\bisim}{\termseven}{\termeight} + \sum_{\termseven} \coeffdxone_{\termseven} 
 + \sum_{\termeight}\coeffdyone_{\termeight} \\
  & +   \sum_{\valone,\valtwo}\coeffdtwo_{\valone, \valtwo}\cdot  
 \appl{\howe\bisim}{\valone}{\valtwo}
  + \sum_{\valone}\coeffdxtwo_\valone
  +\sum_{\valtwo}\coeffdytwo_\valtwo\\
\end{align*}
We can now use equations (\ref{eqn1}) and (\ref{eqn1b}): 
$$\appl{\howe\bisim }{\sem{\termsix\termthree}}{\sem{\termfour\termfive}} \leq \appl{\howe\bisim}{\sem \termsix}{\sem \termfour} + \appl{\howe\bisim}{\sem \termthree}{\sem \termfive} \leq \epsone + \epsthree $$

\end{itemize}
\end{proof}

%%%%%%%%%%%%%%%%%%%%%%%%%%%%%%%%%%%%%%%%%%%%%%%%%%%%%%%

\condskip
\begin{lemma}\label{KL1}
$$\appl{\howe\bisim}{\subst\termone\varone\valone}{\subst\termtwo\varone\valone} \leq \appl{\howe\bisim}{\abstr\varone\termone}{\abstr\varone\termtwo}$$
\end{lemma}
\begin{proof}
Let be $\epsone$ such that:
$\judghowe{}\bisim{\abstr \varone \termone}{\abstr \varone \termtwo}\epsone$

The only rule that can have been applied is:
$$
\AxiomC{$\judghowe{\varone}\bisim\termone \termfour \epsthree $}
\AxiomC{$\appl\bisim{\abstr\varone\termfour}{\abstr \varone \termtwo} \leq \epstwo $}
\AxiomC{$\wfj {}{\abstr \varone \termtwo} $}
  \TrinaryInfC{$\judghowe{}{\bisim}{\abstr\varone\termone}{\abstr \varone \termtwo}{\epsone = \epsthree + \epstwo} $}
\DisplayProof
$$
We can now apply Lemma \ref{pseudosubst} to $\judghowe{\varone}\bisim\termone \termfour \epsone$, and we see that:
$\appl{\howe\bisim}{\subst\termone \varone \valone}{\subst \termfour \varone \valone} \leq \appl{\howe \bisim}\termone \termfour \leq \epsone$.
Moreover, we know that $\appl\bisim{\abstr\varone\termfour}{\abstr \varone \termtwo} \leq \epstwo $. Since $\bisim$ is a fixpoint for $\functionnal$, we can see that:
\begin{align*}
\epsthree &\geq  \appl \bisim {\abstr \varone \termfour}{\abstr \varone \termtwo}  
= \appl \bisim {\dval {\abstr \varone \termfour}}{\dval{\abstr \varone \termtwo}}
 \geq \appl \bisim {\subst  \termfour \varone \valone}{\subst  \termtwo \varone \valone} 
\end{align*}
and now we can conclude by Lemma \ref{pseudotrans} that:
$\appl{\howe\bisim}{\subst\termone \varone \valone}{\subst \termtwo \varone \valone} \leq \epstwo + \epsthree $
\end{proof}
\condskip
Now we extend these two lemmas to $\closure{\howe \bisim}$:
\begin{lemma}
Let be $\termone$, $\termtwo$ two terms.
Then $$\appl{\closure{\howe\bisim}}{\dval{\sem\termone}}{\dval{\sem \termtwo}}\leq \appl{\closure{\howe\bisim}}{\termone}{ \termtwo} $$
\end{lemma}
\begin{proof}
Let be $\epsone$ such that the judgement $\judgclh{\termone}{ \termtwo} \epsone$ is valid by the rules of figure \ref{closure}. 
We are going to show by induction on the structure of its derivation that: $\appl{\closure{\howe\bisim}}{\dval{\sem\termone}}{\dval{\sem \termtwo}} \leq \epsone$. We consider different cases depending of the structure of the proof tree used to derive the validity of $\judgclh{\termone}{ \termtwo} \epsone$:
\begin{itemize}
\item If the proof tree is: 
$$\AxiomC{$\appl{\howe\bisim}{\termone}{\termtwo}\leq\epsone$}
\UnaryInfC{$\judgclh{\termone}{\termtwo} \epsone $}
\DisplayProof$$ 
We can use Lemma \ref{KL}, and we obtain that $\appl {\howe \bisim}{\sem \termone}{\sem \termtwo} \leq \epsone$.
Now we can see that:
\begin{align*}
\appl {\closure{\howe \bisim}}{\dval{\sem \termone}}{\dval{\sem \termtwo}} & \leq \appl{\howe \bisim}{\dval {\sem \termone}}{\dval{\sem \termtwo}} &&\text{ since } \leqmetr {\howe \bisim}{\closure{\howe \bisim}} \\
& = \appl{\howe \bisim}{\sem \termone}{\sem \termtwo} &&\text{ by construction of the extension of $\howe \bisim$ to $\setterm$ }\\
&\leq \epsone
\end{align*}
\item If the proof tree is of the form:
$$
\AxiomC{$\judgclh{\termone}{\termtwo}\epsthree$}
\AxiomC{$\judgclh{\termtwo}{\termone}\epstwo$}
\BinaryInfC{$\judgclh{\termone}{\termtwo} {\epsone = \min(\epsthree, \epstwo)} $}
\DisplayProof
$$
We can apply the induction hypothesis to $\judgclh \termone \termtwo \epsthree$ and $\judgclh \termtwo \termone \epstwo$. 
We obtain that 
$\appl{\closure{\howe\bisim}}{\dval{\sem\termone}}{\dval{\sem\termtwo}} \leq \epsthree$ and that $\appl{\closure{\howe\bisim}}{\dval{\sem\termtwo}}{\dval{\sem\termone}} \leq \epstwo$. Since $\closure{\howe\bisim}$ is symmetric, it means that:
$\appl{\closure{\howe\bisim}}{\dval{\sem\termone}}{\dval{\sem\termtwo}} \leq \epstwo$. And so we have the result.
\item If the proof tree is of the form:
$$
\AxiomC{$\judgclh{\termone}{\stateone}\epsthree$}
\AxiomC{$\judgclh{\stateone}{\termtwo}\leq\epstwo$}
\BinaryInfC{$\judgclh{\termone}{\termtwo} {\epsone = \epsthree + \epstwo} $}
\DisplayProof$$
If $\epsone = 1$, the result holds. Otherwise, please observe that $\stateone$ cannot be a distinguished value. So there exist a closed term $\termthree$ such that $\stateone = \termthree$.
By induction hypothesis: $\appl{\closure{\howe\bisim}}{\dval{\sem\termone}}{\dval{\sem\termthree}}\leq\epsthree$, and
$\appl{\closure{\howe\bisim}}{\dval{\sem\termthree}}{\dval{\sem\termtwo}}\leq \epstwo$.
So by Lemma \ref{lifting}, and since $\closure{\howe \bisim}$ verifies the triangular inequality,  we have:
$\appl{\closure{\howe\bisim}}{\dval{\sem\termone}}{\dval{\sem\termtwo}} \leq \epsthree + \epstwo$.
\end{itemize}
\end{proof}
\condskip
\begin{lemma}
For every $\termone$, $\termtwo$: $$
\appl{\closure{\howe\bisim}}{\subst\termone\varone\valone}{\subst\termtwo\varone\valone} \leq \appl{\closure {\howe\bisim}}{\dval{\abstr \varone \termone}}{\dval{\abstr\varone\termtwo}}$$
\end{lemma}
\begin{proof}
Let be $\epsone$ such that the judgement 
$\judgclh {\dval{\abstr \varone \termone}}{\dval{\abstr \varone \termtwo}} \epsone$ is valid.
As for the previous lemma, the proof is by induction on the structure of the proof tree for this judgement.
\end{proof}
\end{proof}
\condskip
Since $\howe \bisim$ is non-expansive by construction, we now have the result we were aiming for:
\condskip
\begin{theorem}
$\bisim$ is non-expansive.
\end{theorem}
\condskip

\begin{proof}
As a consequence of Theorem \ref{bisimfix}, $\bisim = \howe \bisim$. Since $\howe \bisim$ is non-expansive, the result holds. 
\end{proof}

%As a consequence, we can see that $\bisim$ is sound with respect to the context metric.
%%%%%%%%%%%%%%%%%%%%%%%%%%%%%%%%%%%%%%%%%%
\subsection{On Full-Abstraction and Pairs}
%%%%%%%%%%%%%%%%%%%%%%%%%%%%%%%%%%%%%%%%%%
The bisimulation distance is a sound approximation of the context
distance.  But how about full-abstraction? Is there any hope to prove
that the two coincide? The answer is negative: there are terms whose
distance is \emph{strictly} higher in the bisimulation metric than in
the context (or trace) metric.
\condskip
\begin{example}
Consider the following terms: $\termone$ corresponds to the program
that takes an argument, and then returns $\identity$ with probability
$\frac 1 2$, and diverges with probability $\frac 1 2$. $\termtwo$
corresponds to the program which chooses first between the function
which return $\identity$ whenever it is called, and the function which
diverges whenever called. Formally:
$$
\termone \defi \abstr \varone {(\psum \identity \diver)}; \qquad
\termtwo \defi \psum {(\abstr \varone \identity)}{(\abstr \varone
  \diver)}.
$$ 
These two terms are at distance $0$ for the context distance:
since the calculus is linear, the step where the choice is done is
irrelevant. However, $\appl \bisim \termone
\termtwo = \frac 1 2$: the proof, use the characterisation of bisimulation distance by testing
from \cite{DesharnaisCONCUR99}, in which not only linear tests, but
also more complicated tests (like \emph{threshold tests}) are
available. 
\end{example}
\condskip
  
But how about pairs? Indeed, for the sake of simplicity, we have
presented the metatheory of the bisimulation metric for a purely
applicative $\lambda$-calculus. Following the lines of our discussion in Section~\ref{sect:pairs}, however,
the LMC $\markovterm$ can be extended into one handling pairs in 
a relatively simple way. The difficulties we encountered when trying
to evaluate the (trace, or context) distance between pairs of terms
unfortunately remain: it is not clear whether coinduction could provide any additional
\emph{advantage} over contextual distance.
As for the trace metric in the previous section, we would like to
extend the bisimulation metric to a language with pairs. In order to
do that, we add the action $\tenseur \termfour$ to the LMC
$\markovterm$. We transform the definition to the probability matrix
$\probmatrterm$ by adding:
$$
\probmatrterm {(\pair \termone \termtwo)}{({\tenseur \termthree})} = \sum_{\valone, \valtwo} \sem \termone( \valone) \cdot \sem \termtwo (\valtwo) \cdot \dirac { {\subst \termthree {\varone, \vartwo} {\valone, \valtwo}} } 
$$
We now have to transform the definition of validity for Howe's judgement in order to consider the case of pairs:
 $$
\AxiomC{$
\begin{array}{c}
\judghowe \contone {\metrone}\termone\termfour\epsone\\
\judghowe \conttwo \metrone\termtwo\termfive \epsthree
\end{array}$
}
\AxiomC{$
\begin{array}{c}
\wfj{\contone, \conttwo}{\termthree}\\
\appl\metrone{\pair \termfour\termfive}{\termthree} \leq \epstwo 
 \end{array}$}
  \BinaryInfC{$\judghowe {\contone, \conttwo} \metrone{\pair \termone\termtwo}\termthree{\epsone +\epsthree + \epstwo}$}
  \DisplayProof$$
%%%%%%%%%%%%%%%%%%%%%%%%%%%%%
\section{The Tuple Distance}\label{sect:tupledistance}
%%%%%%%%%%%%%%%%%%%%%%%%%%%%%
%%%%%%%%%%%%%%%%%%%%%%%%%%%%%%%%%%%%%%%%%%%%%%%%%%%%
%\subsection{A Markov Chain for Tuples of Values}
%%%%%%%%%%%%%%%%%%%%%%%%%%%%%%%%%%%%%%%%%%%%%%%%%%%%
The two metrics we have just defined have been shown to be
non-expansive, even if the calculus is extended with pairs. In that case, 
however, they do not represent so much of an
improvement with respect to the context distance. Please recall
\emph{where} the problem comes from: we would like to define actions
starting from $\pair \termone \termtwo$, and respecting the affine
paradigm. We have seen that taking \emph{projections} as actions lead
to an unsound metric, and we have circumvented the problem by
considering an action $\tenseur \termthree$,
following~\cite{DengZhang}. Intuitively the action $\tenseur
\termthree$ corresponds to replacing the free variables of
$\termthree$ (which are supposed to be included in
$\{\varone,\vartwo\}$) by the components of the pair: if for instance
$\valone$ and $\valtwo$ are values, we have that $\doact{\pair \valone
  \valtwo}{\tenseur \termthree}{\dirac{\subst \termthree{\varone,
      \vartwo}{\valone, \valtwo}}}$. But what can any environment
$\termthree$ do if we give it $\valone$ and $\valtwo$ as two
values to interact with? Let us suppose that both $\valone$ and
$\valtwo$ are functions, and remember that we are in an affine
setting. The environment can (probabilistically) pass some arguments
to $\valone$, and independently some other arguments to $\valtwo$, and
then possibly pass to one of the two programs an argument that
contains the other one. The idea behind the construction we present in this
section, then, is to \emph{keep} the information about the two components
of the pairs in the states until they really interact with each other.

Our idea can be made concrete by introducing another LMC, whose states
are not closed terms anymore, but \emph{tuples} in the form
$\tuple{\valone_1}{\cdots}{\valone_n}$, where $\valone_1, \cdots
\valone_n $ are values. The possible actions the environment can
perform on a tuple $\tuple{\valone_1}{\cdots}{\valone_n}$
correspond to the choice of an index $i \in \{1,\cdots,n\}$ and of an
action to apply to the value $\valone_i$. If $\valone_i$ is a pair,
the only possible action is to split it into two components. We call
this action $\actcut i$. If $\valone_i$ is a function, the environment
can pass it an argument, which can possibly be constructed using other
$\valone_j$'s. More precisely, the argument is built by way of an open
term $\ctxone$, and a typing context $\contone$, such that $\wfj
\contone \ctxone$, and $\contone$ is a subset of $\{\varone_j \mid j
\neq i\}$: the free variables of $\contone$ represent the places where
other values $\valone_j$, with $j \neq i$, are used. Moreover, we ask
that for any value $\valtwo_1, \cdots \valtwo_n$, the term obtained in
substituting $\varone_j$ by $\valtwo_j$ is a value: it
means that $\ctxone$ is one of the $\varone_j$, or of the form $\abstr
\vartwo \ctxtwo$. We call a pair $(\contone, \ctxone)$ which verifies
these conditions a \emph{$(n,i)$-open value}. Formally, the LMC
$\markovtupl = (\settupl, \labelstupl, \probmatrtupl)$ is defined in
Figure \ref{markovtupl}.
\begin{figure}
\begin{center}
\fbox{%\footnotesize
\begin{minipage}{.97\textwidth}
\begin{align*}
\settupl&= \left\{\tuple {\valone_1}{\cdots}{\valone_n} \mid \valone_1, \cdots,\valone_n \text{ closed values } \right\}\\
\labelstupl&= \{ \actcut i \mid i \in \NN \} \cup \{\actappl{ (\contone, \ctxone)} i \mid i \in \NN, (\contone, \ctxone) \text{ a }(n,i)\text{-open-value }\}.
\end{align*}
\vspace{-20pt}
\begin{align*}
\applprob{\probmatrtupl} { \tuple{\torvone_1,\cdots}{{\pair \termtwo\termthree}}{\cdots, \torvone_n}}{\actcut i}{  \tuple {\torvone_1,\cdots \torvone_{i-1}}{ \valone, \valtwo}{\torvone_{i+1}\cdots, \torvone_n}} &=    \sem  \termtwo (\valone) \cdot \sem \termthree(\valtwo)\\
\applprob {\probmatrtupl }{\tuple{\torvone_1,\cdots}{{\abstr \vartwo \termtwo}}{\cdots, \torvone_n}}{\actappl{(\contone,\ctxone)} i}{\tuple{\torvone_{h_1}, \cdots}{ \valtwo}{\cdots, \torvone_{h_m}}} & = \sem {\subst\termtwo\vartwo {\subst \ctxone{\varone_{j_l}} {\torvone_{j_l}}}}(\valtwo)\\ 
\text{with }  \{1,\cdots, n \} = {i} \cup \{j_1,...,j_k\} \cup \{h_1,...,h_m\} &\text{ (disjoint union) } \text{ and } \contone = \varone_{\j_1}, \cdots \varone_{j_k}
\end{align*} 
\end{minipage}}
\end{center}
\caption{The Tuple LMC} \label{markovtupl}
\end{figure}
%%%%%%%%%%%%%%%%%%%%%%%
\subsection{The Metric}
%%%%%%%%%%%%%%%%%%%%%%%
We are going to define a metric on closed terms which corresponds to
linear tests in $\markovtupl$. First, we define \emph{tuple traces}
simply as words over $\labelstupl$. The probability to succeed in
doing a trace $\traceone$ starting from a tuple $\tuplone \in
\settupl$ can be naturally defined, and paves the way to defining a
metric on tuples of values:
%{\footnotesize
\begin{align*}
\probtrtupl \tuplone \epsone &= 1;\\
\probtrtupl \tuplone {\concat \actone \traceone}&= \sum_\tupltwo \applprob\probmatrone \tuplone \actone \tupltwo \cdot \probtrtupl \tupltwo \traceone;\\
\appl\metrtrtupl \tuplone \tupltwo&= \sup_\traceone \abs{\probtrtupl \tuplone \traceone - \probtrtupl \tupltwo \traceone}.
\end{align*}%} 
What we need, however, is a metric on \emph{programs}. Please remember
that states of the LMC $\markovtupl$ are tuples of \emph{values}. 
Any program $\termone$, however,  can be viewed as the distribution
of values obtained by evaluating it, i.e. its semantics
$\sem{\termone}$:
%{\footnotesize
\begin{align*}
\appl\metrtrtupl \termone \termtwo = \sup_\traceone \abs {&{\sum \sem \termone(\valone)\cdot \probtrtupl {\tuplonea \valone}{\traceone}} \\ & -\,{\sum \sem \termtwo(\valtwo) \cdot\probtrtupl {\tuplonea \valtwo}{\traceone}}}.
\end{align*}
%}
The just introduced metric should at least be put in relation to the
context metric for it to be useful. We know from Section
\ref{sect:tracedistance} that the context metric coincides with the
trace metric. The following theorem relates the trace
metric $\metrtr$ and the metric $\metrtrtupl$: \UGO{From now on, there
  is definitely the weakest part of the paper. I think you should at
  least try to rephrase the sentence after Theorem 7, which does not
  make much sense. }
\condskip
\begin{theorem}\label{nonexpansivtupl}
Let $I$ be any finite set of variables, and
$\{\valone_\varone\}_{\varone \in I}$ and $\{\valtwo_\varone\}_{\varone \in I}$  any two collections of values.
For any open term $\ctxone$ such that $\wfj{I}{\ctxone} $,
 it holds
that:
\begin{align*}
&\appl \metrtr {{\subst\ctxone\varone{\valone_\varone}_{(\varone \in I)}}}{{\subst \ctxone\varone{\valtwo_{\varone}}}_{(\varone \in I)}} \\ &
\qquad  \leq \appl \metrtrtupl {\tuplonea{\valone_\varone}_{(\varone\in I)}}{\tuplonea{\valtwo_\varone}_{(\varone \in I)}}
\end{align*}
\end{theorem}
\begin{proof}
The proof of Theorem \ref{nonexpansivtupl} is similar to the proof of non-expansiveness for the trace metric : first we define a small step semantics, which corresponds to the transition relation in the Markov Chain $\markovtupl$, then we define another small step semantics, which corresponds to keep separated the context, which is now seen as a term with several holes, and the tuple used to fill it, and we end the proof by defining a notion of $\epsone$-parentality for disributions over pairs of contexts and tuple, and showing a stability result for $\epsone$-parents distributions. These steps are displayed in more details below.

%%%%%%%%%%%%%%%%%%%%%%%%%%%%%%%%%%%%%%%%%%%%%%%%%%%%%%
%Proof of non-expansiveness for tuple
%%%%%%%%%%%%%%%%%%%%%%%%%%%%%%%%%%%%%%%%%%%%%%%%%%%%%%

\subsubsection{Trace Semantics Big Steps for Tuples}
We're going to be interested in the labelled transition system on finite distributions over $\settupl$ induced by the Markov Chain $\markovtupl$.
\begin{definition}
 We'll note $\fdistrs \settupl$ the set of finite distributions over $\settupl$. We define a reduction relation $\wtrbs \distrone \actone \distrtwo$, where $\distrone, \distrtwo \in \fdistrs \settupl$, and $\actone \in \labelstupl$, by :
$$\wtrbs \distrone \actone {\sum_\tuplone \distrone(\tuplone) \cdot \probmatrtupl(\tuplone)(\actone) } $$
\end{definition}

Now, we define the success probability of a trace for a distribution as :
\begin{definition}
If $\traceone = \actone_1 \cdots \actone_n$, 
$$\probtrbstupl \distrone \traceone = \sumdistr \distrtwo \text{ with }\wtrbs{\distrone}{\actone_1}{\wtrbs{\cdots}{\actone_n}{\distrtwo}}$$
\end{definition}

The relation between this deterministic labelled transition system and the Markov Chain $\markovtupl$ can be expressed by the following lemma :
\begin{lemma}
Let be $\tuplone \in \settupl$, and $\traceone$ a trace.
Then $\probtrbstupl{\dirac \tuplone }{\traceone} = \probtrtupl \tuplone \actone$.
\end{lemma}

\subsubsection{Trace Semantics Small Steps for Tuples}
We would like now to have a notion of small-step semantics for tuples corresponding to the trace semantics of the Markov Chain.
Since we are now small steps, we should consider not only tuples of values, but tuples of terms as well.
Moreover, during the execution, we should remember which term of the tuple is being reduced. For this reason, we must add intermediate states, where there is explicit focus on terms being evaluated.

\begin{definition}
\begin{itemize}
\item We define a set $\torvstates$ consisting in closed terms of $\calcwpair$, and distinguished values of $\calcwpair$:
$\torvstates = \{ \termone \mid \termone \text{ closed term } \} \cup \{\dval \valone \mid \valone \text{ closed value } \}$. Then we define the corresponding set of tuples $\setone = \left\{\tuple {\torvone_1}{\cdots}{\torvone_n} \mid \torvone_1, \cdots,\torvone_n \in \torvstates \right\}$
\item $\torvstateswf =  \{ \termone \mid \termone \text{ closed term } \} \cup \{ \focus i \termone \mid \termone \text{ closed term}, \, i \in \NN  \} \cup \{\dval \valone \mid \valone \text{ closed value } \}$, and  $\setonewf = \left\{\tuple {\torvone_1}{\cdots}{\torvone_n} \mid \torvone_1, \cdots,\torvone_n \in \torvstateswf \text{ and the focus integer are all distincts }\right\}$
\end{itemize}
\end{definition}

The term which should be reduced first is the term which has the smaller focus index. That's the sense of the following definition.
\begin{definition}
For any $\tuplone \in \setonewf$, we note $\focusindex \tuplone$ defined as:
\begin{itemize}
\item $\focusindex \tuplone = \infty$ if $\tuplone$ has no element with focus.
\item $\focusindex {\tuple {\torvone_1}{\cdots}{\torvone_n}} = j \text{ such that } \torvone_i = \focus j \termone \text{ and } j \text{ is the smaller focus in }\tuplone $
\end{itemize}

\end{definition}

We now define a small step probabilistic labelled reduction relation, where the actions can be divided in two kinds :
\begin{definition}
We define a labelled reduction relation $ \wtrtupl \tuplone \actone \distrone$ where $\tuplone \in \setonewf$, $\distrone$ a distribution over $\setonewf$, and where $\actone \in\actss = \{ \internact\} \cup  \{ \acteval i \mid i \in \NN \} \cup \{ \actcut i \mid i \in \NN \}\cup \{\actappl{(\contone,\ctxone)} i \mid i \in \NN, \, (\contone,\ctxone) \text{ a } (n,i) \text{open-value for a }n \in \NN \}$.
The rules are the one given in figure \ref{sstrsemtupl}.
\end{definition}
$\internact$ is called an internal action, and corresponds to the internal reduction terms under focus in the tuple. The other actions are called external actions, and correspond to interactions with the environment. The definition given in Figure \ref{sstrsemtupl} use the small step semantics for term $\ssp{}{}$. 
\begin{figure}
{\small
$$\AxiomC{$\focusindex {\tuple{\torvone_1}{\cdots, \focus i \termone}{\cdots, \torvone_n}} = i $}
  \AxiomC{$\ssp \termone \distrone $}
  \BinaryInfC{$\wtrtupl{\tuple{\torvone_1, \cdots}{\focus i \termone, \cdots}{\torvone_n}}{\internact}{\sum \distrone(\termtwo) \cdot \dirac {\tuple{\torvone_1, \cdots}{\focus i \termtwo, \cdots}{\torvone_n}} }$} 
  \DisplayProof$$
$$\AxiomC{$\focusindex {\tuple{\torvone_1}{\cdots, \focus i \termone}{\cdots, \torvone_n}} = i $}
  \AxiomC{$ \termone \not \rightarrow $}
  \BinaryInfC{$\wtrtupl{\tuple{\torvone_1, \cdots}{\focus i \termone, \cdots}{\torvone_n}} {\internact}{\emptydistr}$} 
  \DisplayProof$$
$$\AxiomC{$\focusindex {\tuple{\torvone_1}{\cdots, \focus i \valone}{\cdots, \torvone_n}} = i $}
  \AxiomC{$\valone \text{ is a value } $}
  \BinaryInfC{$\wtrtupl{\tuple{\torvone_1, \cdots}{\focus i \valone, \cdots}{\torvone_n}}{\internact}{ \dirac {\tuple{\torvone_1, \cdots}{\dval \valone, \cdots}{\torvone_n}} }$} 
  \DisplayProof$$
$$\AxiomC{$\focusindex {\tuple{\torvone_1}{\cdots}{ \torvone_n}} = \infty $}
  \UnaryInfC{$\wtrtupl{\tuple{\torvone_1}{\cdots, \torvone_{i-1}, \termone, \torvone_{i+1}}{\cdots,\torvone_n}}{\acteval i}{ \dirac {\tuple{\torvone_1, \cdots}{\torvone_{i-1},\focus 1 \termone,\torvone_{i+1}, \cdots}{\torvone_n}} }$} 
  \DisplayProof$$

$$\AxiomC{$\focusindex {\tuple{\torvone_1}{\cdots}{ \torvone_n}} = \infty $}
  \UnaryInfC{$\wtrtupl{\tuple{\torvone_1}{\cdots, \torvone_{i-1}, \dval {\pair \termone \termtwo}, \torvone_{i+1}}{\cdots,\torvone_n}}{\actcut i}{ \dirac {\tuple{\torvone_1, \cdots}{\torvone_{i-1},\focus 1 \termone, \focus 2 \termtwo,\torvone_{i+1}, \cdots}{\torvone_n}} }$} 
  \DisplayProof$$

$$\AxiomC{$\focusindex {\tuple{\torvone_1}{\cdots}{ \torvone_n}} = \infty $}
  \UnaryInfC{$\wtrtupl{\tuple{\torvone_1}{\cdots, \torvone_{i-1}, \dval {\abstr\vartwo\termone}, \torvone_{i+1}, \cdots, \torvone_{j-1}, \dval \valone, \torvone_{j+1}}{\cdots,\torvone_n}}{\actappl {\varone_j} i}{ \dirac {\tuple{\torvone_1, \cdots}{\torvone_{i-1},\focus 1 {\subst \termone \vartwo \valone}, \torvone_{i+1}, \cdots, \torvone_{j-1}, \torvone_{j+1}, \cdots}{\torvone_n}} }$} 
  \DisplayProof$$

$$\AxiomC{$\focusindex {\tuple{\torvone_1}{\cdots}{ \torvone_n}} = \infty $}
\AxiomC{$\wfj {(x_{j_k})_{1\leq k \leq l}} \ctxtwo$}
\AxiomC{$\{1,\cdots,n\} = \{j_k \mid 1 \leq k \leq n \} \sqcup \{h_k \mid 1 \leq k \leq m\} $}
  \TrinaryInfC{$\wtrtupl{\tuple{\torvone_1}{\cdots, \torvone_{i-1}, \dval {\abstr\vartwo\termone}, \torvone_{i+1}}{\cdots,\torvone_n}}{\actappl {\abstr \varthree \ctxtwo} i}{ \dirac {\tuple{\torvone_{h_1}}{\cdots,\focus 1 {\subst \termone \vartwo {\abstr\varthree{{\subst \ctxtwo {\varone_{j_k}} {\torvone_{j_k}}}_{1\leq k \leq l}}}}, \cdots}{\torvone_{h_m}}} }$} 
  \DisplayProof$$
}
\caption{small-step trace semantics for tuples}\label{sstrsemtupl}
\end{figure}

We want to formalize the probability of doing a trace for a distribution. First we lift the trace semantics to a reduction (non probabilistic) to distributions. We'll note $\fdistrs \setonewf$ the set of finite distributions over $\setonewf$.

\begin{definition}
We define a labelled relation $\wtrcdistrtupl \distrone \actone \distrtwo$, where $\distrone, \distrtwo \in \fdistrs \setonewf$,  and $\actone \in \actss$. The rules are the one given in Figure \ref{trdistrtupl}.
\end{definition}

\begin{figure}
$$\AxiomC{$\wtrtupl \tuplone \internact \distrtwo  $}
  \UnaryInfC{$\wtrdistrtupl{\distrone \disjplus p \cdot \dirac \tuplone}{\emptytr}{\distrone + p \cdot \distrtwo }$} 
  \DisplayProof$$
$$\AxiomC{$\wtrtupl \tuplone \actone \distrtwo_\tuplone  $}
  \AxiomC{$\distrone \text{ in normal form } $}
  \BinaryInfC{$\wtrdistrtupl{\distrone}{\actone}{\sum_{\tuplone \text{s.t }\focusindex \tuplone = \infty} \distrone(\tuplone)\cdot \distrtwo_\tuplone }$} 
  \DisplayProof$$
$$
\AxiomC{$\wtrdistrtupl \distrone \traceone \distrtwo $}
\AxiomC{$\wtrcdistrtupl \distrtwo \tracetwo \distrthree $}
\BinaryInfC{$\wtrcdistrtupl  {\distrone}{\concat \traceone \tracetwo} \distrthree$}
\DisplayProof
$$
\caption{small-step trace semantics on distributions of tuples}\label{trdistrtupl}
\end{figure}

\begin{definition}
$\probtrsstupl \distrone \traceone = \max\{\sum_{\focusindex \tuplone = \infty} {\normal \distrtwo}(\tuplone) \mid \wtrcdistrtupl \distrone \traceone {\normal \distrtwo}\}$
\end{definition}

Please observe that for any (external or internal) action $\actone$, the relations (between tuples and distributions over tuples) $\wtrtupl \tuplone \actone \distrone$, and $\wtrbs \tuplone \actone \distrone$ are deterministic. It's not the case anymore when we lift to relations between distributions, but we have the following lemma : 

\begin{lemma}
The reduction $\wtrdistrtupl \cdot  \tau \cdot$ on distributions over $\setonewf$ is strongly normalizing. 
\end{lemma}
\begin{proof}
It follows from the fact that the relation $\wtr{}{}{}$ over distributions of terms is strongly normalizing.
\end{proof}

We note $\normal \distrone $ the normal form of $\distrone$ for the relation $\wtrdistrtupl \cdot  \tau \cdot$. By abuse of notation, if $\torvone \in \torvstateswf$, we note $\normal \torvone$ for $\normal {\dirac \torvone}$. We can in fact be more precised on the shape of the normal form of a distribution :

\begin{definition}
 Let be $\torvone \in \torvstateswf$. We define $\reduced \torvone$ by :
\begin{itemize}
\item If $\torvone = \focus i \termone$, then $\reduced \torvone = \sum_\valone\sem \torvone(\valone) \cdot \dirac{\dval \valone}$
\item otherwise, $\reduced \torvone = \dirac \torvone$  
\end{itemize}
\end{definition}
\begin{lemma}\label {normalformsst}
 Let be $\tuplone = {\tuple{\torvone_1}{\cdots}{\torvone_n}}\in \setonewf$. 
Then $\normal \tuplone  = \sum_{\torvtwo_1, \cdots \torvtwo_n} \prod_{1 \leq i \leq n} (\reduced {\torvone_i})(\torvtwo_i) \cdot \dirac{ \tuple{\torvtwo_1}{\cdots}{\torvtwo_n}} $
\end{lemma}
\begin{proof}
The proof is by induction on the maximal number of reduction steps from $\distrone$ to $\normal \distrone$ (which is well defined since $\wtrdistrtupl \cdot \tau \cdot$ is strongly normalizing)
\end{proof}

Now we want to compare the probability to do a trace for the small-step semantics and for the big-step semantics. For doing that, we show first the following lemma :
\begin{lemma}\label{bsssauxext}
Let be $\actone \in \actbs$, and $\distrone \in \fdistrs \statebs$.
%Let be $\distrone$ a distribution on tuples without focus, and $\actone \in \actbs$.
Then let be $\distrtwo$ the distribution over $\statebs$ such that $\wtrbs \distrone \actone \distrtwo$.
 Let be $\distrthree$ the distribution over $\statess$ such that : $ \wtrtupl{\wtrtupl{\wtrtupl {\dval\distrone} {\actone} \distrthree}{\tau}{\cdots}}{\tau}{\normal \distrthree}$.
Then : $$\dval \distrtwo = \normal \distrthree  $$
\end{lemma}
\begin{proof}
Let be $\actone \in \actbs$.
We can see that for every $\tuplone \in \statebs$, there exists an (only one) $\tupltwo \in \statess$ such that :
   $\wtrtupl \tuplone \actone \dirac{\tupltwo}$.
It is sufficient to show that : if $\distrone$ is the distribution over $\statebs$ such that $\wtrbs \tuplone \actone \distrone$, we have that $\normal \tupltwo = \dval \distrone$. The proof of that is by case analysis on the rules of $\wtrtupl{}{\actone}{}$, and using the characterisation given in Lemma \ref{normalformsst} of the normal form for $\wtrtupl{}{\internact}{}$
 
\hide{It is sufficient to show : for every $\tuplone \in \setone$, for every $\tupltwo \in \setonewf$, and for every $\distrone$ distribution over $\setone$,  if $\wtrtupl \tuplone \actone \tupltwo$ and $\wtrbs \tuplone \actone \distrone$, then $\distrone = \reduced \tupltwo$.
The proof of that is by case analysis on the external action $\actone$.
For example, if the action is $\acteval i$, we have :
$\wtrtupl {(\tuple{\torvone_1, \cdots}{\termone}{\cdots, \torvone_n})}{\acteval i}{(\tuple{\torvone_1, \cdots}{\focus 1 \termone}{\cdots, \torvone_n})}$, and $\wtrbs {(\tuple{\torvone_1, \cdots}{\termone}{\cdots, \torvone_n})}{\acteval i}{(\tuple{\torvone_1, \cdots}{\sem \termone}{\cdots, \torvone_n})}$, and the result holds.}
\end{proof}
\hide{
and :
\begin{lemma}
 \label{bsssauxint}
Let be $\distrone$ a distribution over $\setonewf$ (that is, a distribution over tuples potentially with focus).
Then $\normal \distrone = \sum_\tuplone \distrone(\tuplone) \cdot \reduced \tuplone$
\end{lemma}
\begin{proof}
The proof is by induction on the length of the longest reduction path to reach the normal form (if you start from $\distrone$). 
\end{proof}
}
Now we can extend this result to traces :

\begin{lemma}
Let be $\traceone$ a word over $\actbs$. Let be $\distrone$ a distribution over $\statebs$. Then :
$\probtrbstupl \distrone \traceone = \probtrsstupl {\dval\distrone} \traceone$
\end{lemma}
\begin{proof}
The proof is by induction on the length of $\traceone$.

\begin{itemize}
\item if $\traceone = \emptytr$ : $\probtrbstupl \distrone \emptytr = \sumdistr \distrone$. 
\hide{\begin{align*}
\probtrsstupl {\dval\distrone} \emptytr =&  \max\{\sum_{\focusindex \tuplone = \infty} {\normal {\distrtwo}}(\tuplone) \mid \wtrc {\dval\distrone} \emptytr {\normal \distrtwo}\} \\
& = \sum_{\focusindex \tuplone = \infty} {\dval \distrone}(\tuplone) \\&= \sumdistr \distrone
\end{align*}}
Since $\dval \distrone$ is a normal form for $\wtrtupl \cdot \internact \cdot$, we have that : $\probtrsstupl {\dval\distrone} \emptytr =  \sum_{\focusindex \tuplone = \infty} {\dval \distrone}(\tuplone) = \sumdistr \distrone $
\item if $\traceone = \concat \actone \tracetwo$ then let be $\distrtwo$ such that $\wtrbs{\distrone}{\actone}{\distrtwo}$. Then $\probtrbstupl \distrone \traceone = \probtrbstupl \distrtwo \tracetwo$. We apply Lemma \ref{bsssauxext}, and we obtain that : $\wtrtupl{\dval \distrone}{\actone}{\distrthree}$, and $\normal \distrthree = \dval \distrtwo$.
Moreover, we have that : 
$\probtrsstupl {\dval\distrone} \traceone = \probtrsstupl{\distrthree} \tracetwo = \probtrsstupl{\normal \distrthree}{\tracetwo}$
\hide{
 \begin{align*}
\probtrsstupl {\dval\distrone} \traceone =&  \max\{\sum_{\focusindex \tuplone = \infty} {\normal {\distrtwo}}(\tuplone) \mid \wtrc {\dval\distrone} \traceone {\normal \distrtwo}\} \\

& = \sum_{\focusindex \tuplone = \infty} {\dval \distrone}(\tuplone) \\&= \sumdistr \distrone
\end{align*}}
\end{itemize}
\end{proof}

\hide{
If we consider only terms without pairs, we recover the trace distance :
\begin{lemma}
For $\termone$, $\termtwo$ two terms of type without pairs. 
$$\appl{\metrtrtupl}{\termone}{\termtwo} = \appl{\metrtr}{\termone}{\termtwo} $$
\end{lemma}
\begin{proof}
\begin{itemize}
\item Any trace on the Markov Chain on terms can be translated on an equivalent trace on the Markov Chain for tuples who talks only about the first element of the tupe.
\item Any trace on the Markov Chain for tuples, which can be done with a non-zero probability starting from $\tuplonea \termone$ or $\tuplonea \termtwo$ (thai is, we especially don't consider trace with the $\actcut i$ action), can be simulated (for these two state) by a context of the language without pairs (?)
\end{itemize}
\end{proof}
}

\subsubsection{Trace semantics for distribution over contexts and tuples}

Here we consider the same traces used for defining trace semantics for distribution on closed terms.

We are first going to introduce useful notations :
\begin{definition}
We define an operator $\parall n \functionone \functiontwo$ on functions by :
If $\functionone : A \rightarrow \NN$, $\functiontwo : B \rightarrow \NN$ such that : 
\begin{itemize}
\item $A \cap B = \emptyset$
\item $\image \functionone \subseteq \{1, \cdots, n \}$
\end{itemize}
Then $\parall n \functionone \functiontwo : A \cup B \rightarrow \NN$ is defined by :
\begin{itemize}
\item $\parall n \functionone \functiontwo (x)= \functionone(x)$ if $x \in A$
\item $\parall n \functionone \functiontwo (x) = n+ \functiontwo(x)$ if $x \in B$
\end{itemize}
\end{definition}

We now want to define a set of pairs of context with several holes, and tuples used for filling these holes. Formally the idea is the following :
We first define things for the untyped case (without pairs) :

\begin{definition}
\begin{itemize}
\hide{\item For any set of possible open terms $\ctxone_1,...,\ctxone_n$ we say : $$\wfj {\varone_1, \cdots, \varone_m}{\ctxone_1, \cdots \ctxone_n}$$ if there exists a partition of the variables ${\varone_1, \cdots \varone_m}$ in $n$ disjoints sets $\Gamma_i$ such that $\wfj {\Gamma_i}{\ctxone_i}$ for every $1 \leq i \leq n$.}
\item Let be $\functionone : \variables \rightarrow \NN$ a partial injective function, $\ctxone$ an (open) term, and $\tuple{\torvone_1}{ \cdots}{ \torvone_n}$ an element of $\setonewf$. 
We define the judgment $\wfj{\varone_1,... \varone_m}{(\ctxone, \functionone, \tuple{\torvone_1}{\cdots}{\torvone_n})}$ by :
\begin{itemize}
\item $\{\varone_1,\cdots, \varone_m\}\cap\domain \functionone = \emptyset$ 
\item $ \wfj {\varone_1  \cdots \varone_n, (\vartwo)_{\vartwo \in \domain \functionone}}{\ctxone} $
\item $\image \functionone \subseteq \{1\cdots n\}$
\end{itemize}
We define the judgment $\wfjval{\varone_1,... \varone_m}{(\ctxone, \functionone, \tuple{\torvone_1}{\cdots}{\torvone_n})}$ by :
\begin{itemize}
\item $\wfj{\varone_1,... \varone_m}{(\ctxone, \functionone, \tuple{\torvone_1}{\cdots}{\torvone_n})}$
\item if $\ctxone = \vartwo$, then there exists a value $\valone$ such that : $\torvone_{\functionone(\vartwo)} = \dval \valone$
\item if $\ctxone$ is not a variable, $\ctxone$ is an abstraction (that is, $\ctxone = \abstr \vartwo \ctxtwo$, where $\ctxtwo$ is an open term), or $\ctxone$ is a pair (that is $\ctxone = \pair {\ctxtwo_1} {\ctxtwo_2}$, where $\ctxtwo_1$ and $\ctxtwo_2$ are open terms).
\end{itemize}

\item  We define the set of pairs of context and tuples which are well formed :
$$\ctwellformed = \{(\ctxone, \functionone, \tuple {\torvone_1}{\cdots}{\torvone_n}) \text{ such that } \wfj{\emptyset} {(\ctxone, \functionone, \tuple {\torvone_1}{\cdots}{\torvone_n})}   \} $$
%\item We define a subset of $$\ctwellformed$$ which correponds, for a given context, to minimal tuples :
%$$\ctwellformedmin = \{(\ctxone, \tuple {\termone_1}{\cdots}{\termone_n}) \mid \,\left(\wfjt {}{\termone_i}{\typone_i} \right)_{1 \leq i \leq n} \text{ and } \wfj {x_i : \typone_i} \ctxone \and \not\exist \{l_1,...,l_k\} \subseteq\{1,..,n\} \text{ such that }  \wfj {(x_{l_j} : \typone_{l_j})_{1\leq j \leq k}} \ctxone \} $$
\item We define a notion of congruence for elements in $\ctwellformed$ :
For every permutation $\permut : \{1,..,n\} \rightarrow \{1,...,n\}$, $(\ctxone, \functionone \tuple {\torvone_1}{\cdots}{\torvone_n}) \equiv (\ctxone, \compose{\permut^{-1}}\functionone , (\tuple {\torvone_{\permut(1)}}{\cdots}{\torvone_{\permut(n)}})) $
\end{itemize}
\end{definition}

We should modify the definition if we consider a typed calculus :
\begin{definition}
\begin{itemize}
\item We define the judgment $\wfjt{\varone_1 : \typone_1,... \varone_m : \typone_m}{(\ctxone, \functionone, \tuple{\termone_1}{\cdots}{\termone_n})}{\typtwo}$ by :
\begin{itemize}
\item $\{\varone_1,\cdots, \varone_m\}\cap\domain \functionone = \emptyset$ 
\item $ \wfjt {\varone_1 : \typone_1 \cdots \varone_n : \typone_n, (\vartwo : \typthree_\vartwo)_{\vartwo \in \domain \functionone}}{\ctxone}{\typtwo} $
\item $\image \functionone \subseteq \{1, \cdots,n\}, \text{ and }\wfjt{}{\termone_{\functionone (\vartwo)}}{\typthree_{\vartwo}}$ for every $\vartwo \in \domain \functionone $
\end{itemize}

\hide{
\item For any set of possible open terms $\ctxone_1,...,\ctxone_n$ we say : ${\wfjt {\varone_1 : \typone_1, \cdots, \varone_m : \typone_m}{\ctxone_1}{\typtwo_1}}, \cdots \ctxone_n : \typtwo_n$ if there exists a partition of the typing context ${\varone_1 : \typone_1, \cdots \varone_m : \typone_m}$ in $n$ disjoints typing context $\Gamma_i$ such that $\wfjt {\Gamma_i}{\ctxone_i}{\typtwo_i}$ for every $1 \leq i \leq n$.}
\item  We define the set of pairs of context and tuples which are well formed :
$$
\ctwellformedt \typone = \{(\ctxone, \functionone, \tuple {\termone_1}{\cdots}{\termone_n}) \mid \wfjt \emptyset{(\ctxone, \functionone, \tuple {\termone_1}{\cdots}{\termone_n})} \typone \} 
$$
%\item We define a subset of $$\ctwellformed$$ which correponds, for a given context, to minimal tuples :
%$$\ctwellformedmin = \{(\ctxone, \tuple {\termone_1}{\cdots}{\termone_n}) \mid \,\left(\wfjt {}{\termone_i}{\typone_i} \right)_{1 \leq i \leq n} \text{ and } \wfj {x_i : \typone_i} \ctxone \and \not\exist \{l_1,...,l_k\} \subseteq\{1,..,n\} \text{ such that }  \wfj {(x_{l_j} : \typone_{l_j})_{1\leq j \leq k}} \ctxone \} $$
\item We define a notion of congruence for elements in $\ctwellformed$ :
For every permutation $\sigma : \{1,..,n\} \rightarrow \{1,...,n\}$, $(\ctxone, \functionone \tuple {\termone_1}{\cdots}{\termone_n}) \equiv (\ctxone, \compose{\sigma^{-1}}\functionone , \tuple {\termone_{\sigma(1)}}{\cdots}{\termone_{\sigma(n)}}) $
\end{itemize}
\end{definition}

In the following, we consider equivalence class of $\equiv$. It corresponds to reorder elements of the tuple, and to modify the function $\functionone$ in order to have still the same mapping from the free variables of $\ctxone$.

We define a small-step semantics on elements of $\ctwellformed$.

\begin{figure}[!h]  
{\small
\hide{
$$\AxiomC{$\sssp \termtwo \distrtwo $}
  \UnaryInfC{$\wtr{\distrone \disjplus p \cdot \dirac{(\termtwo,\functionone, \tuplone)}}{\emptytr}{\distrone + p\cdot (\distrtwo,\functionone, \tuplone) }$} 
  \DisplayProof \qquad
  \AxiomC{$ $}
  \UnaryInfC{$\wtr{p \cdot \dirac{(\abstr\varone\termtwo,\functionone, \tuplone)}}{\app\valone}{p\cdot (\subst\termtwo\varone \valone, \functionone,\tuplone) }$}
  \DisplayProof 
  $$
  }
    $$
  \AxiomC{$\wtr {\dirac{(\ctxone, \functionone, \tuple {\torvone_1}{\cdots}{\torvone_n})}} \emptytr \distrtwo $}
  \AxiomC{$
    \begin{array}{c}
      \wfj{}{(\ctxtwo,\functiontwo, \tuple{\torvtwo_1}{\cdots}{\torvtwo_p})}\\ 
      \wfj{}{(\ctxone,\functionone, \tuple{\torvone_1}{\cdots}{\torvone_n})}
      \end{array}
      $}
  \AxiomC{$\domain \functionone \cap \domain \functiontwo = \emptyset $}
  \TrinaryInfC{$\wtr{\distrone \disjplus p \cdot \dirac{(\ctxone\ctxtwo, \parall n \functionone \functiontwo  \tuple {\torvone_1,\cdots, \torvone_n}{\torvtwo_1}{\cdots \torvtwo_p})}}{\emptytr}{\distrone + p\cdot \sum \distrtwo(\ctxthree, \functionthree, \tuple{\torvthree_1}{\cdots}{\torvthree_q}) \cdot (\ctxthree\ctxtwo, \parall q \functionthree \functiontwo, \tuple  {\torvthree_1,\cdots, \torvthree_q}{\torvtwo_1}{\cdots \torvtwo_p}) }$}
  \DisplayProof $$  \\ 
$$
  \AxiomC{$\wtr {\dirac{(\ctxtwo, \functiontwo, \tuple{\torvtwo_1}{\cdots}{\torvtwo_p})}}{\emptytr}{\distrtwo}$}
  \AxiomC{$
\begin{array}{c}
\wfjval{}{(\ctxone, \functionone, \tuple{\torvone_1}{\cdots}{\torvone_n})}  \\ 
\wfj{}{(\ctxtwo, \functiontwo, \tuple{\torvtwo_1}{\cdots}{\torvtwo_p})} \\ \domain \functionone \cap \domain \functiontwo = \emptyset 
\end{array}$}
  %\AxiomC{$ \subst\ctxone \varone {\termone_{\phi{(\varone)}}}  \text{ is a value } $}
  \BinaryInfC{$\wtr{\distrone \disjplus p\cdot { \dirac{(\ctxone\ctxtwo, \parall n \functionone \functiontwo  \tuple {\torvone_1,\cdots, \torvone_n}{\torvtwo_1}{\cdots \torvtwo_p})}}}{\emptytr}{\distrone + p\cdot \sum \distrtwo(\ctxthree, \functionthree, \tuple{\torvthree_1}{\cdots}{\torvthree_q}) \cdot (\ctxone\ctxthree, \parall n \functionone \functionthree, \tuple  {\torvone_1,\cdots, \torvone_n}{\torvthree_1}{\cdots \torvthree_q})  }$}
  \DisplayProof 
 $$
 \\
  
 $$
 \AxiomC{$\focusindex{\tuple{\termone}{\torvone_1, \cdots}{\torvone_n}} = \infty$}
\UnaryInfC{$\wtr{\distrone \disjplus p\cdot \dirac{\varone, \parall 1 {\{\varone \rightarrow 1\}}\functionone, \tuple{\termone, \torvone_1}{\cdots}{\torvone_n}}}{\emptytr}{\distrone + p\cdot (\varone, {\parall 1  {\{\varone \rightarrow 1\}}\functionone}, \tuple{\focus 1 \termone}{\torvone_2}{\cdots, \torvone_n}})$}
\DisplayProof $$ 
 $$
 \AxiomC{$\sssp {\termone}  \distrtwo$}
\AxiomC{$\focusindex {\tuple{\focus i \termone}{\torvone_1, \cdots}{\torvone_n}} = i $}
\BinaryInfC{$\wtr{\distrone \disjplus p\cdot \dirac{\varone, \parall 1 {\{\varone \rightarrow 1\}}\functionone, \tuple{\focus i \termone, \torvone_1}{\cdots}{\torvone_n}}}{\emptytr}{\distrone + p\cdot (\varone, {\parall 1  {\{\varone \rightarrow 1\}}\functionone}, \tuple{\focus i \distrtwo}{\torvone_1}{\cdots, \torvone_n}})$}
\DisplayProof $$
 $$
\AxiomC{$\focusindex {\tuple{\focus i \valone}{\torvone_1, \cdots}{\torvone_n}} = i $}
\UnaryInfC{$\wtr{\distrone \disjplus p\cdot \dirac{\varone, \parall 1 {\{\varone \rightarrow 1\}}\functionone, \tuple{\focus i \valone, \torvone_1}{\cdots}{\torvone_n}}}{\emptytr}{\distrone + p\cdot (\varone, {\parall 1  {\{\varone \rightarrow 1\}}\functionone}, \tuple{\dval \valone}{\torvone_1}{\cdots, \torvone_n}})$}
\DisplayProof $$
\\ 
$$
\AxiomC{$\wfjval{}{(\ctxone, \functionone, \tuple{\torvone_1}{\cdots}{\torvone_q})} $}
\AxiomC{$\valone = \subst {\ctxone}{\varthree}{\termone_{\functionone{(\varthree)}}}_{\varthree \in {\freevar \ctxone}}$}
\BinaryInfC{$\wtr{\distrone \disjplus p \cdot\dirac{\varone\ctxone,\parall q {\parall 1 {(\varone \rightarrow 1)}\functionone} \functiontwo, \tuple{\dval{\abstr\vartwo \termthree},\torvone_1 \cdots \torvone_q}{\torvtwo_1\cdots}{\torvtwo_n}}}{\emptytr}{\distrone + p \cdot\dirac {\varone,\parall 1 {(\varone \rightarrow 1)} \functiontwo,\tuple {\focus 1{\subst \termthree \varone {\valone}}}{\torvtwo_1,\cdots}{\torvtwo_n}}}$}
\DisplayProof$$
$$
\AxiomC{$
\begin{array}{c}
\wfj{}{(\ctxone, \functionone, \tuple{\torvone_1}{\cdots}{\torvone_n})}  \\ 
\wfjval{}{(\ctxtwo, \functiontwo, \tuple{\torvtwo_1}{\cdots}{\torvtwo_p})} \\ \domain \functionone \cap \domain \functiontwo = \emptyset 
\end{array}$}
\UnaryInfC{$\wtr{\distrone \disjplus p \cdot\dirac{(\abstr\varone \ctxone)\ctxtwo, \parall n \functionone \functiontwo, \tuple{\torvone_1, \cdots \torvone_n, \torvtwo_1}{\cdots}{\torvtwo_p}}}{\emptytr}{\distrone + p \cdot \dirac {\subst \ctxone \varone \ctxtwo, \parall n \functionone \functiontwo, \tuple{\torvone_1, \cdots \torvone_n, \torvtwo_1}{\cdots}{\torvtwo_p}}}$}
\DisplayProof
$$ $$
 \AxiomC{}
 \UnaryInfC{$\wtr {p \cdot \dirac{\abstr \varone \ctxone,\functionone, \tuple{\torvone_1}{\cdots}{\torvone_n}}}{\app \valone}{p \cdot\dirac{({\subst\ctxone \varone \valone},\functionone, \tuple{\torvone_1}{\cdots}{\torvone_n})}}  $} 
\DisplayProof
$$ $$
\AxiomC{$ $}
\UnaryInfC{$\wtr{p \cdot \dirac{(\varone,\parall 1 {(\varone \rightarrow 1)}\functionone, \tuple{\dval{\abstr \vartwo\termtwo}, \torvone_1}{\cdots}{\torvone_n})}}{\app \valone}{p \cdot \dirac {(\varone, \parall 1 {(\varone \rightarrow 1)}{\functionone} ,\tuple{\focus 1 {\subst \termtwo \vartwo \valone}}{\torvone_1}{\cdots, \torvone_n}}}$}
\DisplayProof
$$
 $$
\AxiomC{$\wtr {\distrone_i}{\app \valone}{\distrtwo_i} $}
\AxiomC{$\forall i, \distrone_i \text{ value distribution } $}
\AxiomC{$\distrthree$ stopped distribution}
\TrinaryInfC{$\wtr {\stackrel{\cdot}{\sum_i} \distrone_i \dotplus \distrthree} {\app \valone} {\sum_i{\distrtwo_i}} $}
\DisplayProof$$ 
$$\AxiomC{}
  \UnaryInfC{$\wtrc{\distrone}{\emptytr}{\distrone}$}
  \DisplayProof 
\qquad 
\AxiomC{$\wtr \distrone \traceone \distrtwo $}
\AxiomC{$\wtrc \distrtwo \tracetwo \distrthree $}
\BinaryInfC{$\wtrc  {\distrone}{\concat \traceone \tracetwo} \distrthree$}
\DisplayProof
$$
}
  \caption{small-step trace relation on distributions over $\ctwellformed$ (without pairs)}\label{sstrc}
  \end{figure}
Please observe that the rules would be exactly the same for an strictly linear (that is, not affine) calculus. The only thing to change would be the definition of : $\wfj{\Gamma}{(\ctxone, \functionone, \tuplone)} $.

{We need a definition of $\ctwellformed$ taking into account possible other free variables}
We need to add rules specific for the language with pairs :
\begin{figure}
{\small
 $$
  \AxiomC{$
    \begin{array}{c}
      \wfjt{}{{(\ctxone, \functionone, \tuple {\termone_1}{\cdots}{\termone_n})}}{\tprod \typone \typtwo}\\
      \wfjt{\varone : \typone, \vartwo : \typtwo}{(\ctxtwo, \functiontwo, \tuple{\termtwo_1}{\cdots}{\termtwo_p })}{\typthree}\\
        \domain \functionone \cap(\domain\functiontwo \cup\{\varone, \vartwo\}) = \emptyset
      \end{array}
    $}
  \AxiomC{$\wtr {\dirac{(\ctxone, \functionone, \tuple {\termone_1}{\cdots}{\termone_n})}} \emptytr \distrtwo $}
  \BinaryInfC{$
\begin{array}{c}
\wtr{\distrone \disjplus p \cdot \dirac{(\letin \varone \vartwo\ctxone\ctxtwo, \parall n \functionone \functiontwo  \tuple {\termone_1,\cdots, \termone_n}{\termtwo_1}{\cdots \termtwo_p})}}{\emptytr} {} \\ {\distrone + p\cdot \sum \distrtwo(\ctxthree, \functionthree, \tuple{\termthree_1}{\cdots}{\termthree_q}) \cdot (\letin \varone\vartwo \ctxthree\ctxtwo, \parall q \functionthree \functiontwo, \tuple  {\termthree_1,\cdots, \termthree_q}{\termtwo_1}{\cdots \termtwo_p}) }
\end{array}
$}
  \DisplayProof $$ \\

$$
  \AxiomC{$\wtr {\dirac{(\ctxone, \functionone, \tuple {\termone_1}{\cdots}{\termone_n})}} \emptytr \distrtwo $}
  \AxiomC{$
\begin{array}{c}
 \domain \functionone , \domain \functiontwo, \{\varone, \vartwo\}, \domain \functionthree \\ \text{ disjoints sets }
\end{array}
 $}
  \AxiomC{$
    \begin{array}{c}
      \wfjt{}{\ctxone, \functionone, \tuple{\termone_1}{\cdots}{\termone_n}}{\typone}\\
       \wfjt{}{\ctxtwo, \functiontwo, \tuple{\termtwo_1}{\cdots}{\termtwo_m}}{\typtwo}\\
       \wfjt{\varone : \typone, \vartwo : \typtwo}{(\ctxthree,\functionthree, \tuple{\termthree_1}{\cdots}{\termthree_l} }{\typthree}
      \end{array}
    $}
  \TrinaryInfC{$
\begin{array}{c}
\wtr{\distrone \disjplus p \cdot \dirac{(\letin \varone \vartwo {\pair\ctxone\ctxtwo}\ctxthree, \parall n{ \functionone}{\parall m \functiontwo \functionthree},  \tuple {\termone_1,\cdots, \termone_n}{\termtwo_1}{\cdots \termtwo_m, \termthree_1 \cdots \termthree_l})}}{\emptytr} {} \\ {\distrone + p\cdot \sum \distrtwo(\ctxfour, \functionfour, \tuple{\termfour_1}{\cdots}{\termfour_q}) \cdot (\letin \varone\vartwo {\pair\ctxfour\ctxtwo}\ctxthree, \parall q \functionfour {\parall m \functiontwo\functionthree}, \tuple  {\termfour_1,\cdots, \termfour_q}{\termtwo_1}{\cdots \termtwo_m, \termthree_1, \cdots \termthree_l}) }
\end{array}
$}
  \DisplayProof $$ 

$$
  \AxiomC{$\wtr {\dirac{(\ctxtwo, \functiontwo, \tuple {\termtwo_1}{\cdots}{\termtwo_m})}} \emptytr \distrtwo $}
  \AxiomC{$\subst \ctxone{\varone}{\termone_{\functionone{(\varone)}}}_{\varone \in \domain \functionone} \text{ value } $}
  \AxiomC{$
    \begin{array}{c}
      \wfjt{}{\ctxone, \functionone, \tuple{\termone_1}{\cdots}{\termone_n}}{\typone}\\
       \wfjt{}{\ctxtwo, \functiontwo, \tuple{\termtwo_1}{\cdots}{\termtwo_m}}{\typtwo}\\
       \wfjt{\varone : \typone, \vartwo : \typtwo}{(\ctxthree,\functionthree, \tuple{\termthree_1}{\cdots}{\termthree_l} }{\typthree}\\
       \domain \functionone , \domain \functiontwo, \{\varone, \vartwo\}, \domain \functionthree  \text{ disjoints sets }
      \end{array}
    $}
  \TrinaryInfC{$
\begin{array}{c}
\wtr{\distrone \disjplus p \cdot \dirac{(\letin \varone \vartwo {\pair\ctxone\ctxtwo}\ctxthree, \parall n \functionone  {\parall m \functiontwo \functionthree}  \tuple {\termone_1,\cdots, \termone_n}{\termtwo_1}{\cdots \termtwo_m, \termthree_1 \cdots \termthree_l})}}{\emptytr} {} \\ {\distrone + p\cdot \sum \distrtwo(\ctxfour, \functionfour, \tuple{\termfour_1}{\cdots}{\termfour_q}) \cdot \dirac{(\letin \varone\vartwo {\pair\ctxfour\ctxtwo}\ctxthree, \parall n \functionone {\parall q \functionfour\functionthree}, \tuple  {\termone_1,\cdots, \termone_n}{\termfour_1}{\cdots \termfour_q, \termthree_1, \cdots \termthree_l}) }}
\end{array}
$}
  \DisplayProof $$ 

$$
  \AxiomC{$
    \begin{array}{c}
    \subst \ctxone{\varthree}{\termone_{\functionone{(\varthree)}}}_{\varthree \in \domain \functionone} \text{ value } \\
     \subst \ctxtwo{\varthree}{\termtwo_{\functiontwo{(\varthree)}}}_{\varthree \in \domain \functiontwo} \text{ value } \\
  \end{array}$}
  \AxiomC{$
    \begin{array}{c}
      \wfjt{}{\ctxone, \functionone, \tuple{\termone_1}{\cdots}{\termone_n}}{\typone}\\
       \wfjt{}{\ctxtwo, \functiontwo, \tuple{\termtwo_1}{\cdots}{\termtwo_m}}{\typtwo}\\
       \wfjt{\varone : \typone, \vartwo : \typtwo}{(\ctxthree,\functionthree, \tuple{\termthree_1}{\cdots}{\termthree_l} }{\typthree}\\
       \domain \functionone , \domain \functiontwo, \{\varone, \vartwo\}, \domain \functionthree  \text{ disjoints sets }
      \end{array}
    $}
  \BinaryInfC{$
\begin{array}{c}
\wtr{\distrone \disjplus p \cdot \dirac{(\letin \varone \vartwo {\pair\ctxone\ctxtwo}\ctxthree, \parall n \functionone {\parall m \functiontwo\functionthree}  \tuple {\termone_1,\cdots, \termone_n}{\termtwo_1}{\cdots \termtwo_m, \termthree_1 \cdots \termthree_l})}}{\emptytr} {} \\ {\distrone + p\cdot \dirac{( \subst{\subst\ctxthree \varone \ctxone}\vartwo \ctxtwo, \parall n \functionone {\parall m \functiontwo\functionthree}, \tuple  {\termone_1,\cdots, \termone_n}{\termtwo_1}{\cdots \termtwo_m, \termthree_1, \cdots \termthree_l}) }}
\end{array}
$}
  \DisplayProof $$ 

$$
  \AxiomC{$
    {\termone_1 = \pair \termtwo \termthree}  
  $}
  \AxiomC{$
    \begin{array}{c}
       \wfjt{\varone : \typone, \vartwo : \typtwo}{(\ctxone,\functionone, \tuple{\termone_2}{\cdots}{\termone_n})}{\typthree}\\ 
       \varthree \not \in \domain \functionone
      \end{array}
    $}
  \BinaryInfC{$
\begin{array}{c}
\wtr{\distrone \disjplus p \cdot \dirac{(\letin \varone \vartwo {\varthree}\ctxone, \parall 1 {\{\varthree \rightarrow 1\}} \functionone , \tuple {\termone_1} \cdots {\termone_n})}} {\emptytr} {} \\ {\distrone + p\cdot \dirac{( \ctxone , \parall 2 {\{\varone \rightarrow 1, \vartwo \rightarrow 2\}}\functionone , \tuple {\termtwo, \termthree}{\termone_2}{\cdots \termone_n})}}
\end{array}
$}
  \DisplayProof $$ 

$$
  \AxiomC{$
    {\termone_1 = \pair \termtwo \termthree}  
  $}
  \UnaryInfC{$
\begin{array}{c}
\wtr{p \cdot \dirac{({\varthree}, \parall 1 {\{\varthree \rightarrow 1\}} \functionone , \tuple {\termone_1} \cdots {\termone_n})}} {\actpair \ctxone} {} \\ { p\cdot \dirac{(\ctxone , \parall 2 {\{\varone \rightarrow 1, \vartwo \rightarrow 2\}}\functionone , \tuple {\termtwo, \termthree}{\termone_2}{\cdots \termone_n})}}
\end{array}
$}
  \DisplayProof $$ 

$$
  \AxiomC{$
    {\termone_1 = \pair \termtwo \termthree}  
  $}
  \UnaryInfC{$
\begin{array}{c}
\wtr{p \cdot \dirac{({\pair \ctxtwo \ctxthree}, \functionone , \tuple {\termone_1} \cdots {\termone_n})}} {\actpair \ctxone} {} \\ { p\cdot \dirac{(\subst{\subst\ctxone\varone\ctxtwo}\vartwo\ctxthree , \functionone , \tuple{\termone_1}{\cdots}{ \termone_n})}}
\end{array}
$}
  \DisplayProof $$ 
}

\caption{small-step trace relation on distributions over $\ctwellformed$ for pairs}
\end{figure}

Please observe that there is two different non-determinism in the rules : the choice of the part of the distribution which is going to be reduced, and the way the tuple is divided (for the affine case).  The second one is not really meaningful, since we have the following lemma :
\begin{lemma}
Suppose that : $\wfj{}{(\ctxone, \functionone, \tuplone)}$, and let be $\distrone$, $\distrtwo$ such that $\wtr{\dirac{(\ctxone, \functionone, \tuplone)}}{\actone}{\distrone}$ and $\wtr{\dirac{(\ctxone, \functionone, \tuplone)}}{\actone}{\distrtwo}$. Then $ \distrone \equiv \distrtwo$.
\end{lemma}
\begin{proof}
Let be  $\wfj{}{(\ctxone, \functionone, \tuplone)}$. \\ 
We are first going to show the following result :
Suppose that $\tuplone = \tuple{\termone_1, \cdots}{\termone_n, \cdots}{\termone_q} $, et that $\functionone(\freevar \ctxone) \subseteq \{1,\cdots,n\}$. Then let be $\distrone$ such that : $\wtr {(\ctxone, \functionone,\tuplone )}{\actone}{\distrone}$.
Then there exist $\distrtwo$ such that : $\wtr {(\ctxone, \restr \functionone{\freevar \ctxone},\tuple{\termone_1}{\cdots}{\termone_n} }{\actone}{\distrtwo}$, and $\distrone = \sum \distrtwo (\ctxtwo, \functiontwo, \tuple{\termtwo_1}{\cdots}{\termtwo_p})\cdot \dirac{(\ctxtwo, \functionthree, \tuple{\termtwo_1, \cdots \termtwo_p}{\termone_{n+1}}{\cdots, \termone_n})}$, with $\functionthree(\varone) = \functiontwo(\varone)$ if $\varone \in \freevar \ctxone$, and $\functionthree(\varone) = p -n + \functionone(\varone) $ otherwise.
We show that by induction on the derivation of  $\wtr {(\ctxone, \functionone,\tuplone )}{\actone}{\distrone}$.\\
Then it is sufficient to remark that, if the free variables of $\ctxone$ correspond exactly to the terms in the tuple, there is only one possible rule that can be applied.
\end{proof}

\begin{definition}
Let be $\distrone$ a distribution over $\ctwellformed$.
We define $\forget \distrone$ a distribution over closed terms by :
$\forget \distrone = \sum \distrone(\ctxone, \functionone, \tuple{\termone_1}{\cdots}{\termone_n}) \cdot \dirac {\subst{\ctxone}{\varone}{\termone_{\functionone(\varone)}} _{\varone\in\domain \functionone}}$
\end{definition}
We would like to know that, if a distribution on terms can do a trace, then the correponding distribution where we split contexts and terms filling them can do the same trace. Unfortunately, we need to be more precis  on how we split the distribution, and especially on what focus we can have on the components of the tuple. (For example, $\left(\varone, (\varone \rightarrow 1) ,\tuplonea{\termone, \focus 1\termtwo}\right) \not \rightarrow$, since it is not possible to evaluate $\termone$ before having evaluated $\termtwo$.) So we define a notion of coherent tuples in $\setonewf$ for a given context, where the idea is : This context could have triggered the evaluation on the terms which are under focus : 

\begin{lemma}
Let be $\distrone$ a distribution over $\ctwellformed$, and $\traceone$ a trace such that :
$\wtrc {\forget \distrone} \traceone \distrtwo$.
Then there exists $\distrthree$ such that
$\wtrc {\distrone}{\traceone}{\distrthree}$, and $\forget \distrthree = \distrtwo$.
\end{lemma}

The rules of the trace semantics for elements in $\ctwellformed$ are designed to match the one for trace semantics for terms. More precisely, it means that :
\begin{lemma}
Let be $\wfj{}{(\ctxone,\functionone, \tuple{\termone_1}{\cdots}{\termone_n})}$, and let be $\distrone$ such that :
$\wtr{\dirac{(\ctxone, \functionone, \tuple{\termone_1}{\cdots}{\termone_n})}}{\emptytr}{\distrone}$. 
Then $\wtr {\dirac{\subst \ctxone \varone {\termone_{\functionone{\varone}}}_{\varone \in \freevar \ctxone}}} {\emptytr}{\forget \distrone}$ 
\end{lemma}
\begin{proof}
The proof is by case analysis of the derivation of $\wtr{\dirac{(\ctxone, \functionone, \tuple{\termone_1}{\cdots}{\termone_n})}}{\emptytr}{\distrone}$ 
\end{proof}

\begin{lemma}
Let be $\distrone$ a distribution over $\ctwellformed$, and $\traceone$ a trace. 
Suppose that $\wtrc \distrone \traceone \distrtwo$. 
Then $\wtrc {\forget \distrone} \traceone {\forget {\normal\distrtwo}}$
\end{lemma}
\begin{proof}
It is in fact sufficient to show :
\begin{itemize}
\item If $\wtr \distrone \emptytr \distrtwo$, then there exist $\distrthree$, such that $\wtrc  \distrtwo \emptytr \distrthree$, and  $\wtr {\forget \distrone}{\emptytr}{\forget \distrthree}$.   No matter the last rule used in the derivation of $\wtr \distrone \actone \distrtwo$, it is of the form :
$\wtr {\distrone = \distrfour \disjplus p\cdot (\ctxone, \functionone, \tuplone) }{\emptytr}{\distrfour + p \cdot \distrfive}$ with $\wtr{\dirac{(\ctxone,\functionone, \tuplone)}}{\emptytr}{\distrfive}$.
Now we have to consider all the possible $(\ctxtwo, \functiontwo, \tupltwo) \in \supp \distrfour$ such that $\subst \ctxtwo {\varone}{\tupltwo_{\functiontwo(\varone)}}_{\varone \in \freevar \ctxtwo} = \subst \ctxone {\varone}{\tuplone_{\functiontwo(\varone)}}_{\varone \in \freevar \ctxone}  $

\item and : if $\wtr \distrone \actone \distrtwo$, then $\wtr {\forget \distrone} \actone {\forget \distrtwo}$
\end{itemize}
\end{proof}

\subsubsection{Link beetween trace semantics on terms and trace semantics on $\ctwellformed$.}

\begin{definition}
Let be $\distrone$ and $\distrtwo$ two distributions over $\ctwellformed$.
For $\epsone \geq 0$, we say that $\distrone$ and $\distrtwo$ are $\epsone$-related if :
there exist $p_i,...,p_m$ positive reals, and $\ctxtwo_1,...,\ctxtwo_d$ distincts contexts, and $\distrthree_1,...,\distrthree_d$, $\distrfour_1,...,\distrfour_d$ distributions on tuples such that :
\begin{align*}
&\distrone = \sum_j p_j \cdot (\ctxtwo_j, \distrthree_j)\\
&\distrtwo = \sum_j p_j \cdot (\ctxtwo_j, \distrfour_j)\\
&\appl \metrtrtupl {\distrthree_j}{\distrfour_j} \leq \epsone
\end{align*} 
\end{definition}

\begin{lemma}
The relation $\wtr {\cdot}{\emptytr}{\cdot}$ on distributions over $\ctwellformed$ is strongly normalizing.
\end{lemma}

\begin{lemma}
Let be $\distrone$, $\distrtwo$ two $\epsone$-related distributions.
Then $\normal \distrone$ and $\normal \distrtwo$ are related.
\end{lemma}

\begin{lemma}\label{epsonereltupl}
Let be $\distrone$, $\distrtwo$ two $\epsone$-related distribution.
Let be $\normal \distrthree$, and $\normal \distrfour$ in normal form such that : $\wtrc \distrone \traceone {\normal \distrthree}$, and $\wtrc \distrtwo \traceone {\normal \distrfour}$. Then $\distrthree$ and $\distrfour$ are $\epsone$-related
\end{lemma}
Theorem \ref{nonexpansivtupl} is deduced of Lemma \ref{epsonereltupl} in a similar way as for the trace distance.

%%%%%%%%%%%%%%%%%%%%%%%%%%%%%%%%%%%%%%%%%%%%%%%%%%%%%%
%End of proof of non-expansiveness for tuple
%%%%%%%%%%%%%%%%%%%%%%%%%%%%%%%%%%%%%%%%%%%%%%%%%%%%%%

\end{proof}

\condskip
%J'ai essaye de reformuler la phrase qui etait apres le theorem 7. 
%Si tu trouves que ce n'est toujours pas clair, on peut simplement l'enlever.
Theorem \ref{nonexpansivtupl} can be read as a non-expansiveness
result: if we have a system $\mathcal E$, playing the role of the
environment, and which is prepared to interact with $n$ components,
and moreover we have two tuples $\tuplone$ and $\tupltwo$ of length
$n$, then the tuple distance between $\tuplone$ and $\tupltwo$
gives us an upper bound on the trace distance between the system
composed of $\mathcal E$ interacting with $\tuplone$, and the system
composed of $\mathcal E$ interacting with $\tupltwo$.

We can now see that $\metrtrtupl$ coincides with the context metric:
one inequality comes from Theorem \ref{nonexpansivtupl}, the other
comes from the fact that any trace $\traceone$ over $\labelstupl$ and
designed to start from a single value, can be simulated by a context.
\condskip
\begin{theorem}
On programs, $\metrtrtupl=\metrctx$ 
\end{theorem}
\begin{proof}
\begin{itemize}
\item We apply Theorem \ref{nonexpansivtupl} to $\abstr \varone \termone$ and $\abstr \varone \termtwo$, which are values, and the context $\ctxone = \hole$:
\begin{align*}
\appl \metrtr{\termone}{\termtwo} = 
\appl\metrtr{\abstr \varone \termone}{\abstr \varone \termtwo} \\
& \leq \appl\metrtrtupl{\tuplonea {\abstr \varone \termone}}{\abstr \varone \termtwo}\\& = \appl \metrtrtupl{\termone}{\termtwo}
\end{align*}
\item Let be $\traceone$ a trace in the LMC $\markovtupl$ which starts from a single value. Then we can find a context that simulate this trace.
\end{itemize}
\end{proof}
\condskip

%%%%%%%%%%%%%%%%%%%%
\subsection{Examples}
%%%%%%%%%%%%%%%%%%%%
The tuple distance, that we have just proved to be fully-abstract,
can be seen as yet another presentation of the context distance. But
there is much more: it allows to evaluate the distance between
concrete programs, even when the latter contains pairs, in a
relatively easy way. In this section, we will give two examples.

%%%%%%%%%%%%%%%%%%%%%%%%%%%%%%%%
\subsubsection{A Simple Example}
%%%%%%%%%%%%%%%%%%%%%%%%%%%%%%%%
Consider the terms $\termone$ and $\termtwo$ defined in Example
\ref{expair}.  We can prove that $\appl\metrtrtupl{\termone}{\termtwo}
= \frac 3 4$.  
We are first going to show that
  $\appl\metrtrtupl{\termone}{\termtwo} \geq \frac 3 4$. In order to
show that, we are going to present a particular trace $\traceone$ such
that $\abs{ {\probtrtupl {\tuplonea \termone} \traceone} -
  {\probtrtupl {\tuplonea \termtwo} \traceone}} = \frac 3 4$. More
precisely, we take $\traceone = \concat{\actcut 1}{\concat{\actappl{
      (\emcon, \identity)} 1}{\actappl{ (\emcon, \identity)} 2}}$: it
corresponds to first separating the two components of the pair, and then
passing $\identity$ as an argument to the first and to the
second component. The relevant fragment of $\markovtupl$ can be found in Figure
\ref{fig:lmctupleex}. In particular, we can see that $\probtr{\tuplonea
  \termone}{\traceone} = 1 $, and $\probtr {\tuplonea \termtwo}
\traceone = \frac 1 4$.
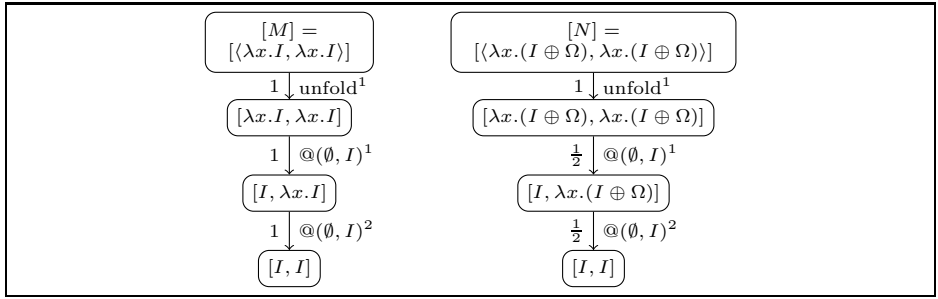
\begin{figure}[!h]
\begin{center}
\fbox{
\begin{minipage}{\figurewidth}
\begin{center}
\begin{tikzpicture}[auto]
\node [draw, rectangle, rounded corners] (M) at (0,0) {\scriptsize $
\begin{array}{c}\tuplonea\termone = \\ \tuplonea{\pair{\abstr \varone\identity}{\abstr \varone \identity}} \end{array}$};
\node[draw, rectangle, rounded corners] (N) at (4,0) {\scriptsize$\begin{array}{c}\tuplonea\termtwo = \\ \tuplonea{\pair{\abstr \varone (\psum \identity \diver)}{\abstr \varone (\psum \identity \diver)}}\end{array}$};
\node [draw, rectangle, rounded corners] (C) at (0,-1) {\scriptsize$\tuplonea{{\abstr \varone\identity},{\abstr \varone \identity}} $};
\node[draw, rectangle, rounded corners] (D) at (4,-1) {\scriptsize$\tuplonea{{\abstr \varone (\psum \identity \diver)},{\abstr \varone (\psum \identity \diver)}}$};
\node[draw, rectangle, rounded corners] (E) at (0,-2){\scriptsize$\tuplonea{\identity, \abstr \varone \identity}$};
\node[draw, rectangle, rounded corners] (F) at (4,-2){\scriptsize$\tuplonea{\identity, \abstr \varone {(\psum \identity\diver)}}$};
\node[draw, rectangle, rounded corners] (G) at (0,-3){\scriptsize$\tuplonea{\identity, \identity}$};
\node[draw, rectangle, rounded corners] (H) at (4,-3){\scriptsize$\tuplonea{\identity, \identity}$};
%\node [draw, circle](A) at (0,-1){};
%\node [draw, circle](B) at (4,-1){};
\draw[->](M) to node[right] {\scriptsize $\actcut 1$} node[left] {\scriptsize $1$} (C);
\draw[->](N) to node[right] {\scriptsize $\actcut 1$}node[left]{\scriptsize $1$}(D);
\draw[->](C) to node[right] {\scriptsize $\actappl{ (\emcon, \identity)} 1$} node[left]{\scriptsize $ 1$} (E);
\draw[->](D) to node[right] {\scriptsize $\actappl{ (\emcon, \identity)} 1$} node [left]{\scriptsize $\frac 1 2$} (F);
\draw[->](E) to node [right]{\scriptsize $\actappl{ (\emcon, \identity)} 2$} node[left]{\scriptsize $1$}(G);
\draw[->](F) to node [right]{\scriptsize $\actappl{ (\emcon, \identity)} 2$}node[left]{\scriptsize $\frac 1 2$}(H);

%\draw[->](A) to node {$1 $}  (C);
%\draw[->](B) to node {$\frac 1 2 $}  (D);
%\draw[->](T)[loop left] to node {$\text{true} $}  (T);
%\draw[->](F) [loop right]to node {$\text{false} $}  (F);
\end{tikzpicture}
\end{center}\end{minipage}
}\end{center}
\caption{The relevant fragment of the tuple LMC}\label{fig:lmctupleex}
\end{figure}
 Now we want to show the reverse inequality, namely that 
  $\appl\metrtrtupl{\termone}{\termtwo} \leq \frac 3 4$. For that, we are going to use the alternative characterisation of trace distance : it is sufficient to find a $\frac 3 4$-bisimulation $\relone$ on the LTS of distributions such that $ (\dirac {\tuplonea\termone}), (\dirac{\tuplonea \termtwo}) \in \relone$

%%%%%%%%%%%%%%%%%%%%%%%%%%%%%%%%%%%%%%%%%%%
\subsubsection{A More Complicated Example}
%%%%%%%%%%%%%%%%%%%%%%%%%%%%%%%%%%%%%%%%%%
Please remember the example we presented in Section~\ref{sect:anatomy}. We
note $\{u_n\}_{n \in \NN}$ the sequence defined as: $u_n = \prod_{1
  \leq i \leq n}{(1 - \frac 1 {2^i} )}$.
Please observe that the sequence $(u_n)_{n\in \NN} $ has a limit strictly between $0$ and $1$.
\condskip
\begin{lemma}
La suite $(\suitone n)_{n \in \NN}$ has a limit $\limit$, and $\frac 1 2 > \limit > 0$
\end{lemma}
\begin{proof}
\begin{itemize}
\item $\suitone n$ is a decreasing and bounded sequence : it has a limit.
\item $\limit > \suitone 1 = \frac 1 2$
\item We consider the sequence : $\suittwo n = \log {\suitone n} = \sum_{1 \leq i \leq n} {\log{(1 - {(\frac 1 2)}^i)}}$. We pose $\suitthree n =  {\log{1 - {(\frac 1 2)}^i}}$. Then we consider $$\abs {\frac{\suitthree n}{\suitthree {n+1}}} = \abs{\frac { {\log{(1 - {(\frac 1 2)}^n)}}}{ {\log{(1 - {(\frac 1 2)}^{n+1})}}}} \rightarrow_{n \rightarrow \infty} \frac 1 2$$
D'Alembert's theorem for infinite sum implies that the serie is convergent and has a finite limit. 
\end{itemize}
\end{proof}
\condskip
\begin{theorem}
For every $n \in \NN$, $\appl \metrtrtupl {\termone_n}{\termtwo_n} = 1 - u_n.$
\end{theorem}
% u_n has a limit strictly between $0$ and $1$.
% dans l'absolu c'est interessant, parce que ca montre que 
% meme si on ajoute toujours plus de differences entre les deux termes a chaque etape,
% on n'atteint tout de meme pas une distance 1.  
%: ce serait bien d'écrire quelque chose pour expliquer 
%ce que ca veut dire pour les termes
\begin{proof}
We first show that $\appl \metrtrtupl {\termone_n}{\termtwo_n} \geq 1
- u_n$. As in the previous example, we do that by finding, for each $n
\in \NN$, a trace $\traceone_n$ such that $\abs{{\probtr
    {\tuplonea{\termone_n}}{\traceone_n}} -
  {\probtr{\tuplonea{\termtwo_n}}{\traceone_n}}} =1- u_n$. We define 
the sequence $(\traceone_n)_{n \in \NN}$ inductively as follows:
$$\traceone_0 = \emptytr \qquad \traceone_{n+1} = \concat{\actcut
  1}{\concat{\actappl {(\emcon, \identity)} 1}{\traceone_{n}}}$$
$\traceone_0$ is the trace which always succeeds, whatever the
starting state is. $\traceone_{n+1}$ corresponds to separating the two
components of the pair which is in first position in the tuple,
then passing the identity as an argument to the first component of
this pair, and then executing $\traceone_{n}$. For this
sequence of traces, the recursive equations of Figure \ref{fig:receq}
are verified (the proof can be found in \cite{EV}).
%la il faudrait faire une figure 
\begin{figure}[!h]
\begin{center}
\fbox{\begin{minipage}{\figurewidth}{\footnotesize
$$\probtr{\tuplonea{\termone_0}}{\traceone_0} = 1 \qquad \probtr{\tuplonea{\termtwo_0}}{\traceone_0} = 1 $$
$$\probtr{\tuplonea{\termone_{n+1}}}{\traceone_{n+1}}= 1\cdot \probtr{\tuplonea{\termone_{n}}}{\traceone_n} $$ 
$$ \probtr{\tuplonea{\termtwo_{n+1}}}{\traceone_{n+1}} = (1 - \frac 1 {2^{n+1}}  ) \cdot \probtr{\tuplonea{\termtwo_n}}{\traceone_n}$$} 
\end{minipage}}
\end{center}
\caption{Recursive equations verified by $\traceone_n$}\label{fig:receq}
\end{figure}
We can see by solving these equations that for every $n \in \NN$, $\probtr{\termone_{n}}{\traceone_n} = 1$ and $\probtr{\termtwo_{n}}{\traceone_n} = u_n$. As a direct consequence, we obtain the result. 
We want now to show that $\appl \metrtrtupl {\termone_n}{\termtwo_n} \leq 1 - u_n$. To do that, we need to establish that there doesn't exist a trace $\tracetwo$ such that  $\abs{{\probtr {\tuplonea{\termone_n}}{\tracetwo}} - {\probtr{\tuplonea{\termtwo_n}}{\tracetwo}}} > 1- u_n$. We're in fact going to show something stronger: for every $n \in \NN$, we're going to define a set $A_n$ of pairs of tuple, which contains the pair $({\tuplonea \termone_n}, \tuplonea {\termtwo_n} )$, and such that for every $(\tuplone, \tupltwo) \in A_n$, for every trace $\tracetwo$, $\abs{{\probtr \tuplone \tracetwo}-{\probtr \tupltwo \tracetwo}} \leq 1 - u_n$.
Intuitively, the idea behind the sequence $\{A_n\}_{n \in \NN}$ is the following: if we start from $\tuplonea{\termone_n}$, do a trace of even length, and end up in a tuple $\tuplone$ with a non-zero probability, and if when we do \emph{the same trace} starting from $\tuplonea{\termtwo_n}$ ending up in the tuple $\tupltwo$, then the pair of tuple $(\tuplone, \tupltwo)$ is in one of the $A_j$, with $j$ smaller than $n$. 
%Intuitively, the set $A_n$ corresponds to the the pairs $(\tuplone, \tupltwo)$ such that there exist a trace $\tracetwo$ for which, starting from $\tuplonea{\termone_n}$ and after having executed the trace $\tracetwo$, we end up in the tuple $\tuplone$ with a non-zero probability, and moreover, starting form $\tuplonea{\termtwo_n}$, and after having executed the \emph{same} trace $\tracetwo$ we end up in the tuple $\tupltwo$ with a non-zero probability.    
\condskip
\begin{definition}
Let be $n \in \NN$. Let $A_n$ be the set of $(\tuplone, \tupltwo) $ such that: there exist $m \in \NN$, and $k_i \geq n +1 $ (for $1 \leq i\leq m$), where:
\begin{align*}
\tuplone &= \tuplonea{\termone_n,[\abstr \varone \diver]^m};\\
\tupltwo &= \tuplonea{\termtwo_n, [\abstr \varone \psumindex \diver {\frac 1 {2^{k_i}}} \identity]_{1 \leq i \leq m}}.
\end{align*}
\end{definition}
\condskip
We want now to give an upper bound to the separation between $\tuplone$ and $\tupltwo$ any trace can induce, if $(\tuplone, \tupltwo) \in A_n$.
\condskip
\begin{lemma}\label{exauxn}
For every $n \in \NN$, for every $(\tuplone, \tupltwo) \in A_n$, we can partition the set of traces as:
\begin{align*}
\words =& \{\traceone \mid \probtr \tuplone \traceone = 0 \text{ and } \probtr \tupltwo \traceone \leq \frac 1 {2}\}\\
& \bigcup \{\traceone \mid \probtr \tuplone \traceone = 1 \text{ and } \probtr \tupltwo \traceone  \geq u_n \}.
\end{align*}
\end{lemma}
\begin{proof}
Let $\traceone \in \words$. We are going to show by induction on the length of $\traceone$ that for every $n \in \NN $, for every $(\tuplone, \tupltwo) \in A_n$, either $\probtr \tuplone \traceone = 0 \text{ and } \probtr \tupltwo \traceone \leq \frac 1 {2}$, or $ \probtr \tuplone \traceone = 1 \text{ and } \probtr \tupltwo \traceone  \geq u_n $.
\begin{varitemize}
\item If $\traceone = \emptytr$, then for every $n \in \NN$ and $(\tuplone, \tupltwo ) \in A_n$ $\probtr \tuplone \traceone = \probtr \tupltwo \traceone = 1$, and we are in the second case.
\item If the length of $\traceone$ is $l >0$. 
Let be $n \in \NN$, and $(\tuplone, \tupltwo) \in A_n$. Then  we can write: \begin{align*}
\tuplone &= \tuplonea{\termone_n,[\abstr \varone \diver]^m};\\
\tupltwo &= \tuplonea{\termtwo_n, [\abstr \varone \psumindex \diver {\frac 1 {2^{k_i}}} \identity]_{1 \leq i \leq m}} \quad \text{ with } k_i \geq n+1.
\end{align*}
We are now going to distinguish the cases depending on which element of the tuple is applied the first action of the trace. 
\begin{varitemize}
\item If the first action is not applied to the first element of the tuple, then  $\traceone = \concat{\actappl{ (\contone, \ctxone)} j}{\tracetwo}$, with $j >1$: Then $\probtr \tuplone \traceone = 0$, and $\probtr \tupltwo \traceone \leq   {\frac 1 {2^{k_j}}} \leq \frac 1 2$: we are in the first case.
\item If the first action is applied to the first element of the tuple : Then we can see that $\traceone = \concat{ \actcut 1}{\tracetwo}$ (since the first element of the tuple is actually a pair, the only action that can be applied to it is the unfold action). 
\begin{varitemize}
\item First, let's consider the case where $n = 0$. Please remember that by definition we have that $\termone_0 = \termtwo_0 = \pair{\abstr \varone \diver}{\abstr \varone \diver}$. Observe that:
\begin{align*}
\tuplone_1 &= \tuplonea{[\abstr \varone \diver]^{m+2}};\\
\tupltwo_1 &=   \tuplonea{{\abstr \varone \diver, \abstr \varone \diver, [\abstr \varone \psumindex \diver {\frac 1 {2^{k_i}}} \identity]_{1 \leq i \leq m}}}.
\end{align*}
With these notations, we can see that $\probtr \tuplone \traceone = \probtr {\tuplone_1}\tracetwo $, and 
$\probtr \tupltwo \traceone =   \probtr{\tupltwo_1} \tracetwo$. If $\tracetwo = \emptytr$, these two expressions are equal to $1$, and we are in the second case. Otherwise, $\probtr {\tuplone_1} \tracetwo = 0$, and $\probtr {\tupltwo_1} \tracetwo \leq \frac 1 2$, and we are in the first case. 

\item Now let's consider the case where ${n \geq 1}$. Please remember that:
\begin{align*} 
\termone_n &= \pair{\abstr \varone \termone_{n-1}}{ \abstr \varone \diver};\\
\termtwo_n &= \pair {\abstr \varone ({\psumindex {\termtwo_{n-1}} {\frac 1 {2^n}} \diver})}{\abstr \varone (\psumindex \diver {\frac 1 {2^{n}}} \identity)}.
\end{align*}
 
Then we'll note:
\begin{align*}
\tuplone_2 &= {\tuplonea{\abstr \varone \termone_{n-1}, \abstr \varone \diver, [\abstr \varone \diver]^m}};\\
\tupltwo_2 &=   {\tuplonea{{\abstr \varone ({\psumindex {\termtwo_{n-1}} {\frac 1 {2^n}} \diver})},{\abstr \varone (\psumindex \diver {\frac 1 {2^{n}}} \identity)} ,[\abstr \varone (\psumindex \diver {\frac 1 {2^{k_i}}} \identity)]_{1 \leq i \leq m}}}.
\end{align*}
With these notations, we can see that $\probtr \tuplone \traceone = \probtr {\tuplone_2}\tracetwo $, and 
$\probtr \tupltwo \traceone =   \probtr{\tupltwo_2} \tracetwo$.
Now we have to consider the different possible form of the trace $\tracetwo$ :
\begin{varitemize}
\item if $\tracetwo = \emptytr$, $\probtr \tuplone \traceone= \probtr \tupltwo \tracetwo = 1$.
\item if $\tracetwo = \concat{\actappl{ (\contone, \ctxone)} j}{\tracethree}$, with $j >1$, we have $\probtr {\tuplone_1} \tracethree = 0$ and $\probtr {\tuplone_2} \tracethree \leq \frac 1 {2^{n}} \leq \frac 1 2$ , and we are in the first case.
\item if $\tracetwo = \concat{\actappl{ (\contone, \ctxone)} 1}{\tracethree}$. Please remember the semantics of this action in the Markov Chain : If we start from $\tuplone_2$, with probability 1 we go to a state
$\tuplone_3$ of the form :
$$\tuplone_3 = \tuplonea{\termone_{n-1}, {[\abstr \varone \diver]}^l}$$ with $l \leq m$. If we start from $\tupltwo_2$, with probability $(1 - \frac 1 {2^n})$ we go in a state $\tupltwo_3$ of the form :
$$\tupltwo_3 = {\tuplonea{\termtwo_{n-1},[\abstr \varone (\psumindex \diver {\frac 1 {2^{k_i}}} \identity)]_{1 \leq i \leq l}}}$$  with  $k_i \geq n$.
Now we can see that : 
\begin{align*}
\probtr \tuplone \traceone &= \probtr {\tuplone_3}\tracethree; \\
\probtr \tupltwo \traceone &= (1 - \frac 1 {2^n}) \cdot   \probtr{\tupltwo_3} \tracethree.
\end{align*}
Moreover, please observe that $(\tuplone_3, \tupltwo_3) \in A_{n-1}$, so we can apply the induction hypothesis (since the length of $\tracethree$ is strictly smaller that the length of $\traceone$). Now, there are two possible cases :
\begin{varitemize}
\item $\probtr {\tuplone_3} \tracethree = 0$, and $\probtr {\tupltwo_3}{\tracethree} \leq \frac 1 2$. Then we can see that the result holds, since it implies that : $\probtr \tuplone \traceone = 0$ and $\probtr \tupltwo \traceone \leq  (1 - \frac 1 {2^n}) \cdot \frac 1 2 \leq \frac 1 2$.
\item$\probtr {\tuplone_3} \tracethree = 1$, and $\probtr {\tupltwo_3}{\tracethree} \geq u_{n-1}$.  Then we can see that the result holds, since it implies that : $\probtr \tuplone \traceone = 1$ and $\probtr \tupltwo \traceone \geq  (1 - \frac 1 {2^n}) \cdot u_{n-1} = u_n$.
\end{varitemize}
\end{varitemize}
\end{varitemize}
\end{varitemize}
\end{varitemize}
\end{proof}
\condskip
 The result we're seeking to show is a direct consequence of Lemma \ref{exauxn}: we can see easily that for any trace $\traceone$,  if $(\tuplone, \tupltwo) \in A_n$, the separation that $\traceone$ can induce is smaller than $1-u_n$.
Indeed, let be $\traceone \in \words$. Since $(\tuplonea {\termone_n}, \tuplonea {\termtwo_n}) \in A_n$, we can see that :
\begin{itemize}
\item Or the trace $\traceone$ is in the first set of the partition given by Lemma \ref{exauxn}, and $\abs{\probtr {\termone_n} \traceone - \probtr {\termtwo_n} \tracetwo}\leq \frac 1 2 \leq 1 - u_n$.
\item Or the trace $\traceone$ is in the second set of this partition, and then $\abs{\probtr {\termone_n} \traceone - \probtr {\termtwo_n} \tracetwo}  \leq 1 - u_n$.
\end{itemize}
\end{proof}
\condskip
%%%%%%%%%%%%%%%%%%%%%%%%%%%%%%%%%%%%%
\subsection{On Tuples and Copying}\label{sect:exponentials}
%%%%%%%%%%%%%%%%%%%%%%%%%%%%%%%%%%%%%
The tuple distance naturally suggests a way to handle
$\lambda$-calculi in which copying is indeed allowed. Although the
details are clearly outside the scope of this paper, we
anyway want to give some hints about why this is the case.

What makes the trace and behavioural distances unsound in presence of
copying is their inability to capture an environment which can access
the program at hand \emph{more than once}.  In our view, however, the
problem does not come from the way those distances are defined in the
abstract, but rather in the way \emph{the underlying LMC} reflects the
operational semantics of the calculus at hand. In a sense, it is in
the responsibility of the LMC to guarantee that the environment can 
access terms multiple times. The LMC $\markovterm$ we introduced in this
paper (which is close to the ones from the
literature~\cite{CrubilleDalLago2014ESOP,DalLagoSangiorgiAlberti2014POPL,DengZhang}),
as an example, is not adequate.

Suppose, however, to extend $\markovtupl$ to an LMC for a
$\lambda$-calculus in the style of Wadler's linear
$\lambda$-calculus~\cite{Wadler}: there, the grammar of terms includes
a construct $!\termone$ whose purpose is marking those subterms which
can indeed be duplicated. The actions the environment can perform on a
term in the form $!\termone$ simply reflects the above: the
environment can create a \emph{new copy} of $!\termone$, \emph{but
  also keeps} the possibility to access $!\termone$ in the future. One
immediately realises that tuples are indeed the right way to model
the access to both $!\termone$ and $\termone$.
%%%%%%%%%%%%%%%%%%%%%
\section{Conclusions}
%%%%%%%%%%%%%%%%%%%%%
We have initiated the study of metrics in higher-order languages, 
starting with the relatively easy
case of affine $\lambda$-terms, where copying capabilities are
simply not available. We showed that three different notions
of distance are sound (and sometime fully-abstract) for the
context distance, the natural generalisation of Morris'
observational equivalence. One of them, the tuple distance,
reflects the inherently monoidal structure of the underlying
calculus, this way allowing to solve some nontrivial distance
problems. 

We are actively working on extending the results described here to the
non-affine case, which for various reasons turns out to be more
difficult, as discussed in Section~\ref{sect:anatomy}. We are in
particular quite optimistic about the possibility of generalising the
tuple distance to a metric reflecting copying. The real challenge,
however, consists in handling the case in which copying is indeed
available, but the number of copies of a given term the environment
can have access to is somehow bounded, maybe polynomially on the value
on an security parameter. That would indeed be a way to get closer to
computational indistinguishability, a central notion in modern
cryptography.

%\begin{thebibliography}{1}
\bibliographystyle{abbrv}
\bibliography{biblio}

%\bibitem{IEEEhowto:kopka}
%H.~Kopka and P.~W. Daly, \emph{A Guide to \LaTeX}, 3rd~ed.\hskip 1em plus
%  0.5em minus 0.4em\relax Harlow, England: Addison-Wesley, 1999.

%\end{thebibliography}

% that's all folks
\end{document}

%% file: macros.tex
\usepackage{tikz}
\usepackage{mathrsfs}
\usepackage{amssymb}
\usepackage{bussproofs}
\usepackage{stmaryrd}
\usepackage{a4wide}

\newcommand{\NN}{\mathbb{N}}

\newtheorem{theorem}{Theorem}

\newtheorem{lemma}{Lemma}
\newtheorem{definition}{Definition}

\newtheorem{example}{Example}

\newenvironment{proof}{\begin{trivlist}
       \item[\hskip \labelsep {\bfseries Proof.}]}{\hfill $\Box$ \end{trivlist}}

\newcommand{\condskip}{}

\newenvironment{varitemize}
{
\begin{list}{\labelitemi}
{
\setlength{\itemsep}{0pt}
 \setlength{\topsep}{0pt}
 \setlength{\parsep}{0pt}
 \setlength{\partopsep}{0pt}
 \setlength{\leftmargin}{15pt}
 \setlength{\rightmargin}{0pt}
 \setlength{\itemindent}{0pt}
 \setlength{\labelsep}{5pt}
 \setlength{\labelwidth}{10pt}}}
{
 \end{list}
}

%% file: macros_notations.tex
\usepackage{color}
\usepackage{stmaryrd}
\newcommand{\hide}[1]{}

\newcommand{\kronecker}[2]{Kr(#1,#2)}
\newcommand{\defi}{\,:=\,}
\newcommand{\terms}{\text{Terms}}
\newcommand{\variables}{\textbf{V}}
\newcommand{\varone}{x}
\newcommand{\vartwo}{y}
\newcommand{\varthree}{z}
\newcommand{\termone}{M}
\newcommand{\termtwo}{N}
\newcommand{\termthree}{L}
\newcommand{\termfour}{K}
\newcommand{\termfive}{T}
\newcommand{\termsix}{U}
\newcommand{\termseven}{P}
\newcommand{\termeight}{Q}
\newcommand{\valone}{V}
\newcommand{\valtwo}{W}
\newcommand{\valset}{\mathcal{V}}
\newcommand{\subst}[3]{#1\{#3/#2\}}
\newcommand{\abstr}[2]{\lambda #1.#2}
\newcommand{\psum}[2]{#1\oplus #2}
\newcommand{\psumindex}[3]{#1 \oplus^{#2} #3}
\newcommand{\pair}[2]{\langle #1, #2\rangle}
\newcommand{\letin}[4]{\texttt{let}\ \pair{#1}{#2}={#3}\ \texttt{in}\ {#4}}
\newcommand{\diver}{\Omega}
\newcommand{\identity}{I}
\newcommand{\emcon}{\emptyset}
\newcommand{\contone}{\Gamma}
\newcommand{\conttwo}{\Delta}
\newcommand{\wfj}[2]{#1\vdash #2}
\newcommand{\wfjval}[2]{#1\vdash #2 : \textbf{ val }}
\newcommand{\wfjt}[3]{#1 \vdash #2 : #3}
\newcommand{\bss}{\Downarrow}
\newcommand{\ssp}[2]{#1 \rightarrow #2}

\newcommand{\sssp}[2]{#1 \Rightarrow #2}
\newcommand{\bssp}[2]{#1\Downarrow #2}
\newcommand{\sem}[1]{\llbracket#1\rrbracket}
\newcommand{\programs}{\textbf{P}}
\newcommand{\ttrue}{\textbf{true}}
\newcommand{\ffalse}{\textbf{false}}
\newcommand{\testone}{t}
\newcommand{\testtwo}{s}
\newcommand{\testthree}{r}
\newcommand{\wtr}[3]{#1 \stackrel{#2}{\rightarrow} #3}
\newcommand{\wtrc}[3]{{#1} \stackrel{#2}{\Rightarrow} #3}
\newcommand{\wtrct}[3]{#1 \stackrel{#2}{\rightarrow}_{\ensct} #3}
\newcommand{\wtrcct}[3]{{#1} \stackrel{#2}{\Rightarrow}_{\ensct} #3}

\newcommand{\emptytr}{\epsilon}
\newcommand{\dval}[1]{\widehat{#1}}
\newcommand{\traceone}{s}
\newcommand{\tracetwo}{t}
\newcommand{\tracethree}{u}
\newcommand{\concat}[2]{{#1 \cdot #2}}
\newcommand{\app}[1]{@#1}
\newcommand{\eval}{\text{eval}}

\newcommand{\sizec}[1]{{\mid #1 \mid}}
\newcommand{\proj}[1]{\pi_{#1}}
\newcommand{\tenseur}[1]{\otimes {#1}}
\newcommand{\contexts}{\mathbf C}
\newcommand{\hole}{[\cdot]}
\newcommand{\ctxone}{\mathcal C}
\newcommand{\ctxtwo}{\mathcal D}
\newcommand{\ctxthree}{\mathcal E}
\newcommand{\ctxfour}{\mathcal F}
\newcommand{\fillc}[2]{#1[#2]}
\newcommand{\freevar}[1]{FV(#1)}
\newcommand{\bnf}{::=}
\newcommand{\midd}{\; \; \mbox{\Large{$\mid$}}\;\;} 
\newcommand{\distrs}[1]{\text{Distr}\left( #1\right)}
\newcommand{\disjplus}{\dotplus}
\newcommand{\dirac}[1]{{\{{#1}^1\}} }
\newcommand{\probone}{p}
\newcommand{\emdist}{\emptyset}
\newcommand{\distrone}{\mathscr{D}}
\newcommand{\distrtwo}{\mathscr{E}}
\newcommand{\distrthree}{\mathscr{F}}
\newcommand{\distrfour}{\mathscr{G}}
\newcommand{\distrfive}{\mathscr{H}}
\newcommand{\distrsix}{\mathscr{I}}
\newcommand{\distrseven}{\mathscr{J}}
\newcommand{\distreight}{\mathscr{K}}
\newcommand{\distrnine}{\mathscr{L}}
\newcommand{\supp}[1]{\mathsf{S}(#1)}
\newcommand{\sumdistr}[1]{\sum{#1}}
\newcommand{\sumsem}[1]{\mathcal{P}^{\text{cv}}(#1)}

\newcommand{\abs}[1]{\lvert #1\rvert}

\newcommand{\forget}[1]{\mathbf{F}({#1})}
\newcommand{\probtr}[2]{Pr(#1,#2)}
\newcommand{\probtrbstupl}[2]{Pr^{\text{bs}}_{\text{mul}}(#1,#2)}
\newcommand{\probtrtupl}[2]{Pr^{\text{mul}}(#1,#2)}
\newcommand{\probtrsstupl}[2]{Pr^{\text{ss}}_{\text{mul}}(#1,#2)}

\newcommand{\mleqpar}[3]{#1 \equiv^{\text{par}}_{#3} #2}
\newcommand{\applprob}[4]{#1 (#2,#3)(#4)}
\newcommand{\normal}[1]{{#1}^\star}
\newcommand{\emptydistr}{\emptyset}
\newcommand{\relone}{R}
\newcommand{\relate}[3]{#2\,#1\,#3}
\newcommand{\leqtr}[2]{{#1} \leq^{\text{tr}} {#2}}
\newcommand{\equivtr}[2]{{#1} \equiv^{\text{tr}} {#2} }
\newcommand{\leqctx}[2]{{#1} \leq^{\text{ctx}} {#2}}
\newcommand{\equivctx}[2]{{#1} \equiv^{\text{ctx}} {#2} }
\newcommand{\metrs}[1]{ \Delta(#1)}

\newcommand{\bisim}{\delta^{\text{b}}}
\newcommand{\metrctx}{\delta^{\text{ctx}}}
\newcommand{\metrtr}{\delta^{\text{tr}}}
\newcommand{\metrtrtupl}{\delta^{\text{mul}}}

\newcommand{\leqmetr}[2]{#1 \leq^{\text{metr}}#2}
\newcommand{\functionnal}{F}
\newcommand{\appl}[3]{{#1}(#2,#3)}
\newcommand{\metrone}{\mu}
\newcommand{\metrtwo}{\rho}
\newcommand{\metrthree}{\nu}
\newcommand{\howe}[1]{{#1}^{H}}
\newcommand{\epsone}{\varepsilon}
\newcommand{\epstwo}{\iota}
\newcommand{\epsthree}{\gamma}
\newcommand{\epsfour}{\beta}
\newcommand{\actions}{\mathcal A}

\newcommand{\words}{\mathcal{T}r}
\newcommand{\ensct}{\contexts\times\programs}
\newcommand{\ltstrace}{\mathcal L^{\text{tr}}}
\newcommand{\ltstracect}{\mathcal L^{\text{tr}}_{\ensct}}
\newcommand{\setone}{S}
\newcommand{\markovone}{\mathscr{M}}
\newcommand{\labelsone}{\mathscr{L}}
\newcommand{\probmatrone}{\mathscr{P}}
\newcommand{\markovterm}{{\mathscr{M}^{\Lambda}}}
\newcommand{\labelsterm}{\mathscr{L}^{\Lambda}}
\newcommand{\probmatrterm}{\mathscr{P}^{\Lambda}}
\newcommand{\setterm}{{S}^{\Lambda}}
\newcommand{\markovtupl}{\mathscr{M}^{\Lambda}_{\text{mul}}}
\newcommand{\labelstupl}{\mathscr{A}^{\Lambda}_{\text{mul}}}
\newcommand{\probmatrtupl}{\mathscr{P}^{\Lambda}_{\text{mul}}}
\newcommand{\settupl}{{S}^{\Lambda}_{\text{mul}}}

\newcommand{\statebs}{\settupl }
\newcommand{\statess}{\setonewf }
\newcommand{\stateone}{s}
\newcommand{\statetwo}{t}
\newcommand{\statethree}{u}
\newcommand{\statefour}{v}
\newcommand{\actone}{a}
\newcommand{\actbs}{\labelstupl }
\newcommand{\actss}{Act^{ss}}

\newcommand{\doact}[3]{{#1}\stackrel{#2}{\rightsquigarrow}{#3}}
\newcommand{\obs}[1]{\text{Obs}(#1)}

\newcommand{\coeffone}{a}
\newcommand{\coefftwo}{b}
\newcommand{\coeffdone}{l}
\newcommand{\coeffdtwo}{h}
\newcommand{\coeffdthree}{n}
\newcommand{\coeffdxone}{x}
\newcommand{\coeffdyone}{y}
\newcommand{\coeffdxtwo}{w}
\newcommand{\coeffdytwo}{z}
\newcommand{\coeffdxthree}{a}
\newcommand{\coeffdythree}{b}
\newcommand{\judghowe}[5]{{#1}\vdash{\appl{\howe{#2}}{#3}{#4}} \leq {#5}}
\newcommand{\judgclh}[3]{\vdash {\appl {\closure {\howe \bisim}} {#1}{#2}} \leq {#3}}
\newcommand{\closure}[1]{#1_\triangle}
\newcommand{\embone}{\Theta}
\newcommand{\calcwpair}{\Lambda^{\pair{}{}}_{\oplus}}
\newcommand{\torvstates}{\mathcal {TV} }
\newcommand{\torvstateswf}{\mathcal {TV}^{\text{focus}}}
\newcommand{\setonewf}{\mathcal{S}^{\text{focus}}}
\newcommand{\torvone}{s}
\newcommand{\torvtwo}{t}
\newcommand{\torvthree}{u}

\newcommand{\tuple}[3]{\left[#1,#2,#3\right]}
\newcommand{\tuplone}{K}
\newcommand{\tupltwo}{H}
\newcommand{\tuplonea   }[1]{\left[#1\right]}
\newcommand{\actcut}[1]{\text{unfold}^{#1}}
\newcommand{\actappl}[2]{@{#1}^{#2}}
\newcommand{\actpair}[1]{\pair{}{}\cdot{#1}}
\newcommand{\acteval}[1]{\text{eval} ^{#1}}
\newcommand{\typone}{\sigma}
\newcommand{\typtwo}{\tau}
\newcommand{\typthree}{\gamma}
\newcommand{\tprod}[2]{#1 \otimes #2}

\newcommand{\ctwellformed}{\mathcal A}
\newcommand{\ctwellformedt}[1]{\mathcal A^{#1}}
\newcommand{\functionone}{\phi}
\newcommand{\functiontwo}{\psi}
\newcommand{\functionthree}{\nu}
\newcommand{\functionfour}{\eta}
\newcommand{\restr}[2]{{(#1)}\restriction_{#2}}
\newcommand{\domain}[1]{\text{Dom}(#1)}
\newcommand{\compose}[2]{#1\circ #2}
\newcommand{\parall}[3]{(#2\,|^{#1}\,#3)}
\newcommand{\image}[1]{Im({#1})}
\newcommand{\suitone}[1]{u_{#1}}
\newcommand{\limit}{l}
\newcommand{\suittwo}[1]{v_{#1}}
\newcommand{\suitthree}[1]{w_{#1}}
\newcommand{\focus}[2]{\text{focus}^{#1}(#2) }
\newcommand{\internact}{\tau}
\newcommand{\focusindex}[1]{f(#1)}
\newcommand{\reduced}[1]{(#1^{\star})}
\newcommand{\permut}{\theta}

\newcommand{\fdistrs}[1]{\Delta(#1)}
\newcommand{\wtrtupl}[3]{#1 \stackrel{#2}{\rightarrow_{\tuplonea{}}} #3}
\newcommand{\wtrdistrtupl}[3]{#1 \stackrel{#2}{\rightarrow_{\Delta(\tuplonea{})}} #3}
\newcommand{\wtrcdistrtupl}[3]{#1 \stackrel{#2}{\Rightarrow_{\Delta(\tuplonea{})}} #3}
\newcommand{\wtrbs}[3]{#1 \stackrel{#2}{\rightarrow}_{\fdistrs \settupl} #3}